\documentclass[a4paper,11pt]{article}
\usepackage{jheppub}
\usepackage[utf8]{inputenc}
\usepackage{url}
\usepackage{mathrsfs}  
\usepackage{pdflscape}

\usepackage{refcount}
\usepackage{booktabs}
\usepackage{todonotes}
\usepackage{enumitem}
\usepackage{adjustbox}
\usepackage{afterpage}
\usepackage{dsfont}
\usepackage{nameref}

\makeatletter
\newcommand\level[1]{%
  \ifcase#1\relax\expandafter\chapter\or
    \expandafter\section\or
    \expandafter\subsection\or
    \expandafter\subsubsection\else
    \def\next{\@level{#1}}\expandafter\next
  \fi}
\newcommand{\@level}[1]{%
  \@startsection{level#1}
    {#1}
    {\z@}%
    {-3.25ex\@plus -1ex \@minus -.2ex}%
    {1.5ex \@plus .2ex}%
    {\normalfont\normalsize\bfseries}}

\newcounter{level4}[subsubsection]
\@namedef{thelevel4}{\thesubsubsection.\arabic{level4}}
\@namedef{level4mark}#1{}
\count@=4
\loop\ifnum\count@<100
  \begingroup\edef\x{\endgroup
    \noexpand\newcounter{level\number\numexpr\count@+1\relax}[level\number\count@]
    \noexpand\@namedef{thelevel\number\numexpr\count@+1\relax}{%
      \noexpand\@nameuse{thelevel\number\count@}.\noexpand\arabic{level\number\numexpr\count@+1\relax}}
    \noexpand\@namedef{level\number\numexpr\count@+1\relax mark}####1{}}
  \x
  \advance\count@\@ne
\repeat
\makeatother
\usepackage[utf8]{inputenc} 
\usepackage{amsmath}
\usepackage{amsthm}
\usepackage{amsfonts,mathrsfs} 
\usepackage{xcolor}
\usepackage[english]{babel} 
\usepackage[autostyle]{csquotes}

\newtheorem*{theorem}{Theorem}

\newlist{inparaenum}{enumerate}{3}
\setlist[inparaenum,1]{label=\arabic*.}
\setlist[inparaenum,2]{label=\emph{\alph*})}
\setlist[inparaenum,3]{label=\emph{\roman*})}

\numberwithin{equation}{section}

\makeatletter
\newcommand\footnoteref[1]{\protected@xdef\@thefnmark{\ref{#1}}\@footnotemark}
\makeatother

\usepackage{xstring}
\usepackage{tikz} 
\usepackage{booktabs}
\usepackage{xcolor}
\usetikzlibrary{decorations.pathreplacing} 
\usetikzlibrary{decorations.pathmorphing} 
\usetikzlibrary{decorations.markings} 
\usetikzlibrary{arrows} 
\usetikzlibrary{shapes} 
\usetikzlibrary{matrix} 
\usetikzlibrary{positioning} 
\usetikzlibrary{shapes.geometric}
\usetikzlibrary{calc}

\tikzstyle{brane}=[draw]
\tikzset{D7/.style={circle, draw=black, inner sep=0pt, fill=white, minimum size=3mm}}
\tikzset{hasse/.style={circle, fill,inner sep=2pt}}
\tikzset{flavor/.style={regular polygon,regular polygon sides=4,inner sep=2.5pt, draw}}
\tikzset{gauge/.style={circle, draw,inner sep=2.5pt}}
\tikzset{gauger/.style={circle, draw=red, fill=red, inner sep=2.5pt}}
\tikzset{gaugeb/.style={circle, draw=lightblue, fill=lightblue, inner sep=2.5pt}}
\tikzset{emptygauge/.style={circle, draw=gray, dashed, inner sep=2.5pt}}
\tikzset{bd/.style={circle, draw=black, inner sep=0pt, fill=black, minimum size=2mm}}
\tikzset{wd/.style={circle, draw=black, inner sep=0pt, fill=white, minimum size=2mm}}
\tikzset{Dynkin/.style={circle, draw=black, inner sep=0pt, fill=white, minimum size=2mm}}
\tikzstyle{ligne}=[draw, thick] 
\tikzset{doublearrow/.style={ draw=black!75, color=black!75, thick, double distance=3pt, }}

\colorlet{lightblue}{blue!65}

\preprint{Imperial/TP/21/AH/06}

\title{Coulomb Branch Global Symmetry and Quiver Addition}

\author{Kirsty Gledhill}
\author{and Amihay Hanany }
\affiliation{Theoretical Physics Group, The Blackett Laboratory, Imperial College London, Prince Consort Road
London, SW7 2AZ, UK}
\emailAdd{k.gledhill20@imperial.ac.uk}
\emailAdd{a.hanany@imperial.ac.uk}

\abstract{To date, the best effort made to simply determine the Coulomb branch global symmetry of a theory from a $3d$ $\mathcal{N}=4$ quiver is by applying an algorithm based on its balanced gauge nodes. This often gives the full global symmetry, but there have been many cases seen where it instead gives only a subgroup. This paper presents a method for constructing several families of $3d$ $\mathcal{N}=4$ unitary quivers where the true global symmetry is enhanced from that predicted by the balance algorithm, motivated by the study of Coulomb branch Hasse diagrams. This provides a rich list of examples on which to test improved algorithms for unfailingly identifying the Coulomb branch global symmetry from a quiver.}

\begin{document}
\maketitle

\section{Introduction}
\label{sec:Intro}
The particles appearing in a quantum field theory are thought of as excitations of the vacuum configuration. The moduli space of a theory, the space of all gauge invariant vacua, is therefore of great interest to study. Since Noether's work in the early $1900s$ \cite{Noether:1918}, symmetries have played a role of utmost importance in physics: inducing conserved currents and charges which identify quantum numbers that are then used to label states. To properly discern the spectrum of a theory then, we need to know the quantum numbers that will be used to describe them, and thus must learn about the global symmetry of the moduli space.\\

\noindent In $1996$, Douglas and Moore \cite{Douglas:1996sw} showed that the information in a Lagrangian theory could be boiled down and encoded in a diagram called a quiver \cite{Nakajima:1994nid}. Recent history has seen much development in the understanding of quiver theories\footnote{Quiver and quiver theory are often used synonymously. For example, ``the Coulomb branch of a quiver" really means the Coulomb branch of the theory the quiver represents.}, and in this paper it is the global symmetries of Coulomb branches of $3d$ $\mathcal{N}=4$ quiver gauge theories that we focus on. In theories with eight supercharges, the Coulomb branch is the subset of the moduli space which is parameterised by the vacuum expectation values of the scalars in the vector multiplets of the theory. It is an algebraic variety, as it is the set of variables which cause the vacuum equations to yield zero subject to the relation of gauge invariance. It can also be understood as a geometric object, giving information about the effective K\"ahler potential. In $2016$, Nakajima and collaborators conjectured in \cite{Braverman:2016wma} that the Coulomb branches of $3d$ $\mathcal{N}=4$ quiver gauge theories are symplectic singularities\footnote{Roughly speaking, according to Beauville's definition, a symplectic singularity is an algebraic variety with a symplectic form which contains zero or more singular points at which the symplectic form is degenerate.} \cite{2000InMat.139..541B}. Many techniques have been developed to study these algebro-geometric objects, but the one of greatest importance in this paper is the monopole formula \cite{Cremonesi:2013lqa}.\\

\noindent In three dimensions, at a generic point on the Coulomb branch, the gauge symmetry is broken to $U(1)^r$ where $r$ is the rank of the gauge group of the theory \cite{PhysRevLett.13.585,Seiberg:1994rs}. Moreover, for every $U(1)$ factor the gauge group contains, there is an associated conserved current, which arises due to a global topological symmetry\footnote{The name topological is due to the fact that the existence of this symmetry is just a consequence of the geometric construction of a gauge field in spacetime.}. For the case of $\mathcal{N}=4$, there is also an $SU(2)_C$ global $R$-symmetry.\footnote{The Higgs branch also experiences an $SU(2)_H$ $R$-symmetry, and so the full moduli space has $SU(2)_C \ \times \ SU(2)_H$ $R$-symmetry.} The monopole formula gives the refined Hilbert series of the Coulomb branch by introducing a fugacity for each of these symmetries, $z$ and $t$ respectively, to grade the gauge invariant operators of the quiver theory by their charges under these symmetries. Note that this formula only works for \textit{good} and \textit{ugly} theories, as defined in \cite{Gaiotto:2008ak}, and in such cases the Coulomb branch is found to be a HyperK\"ahler cone. The Hilbert series is just another name for a partition function or generating function: in this case, the coefficients of $t$ at power $2 \Delta$ give the characters of the representations under the global symmetry $G$ of the gauge invariant operators of the theory with conformal dimension $\Delta$.\footnote{Note that the convention we adopt is to grade the fugacity $t$ by twice the conformal dimension. This is in contrast to some other works, in which $t$ is graded by just one copy of $\Delta$.} The point that is of particular interest to us in this work is that the $t^2$ coefficient is always the character of the adjoint representation of the global symmetry \cite{1992math......4227B}, so to determine the Coulomb branch global symmetry of a quiver one can simply compute the monopole formula to order $t^2$.\\

\noindent Whilst this is a very useful and reliable method, it takes time to implement, and in fact for more complicated quivers it can actually be too computationally challenging to find. We would like to realise the Coulomb branch global symmetry simply from looking at the quiver for the theory in question. By and large, this can be done by applying an algorithm \cite{Cabrera:2018uvz} based on the quiver's balanced \textit{gauge} nodes \cite{Gaiotto:2008ak}, but there are cases for which this algorithm only produces a strict subgroup of the Coulomb branch global symmetry \cite{Mekareeya:2017jgc,Bourget:2020mez}. Due to the failure of the balance algorithm to always give the full Coulomb branch global symmetry, there is a need for a new simple method to quickly determine it by just looking at a quiver.\\

\noindent In an effort to overcome this problem we turn to the Hasse diagram, which can also be used as a tool to learn about the global symmetry of a quiver. It characterises the structure of a moduli space by visualising its stratification as a symplectic singularity into symplectic leaves and transverse slices \cite{Bourget:2019aer,Grimminger:2020dmg}. From a physical perspective, the Hasse diagram is a kind of phase diagram, where by a phase of a theory we mean the collection of all possible vacua with the same set of massless states. Each symplectic leaf represents a different phase, and each transverse slice tells us which moduli are adjusted to move between the leaves (phases) that it connects.\\

\noindent The global symmetries of the lowest elementary slices in the Hasse diagram form a subgroup of the non-Abelian global symmetry, and so finding these lowest slices can aid us in our pursuit of determining the global symmetry of a quiver. Quiver subtraction \cite{Cabrera:2018ann} can be used to find the Coulomb branch Hasse diagram of a quiver. It turns out that, for Coulomb branch Hasse diagrams which contain as a leaf a particular type of quiver $Q_a$, the non-Abelian global symmetry of the Coulomb branch matches that predicted by the Hasse diagram, but is slightly enhanced from that predicted by the balance. This means that we can reverse engineer the process of quiver subtraction, starting from $Q_a$, to obtain quivers whose Coulomb branch global symmetry is enhanced from that predicted by the balance. The family of quivers $Q_a$ that are important in this process are a certain set of quivers containing adjoint matter, many of whom's Coulomb branches appear in the Kostant Brylinski classification \cite{1992math......4227B}. The reverse process of quiver subtraction, quiver addition, has been introduced before in \cite{Bourget:2021siw}. We ammend the method given there to include how to perform quiver addition on quivers $Q_a$ to remove the adjoint matter. This construction provides a rich range of examples on which to test new proposed algorithms to determine the Coulomb branch global symmetry directly from a quiver. \\

\noindent The paper is organised as follows. In Section \ref{sec: CB GS} we review the method of finding the Coulomb branch global symmetry from the balance, and some basics of Hasse diagrams. In Section \ref{sec: addition method}, the ammendment to the balance algorithm for the quivers $Q_a$ is stated, a full list of such quivers $Q_a$ (on which quiver addition can be performed to result in a quiver with enhanced Coulomb branch global symmetry) is provided, and the method of performing this addition is given. In Sections \ref{sec:enhancement to Bn}, \ref{sec:enhancement to G2}, \ref{sec:enhancement to Dn} and \ref{sec:enhancement to A2}, the quivers with enhanced Coulomb branch global symmetry which arise as the results of performing quiver addition on the quivers $Q_a$ given in Table \ref{tab:adj hyper Qs} of Section \ref{sec:quivers with adj matter} are listed. The global symmetries predicted by the balance algorithm in these sections contain a factor that is enhanced to $B_n$, $G_2$, $D_n$ and $A_2$\footnote{The algebraic and Dynkin names for Lie groups, for example $B_n$ and $SO(2n+1)$, will be used synonymously throughout this paper.} respectively. Appendix \ref{app:quiver subtraction} reviews the method of quiver subtraction on quivers both with and without adjoint matter. Appendix \ref{app:fugmap} details the computation of fugacity maps to make explicit the global symmetry using the refined Hilbert series. In Appendix \ref{app:discrete projection}, a discussion and illustration of performing a discrete quotient on the refined Hilbert series is given.\\

\noindent Before we begin, we note the following points which will apply throughout the paper:
\begin{itemize}
    \item The quivers considered below will be restricted to $3d$ $\mathcal{N}=4$ unitary quiver gauge theories for which the monopole formula can be used to compute the Coulomb branch Hilbert series; that is, unitary quivers whose Coulomb branch is a HyperK\"{a}hler cone. For example, the Coulomb branches of many unitary quivers containing only regular matter\footnote{\label{footnote:reg matter}Regular matter is used to refer to that contained in quivers for which any two connected nodes, $n_i$ and $n_j$, the Cartan matrix has either the $ij$ or $ji$ entry as $-1$. This is extended from the definition of just having bifundamental matter to include non-simply laced edges.} which satisfy $\text{gcd}(\text{gauge node ranks})>1$ are not HyperK\"{a}hler cones,\footnote{The reason for this is not known, and so it could be that there exist counterexamples. We exclude the study of those that we \textit{know} diverge.} and so these will not be considered.
    \item Since we exclusively study the global symmetry of Coulomb branches in this note, in the interest of brevity, we will omit the ``Coulomb branch" from Coulomb branch global symmetry, Coulomb branch Hilbert series or Coulomb branch Hasse diagram throughout. 
    \item When discussing algebraic varieties, we will refer to this variety by both its magnetic quiver and the name of the variety itself interchangeably, and vice versa.
\end{itemize}

\section{Coulomb Branch Global Symmetry}
\label{sec: CB GS}

\noindent As mentioned in the introduction, we would like to be able to determine the global symmetry of a quiver without detailed computation. By and large, this can be done for the quivers\footnote{Recall that we only study $3d$ $\mathcal{N}=4$ unitary quivers in this paper.} that contain only regular matter,\footnoteref{footnote:reg matter} $Q_r$, by applying an algorithm based on its balanced \textit{gauge} nodes \cite{Gaiotto:2008ak, Cabrera:2018uvz}. We will refer to the global symmetry determined from this method as the \textbf{balance global symmetry (BGS)}. Although in most cases the BGS \textit{is} the true global symmetry, there are also many instances in which the BGS is instead a strict subset of it, see for instance \cite{Mekareeya:2017jgc,Bourget:2020mez}. In these cases the true global symmetry is enhanced from the BGS, and hence is called the \textbf{enhanced global symmetry (EGS)}. In Section \ref{sec: HDs} we will see how Hasse diagrams can also be used to investigate the global symmetry, but for now we start by reviewing the algorithm for evaluating the BGS.

\subsection{Balance and Global Symmetry}
\label{sec:BGS}
The concept of the balance of nodes in a quiver originated from quivers which have a brane system \cite{Hanany:1996ie}. The balance of a unitary node is related to the net charge difference of the two NS5 branes which trap the D3 branes that form the gauge group corresponding to the node. In three dimensions, this quantity is of interest because it leads to extra operators appearing as we flow to the infrared, which elicit more symmetries than seen in the UV theory.\\

\noindent A node is termed balanced if its excess is equal to zero. The definitions of unbalanced nodes, and more specifically overbalanced and underbalanced nodes, all follow naturally. The excess of a node $i$ in a quiver is given by the total number of flavours it sees, minus twice its rank. That is, if $f_i$ is the number of flavours $i$ sees and $r_i$ is its rank, then its excess $e_i$ is given by \cite{Gaiotto:2008ak}
\begin{equation}\label{excess defn}
    e_i = f_i - 2r_i.
\end{equation}
When we say the number of flavours $i$ sees, we refer to the number of hypermultiplets transforming under the gauge group represented by the node $i$ in a quiver, if no other nodes were gauge nodes. For example, in the quiver for (the closure of)\footnote{We will largely omit this ``closure of" when discussing varieties as nilpotent orbits in the future for ease, but it should always be this.} the minimal nilpotent orbit of $G_2$, given by the balanced affine Dynkin diagram of $G_2$,
\begin{equation}
    \begin{tikzpicture}[x=1cm,y=.8cm]
\node (g1) at (-1,0) [gauge,label=below:{$1$},label={[red]above:{$i_1$}}] {};
\node (g3) at (0,0) [gauge,label=below:{$2$},label={[red]above:{$i_2$}}] {};
\node (g4) at (1,0) [gauge,label=below:{$1$},label={[red]above:{$i_3$}}] {};
\node (g5) at (1.5,0) {$,$};
\draw (g1)--(g3);
\draw[transform canvas = {yshift=-2pt}] (g4)--(g3);
\draw[transform canvas = {yshift=0pt}] (g4)--(g3);
\draw[transform canvas = {yshift =2pt}] (g4)--(g3);
\draw (0.4,0.2)--(0.6,0)--(0.4,-0.2);
\end{tikzpicture}
\end{equation}
the node on the left sees $f_{i_1}=2$ hypermultiplets in its fundamental representation and thus its excess is $e_{i_1}=2-(2 \times 1)=0$, the node in the centre sees $f_{i_2}=1 + (3 \times 1)=4$ hypermultiplets (in its fundamental representation) and thus its excess is $e_{i_2}=4-(2 \times 2)=0$, and the node on the right sees $f_{i_3}=2$ hypermultiplets in its fundamental representation and thus its excess is $e_{i_3}=2-(2 \times 1)=0$. Because all gauge nodes in this quiver\footnote{Here, all the nodes in the quiver are gauge nodes. Quivers for which this is the case are termed \textit{unframed}. If a quiver contains at least one flavour node, it is termed \textit{framed}.} are balanced, we say the quiver is a balanced quiver. If a quiver contains unbalanced gauge nodes, then after ungauging (see below) the balanced nodes naturally split into connected subsets which must take the form of a Dynkin diagram of some Lie algebra. We call these balanced sub-quivers \textit{sub-Dynkin diagrams}.\\

\noindent We are almost ready to describe how to read the global symmetry of quivers of type $Q_r$ from the balance of their nodes, but there is one more important concept which needs introducing: ungauging. In an unframed unitary quiver, there is an overall $U(1)$ which acts trivially and thus needs to be ungauged. For the brane picture, this means the brane system is translationally invariant, and so ungauging this $U(1)$ corresponds to fixing an origin. When computing the Hilbert series using the monopole formula, this manifests as setting a single magnetic charge for one of the gauge nodes to zero\footnote{If we ungauge on a rank one gauge node, this causes it to turn into a rank one flavour node.}  \cite{Cremonesi:2013lqa,Hanany:2020jzl}. The same Coulomb branch is obtained regardless of which (long) node is ungauged in a (non-) simply laced quiver. Ungauging on a short node\footnote{This is a slight abuse of terminology: by ``ungauging on a short node" what we actually mean is to promote the short node to a long node, and then ungauge this long node in the resulting quiver. The phrase ``ungauging on a short node" is meaningless as short nodes do not have a gauge theoretic interpretation.} in a non-simply laced quiver can only be done if the short node is of rank one,\footnote{This is because ungauging on a short node of rank greater than one causes the lattice of magnetic charges to deform such that it no longer corresponds to the weights of a Lie group.} and yields a different Coulomb branch to that obtained by ungauging on a long node. Throughout this paper, in the case of non-simply laced quivers, the Coulomb branch we discuss will always be that obtained from ungauging on a long node. This concept of ungauging will appear in the BGS algorithm, which we are now ready to present.

\level{4}*{BGS algorithm}
\label{Alg: BGS alg}
\noindent The balance global symmetry for a unitary quiver containing only regular matter and whose Coulomb branch is a HyperK\"{a}hler cone, $Q_r$, is computed as follows:
\begin{inparaenum}
    \item If $Q_r$ is framed, and the gauge nodes in it comprise $s$ balanced sub-Dynkin diagrams $D_i$ of the simply connected Lie groups $G_i$, and $k$ unbalanced nodes, then the BGS is given by 
        \begin{equation} \label{BGS framed}
            BGS_{Q_r}^{framed}= \prod_{i=1}^{s} G_i \times U(1)^k.
        \end{equation}
    \item If $Q_r$ is unframed, there are two scenarios:
    \begin{inparaenum}
        \item If $Q_r$ is balanced, then it must be either the affine or twisted affine quiver for some simply connected Lie group $G$ (these quivers are given in Tables \ref{tab:affine Dds} and \ref{tab:twisted affine Dds} respectively), in which case
            \begin{equation}
                BGS_{Q_r}=G.
            \end{equation}
        \item If $Q_r$ is unbalanced, ungauge on an unbalanced node\footnote{Recall that in this note, we only study those quivers for which we ungauge on a long node.} of rank one if such a node exists, or of any rank if not. If this node was of rank one, $Q_r$ will now be framed, and the BGS can be read using step $1$. If this node was of rank greater than one, the quiver will still be unframed and the BGS is given by
            \begin{equation} \label{BGS unframed after ungauging}
                 BGS_{Q_r}^{unframed}= \prod_{i=1}^{s} G_i \times U(1)^{k-1}.
            \end{equation}
        where the definitions of $s$, $G_i$ and $k$ are as in step $1$.
    \end{inparaenum}
\end{inparaenum}

\clearpage

\begin{table}
    \centering
    
 \hspace*{-0.5cm}\begin{tabular}{|c|c|} \hline
 
Coulomb branch & Affine Quiver\\

\hline

$\begin{array}{c}
	a_k\\
	\end{array}$ &
$\raisebox{-.5\height}{\begin{tikzpicture}[x=1cm,y=.8cm]
\node (g8) at (-1,0) [gauge,label=below:{$1$}] {};
\node (g5) at (0,0) {$\cdots$};
\node (g6) at (1,0) [gauge,label=below:{$1$}] {};
\node (g7) at (0,1) [gauge,label=above:{$1$}] {};
\draw (g8)--(g5)--(g6)--(g7)--(g8);
\draw [decorate,decoration={brace,mirror,amplitude=6pt}] (-1,-0.8) --node[below=6pt] {$k$} (1,-0.8);
\end{tikzpicture}}$ \\
\hline

$\begin{array}{c}
	b_k\\
	\\
	k \geq 3
	\end{array}$ &
$\raisebox{-.5\height}{\begin{tikzpicture}[x=1cm,y=.8cm]
\node (g2) at (-2,0) [gauge,label=below:{$1$}] {};
\node (g3) at (-1,0) [gauge,label=below:{$2$}] {};
\node (g4) at (0,0) {$\cdots$};
\node (g5) at (1,0) [gauge,label=below:{$2$}] {};
\node (g12) at (2,0) [gauge,label=below:{$1$}] {};
\node (g11) at (-1,1) [gauge,label=above:{$1$}] {};
\draw (g2)--(g3)--(g4)--(g5);
\draw (g3)--(g11);
\draw[transform canvas={yshift=-1.5pt}] (g5)--(g12);
\draw[transform canvas={yshift=1.5pt}] (g5)--(g12);
\draw (1.4,-0.2)--(1.6,0)--(1.4,0.2);
\draw [decorate,decoration={brace,mirror,amplitude=6pt}] (-1,-0.8) --node[below=6pt] {$k-2$} (1,-0.8);
\end{tikzpicture}}$ \\
\hline

$\begin{array}{c}
	c_k\\
	\\
	k \geq 2
	\end{array}$ &
$\raisebox{-.5\height}{\begin{tikzpicture}[x=1cm,y=.8cm]
\node (g2) at (-2,0) [gauge,label=below:{$1$}] {};
\node (g3) at (-1,0) [gauge,label=below:{$1$}] {};
\node (g4) at (0,0) {$\cdots$};
\node (g5) at (1,0) [gauge,label=below:{$1$}] {};
\node (g12) at (2,0) [gauge,label=below:{$1$}] {};
\draw (g3)--(g4)--(g5);
\draw[transform canvas={yshift=-1.5pt}] (g2)--(g3);
\draw[transform canvas={yshift=1.5pt}] (g2)--(g3);
\draw (-1.6,-0.2)--(-1.4,0)--(-1.6,0.2);
\draw[transform canvas={yshift=-1.5pt}] (g5)--(g12);
\draw[transform canvas={yshift=1.5pt}] (g5)--(g12);
\draw (1.6,-0.2)--(1.4,0)--(1.6,0.2);
\draw [decorate,decoration={brace,mirror,amplitude=6pt}] (-1,-0.8) --node[below=6pt] {$k-1$} (1,-0.8);
\end{tikzpicture}}$ \\
\hline

$\begin{array}{c}
	d_k\\
	\\
	k\geq 4
	\end{array}$ &
$\raisebox{-.5\height}{\begin{tikzpicture}[x=1cm,y=.8cm]
\node (g2) at (-2,0) [gauge,label=below:{$1$}] {};
\node (g3) at (-1,0) [gauge,label=below:{$2$}] {};
\node (g4) at (0,0) {$\cdots$};
\node (g5) at (1,0) [gauge,label=below:{$2$}] {};
\node (g12) at (2,0) [gauge,label=below:{$1$}] {};
\node (g13) at (1,1) [gauge,label=above:{$1$}] {};
\node (g11) at (-1,1) [gauge,label=above:{$1$}] {};
\draw (g2)--(g3)--(g4)--(g5)--(g12);
\draw (g13)--(g5);
\draw (g11)--(g3);
\draw [decorate,decoration={brace,mirror,amplitude=6pt}] (-1,-0.8) --node[below=6pt] {$k-3$} (1,-0.8);
\end{tikzpicture}}$ \\
\hline

$\begin{array}{c}
	e_6
	\end{array}$ &
$\raisebox{-.5\height}{\begin{tikzpicture}[x=1cm,y=.8cm]
\node (g2) at (-2,0) [gauge,label=below:{$1$}] {};
\node (g3) at (-1,0) [gauge,label=below:{$2$}] {};
\node (g4) at (0,0) [gauge,label=below:{$3$}] {};
\node (g5) at (1,0) [gauge,label=below:{$2$}] {};
\node (g12) at (2,0) [gauge,label=below:{$1$}] {};
\node (g10) at (0,1) [gauge,label=right:{$2$}] {};
\node (g11) at (0,2) [gauge,label=right:{$1$}] {};
\draw (g2)--(g3)--(g4)--(g5)--(g12);
\draw (g4)--(g10)--(g11);
\end{tikzpicture}}$ \\
\hline

$\begin{array}{c}
	e_7
	\end{array}$ &
$\raisebox{-.5\height}{\begin{tikzpicture}[x=1cm,y=.8cm]
\node (g3) at (-3,0) [gauge,label=below:{$1$}] {};
\node (g4) at (-2,0) [gauge,label=below:{$2$}] {};
\node (g13) at (-1,0) [gauge,label=below:{$3$}] {};
\node (g5) at (0,0) [gauge,label=below:{$4$}] {};
\node (g12) at (1,0) [gauge,label=below:{$3$}] {};
\node (g6) at (2,0) [gauge,label=below:{$2$}] {};
\node (g10) at (3,0) [gauge,label=below:{$1$}] {};
\node (g7) at (0,1) [gauge,label=above:{$2$}] {};
\draw (g3)--(g4)--(g13)--(g5)--(g12)--(g6)--(g10);
\draw (g7)--(g5);
\end{tikzpicture}}$ \\
\hline

$\begin{array}{c}
	e_8
	\end{array}$ &
$\raisebox{-.5\height}{\begin{tikzpicture}[x=1cm,y=.8cm]
\node (g3) at (-3.5,0) [gauge,label=below:{$1$}] {};
\node (g4) at (-2.5,0) [gauge,label=below:{$2$}] {};
\node (g5) at (-1.5,0) [gauge,label=below:{$3$}] {};
\node (g6) at (-0.5,0) [gauge,label=below:{$4$}] {};
\node (g7) at (0.5,0) [gauge,label=below:{$5$}] {};
\node (g8) at (1.5,0) [gauge,label=below:{$6$}] {};
\node (g9) at (2.5,0) [gauge,label=below:{$4$}] {};
\node (g10) at (3.5,0) [gauge,label=below:{$2$}] {};
\node (g11) at (1.5,1) [gauge,label=above:{$3$}] {};
\draw (g3)--(g4)--(g5)--(g6)--(g7)--(g8)--(g9)--(g10);
\draw (g8)--(g11);
\end{tikzpicture}}$ \\
\hline

$\begin{array}{c}
    \\
	f_4\\
	\\
	\end{array}$ &
$\raisebox{-.5\height}{\begin{tikzpicture}[x=1cm,y=.8cm]
\node (g2) at (-2,0) [gauge,label=below:{$1$}] {};
\node (g3) at (-1,0) [gauge,label=below:{$2$}] {};
\node (g4) at (0,0) [gauge,label=below:{$3$}] {};
\node (g5) at (1,0) [gauge,label=below:{$2$}] {};
\node (g12) at (2,0) [gauge,label=below:{$1$}] {};
\draw (g2)--(g3)--(g4);
\draw (g12)--(g5);
\draw[transform canvas={yshift=-1.5pt}] (g4)--(g5);
\draw[transform canvas={yshift=1.5pt}] (g4)--(g5);
\draw (0.4,-0.2)--(0.6,0)--(0.4,0.2);
\end{tikzpicture}}$ \\
\hline

$\begin{array}{c}
    \\
    g_2\\
    \\
    \end{array}$ &
$\raisebox{-.5\height}{\begin{tikzpicture}[x=1cm,y=.8cm]
\node (g3) at (-1,0) [gauge,label=below:{$1$}] {};
\node (g4) at (0,0) [gauge,label=below:{$2$}] {};
\node (g5) at (1,0) [gauge,label=below:{$1$}] {};
\draw (g3)--(g4);
\draw[transform canvas = {yshift=-2pt}] (g4)--(g5);
\draw[transform canvas = {yshift=0pt}] (g4)--(g5);
\draw[transform canvas = {yshift =2pt}] (g4)--(g5);
\draw (0.4,0.2)--(0.6,0)--(0.4,-0.2);
\end{tikzpicture}}$ \\
\hline

\end{tabular}
 \caption{Quivers of the affine Dynkin diagrams. The first column lists the Coulomb branch varieties, which are all the minimal nilpotent orbits $g_k$ of the semi-simple Lie groups $G_k$, of the quivers in the second column when ungauged on a long node. These quivers will be referred to as ``affine $g_k$".}
    \label{tab:affine Dds}
\end{table}

\begin{table}[hbt!]
    \centering
    
 \hspace*{-0.5cm}\begin{tabular}{|c|c|} \hline
 
Coulomb Branch & Twisted Affine Quiver\\

\hline

$\begin{array}{c}
	\\
	a_2\\
	\\
	\end{array}$ &
$\raisebox{-.5\height}{\begin{tikzpicture}[x=1cm,y=.8cm]
\node (g1) at (-0.5,0) [gauge,label=below:{$2$}] {};
\node (g2) at (0.5,0) [gauge,label=below:{$1$}] {};
\draw[transform canvas={yshift=-1pt}] (g1)--(g2);
\draw[transform canvas={yshift=-3pt}] (g1)--(g2);
\draw[transform canvas={yshift=3pt}] (g1)--(g2);
\draw[transform canvas={yshift=1pt}] (g1)--(g2);
\draw (-0.1,0.2)--(0.1,0)--(-0.1,-0.2);
\end{tikzpicture}}$ \\
\hline

$\begin{array}{c}
	\\
	a_{2k-1}\\
	\\
	\end{array}$ &
$\raisebox{-.5\height}{\begin{tikzpicture}[x=1cm,y=.8cm]
\node (g1) at (-2,0) [gauge,label=below:{$2$}] {};
\node (g2) at (-1,0) [gauge,label=below:{$2$}] {};
\node (g3) at (0,0) {$\cdots$};
\node (g4) at (1,0) [gauge,label=below:{$2$}] {};
\node (g8) at (1,1) [gauge,label=above:{$1$}] {};
\node (g5) at (2,0) [gauge,label=below:{$1$}] {};
\draw (g2)--(g3)--(g4)--(g8);
\draw (g4)--(g5);
\draw[transform canvas={yshift=-1.5pt}] (g1)--(g2);
\draw[transform canvas={yshift=1.5pt}] (g1)--(g2);
\draw (-1.6,0.2)--(-1.4,0)--(-1.6,-0.2);
\draw [decorate,decoration={brace,mirror,amplitude=6pt}] (-1,-0.8) --node[below=6pt] {$k-2$} (1,-0.8);
\end{tikzpicture}}$\\
\hline

$\begin{array}{c}
    \\
	a_{2k}\\
	\\
	k\geq 2
	\end{array}$ &
$\raisebox{-.5\height}{\begin{tikzpicture}[x=1cm,y=.8cm]
\node (g1) at (-2,0) [gauge,label=below:{$2$}] {};
\node (g2) at (-1,0) [gauge,label=below:{$2$}] {};
\node (g3) at (0,0) {$\cdots$};
\node (g4) at (1,0) [gauge,label=below:{$2$}] {};
\node (g8) at (2,0) [gauge,label=below:{$1$}] {};
\draw (g2)--(g3)--(g4);
\draw[transform canvas={yshift=-1.5pt}] (g1)--(g2);
\draw[transform canvas={yshift=1.5pt}] (g1)--(g2);
\draw (-1.6,0.2)--(-1.4,0)--(-1.6,-0.2);
\draw[transform canvas={yshift=-1.5pt}] (g4)--(g8);
\draw[transform canvas={yshift=1.5pt}] (g4)--(g8);
\draw (1.4,0.2)--(1.6,0)--(1.4,-0.2);
\draw [decorate,decoration={brace,mirror,amplitude=6pt}] (-1,-0.8) --node[below=6pt] {$k-1$} (1,-0.8);
\end{tikzpicture}}$\\
\hline

$\begin{array}{c}
	\\
	d_4\\
	\\
	\end{array}$ &
$\raisebox{-.5\height}{\begin{tikzpicture}[x=1cm,y=.8cm]
\node (g1) at (-1,0) [gauge,label=below:{$3$}] {};
\node (g2) at (0,0) [gauge,label=below:{$2$}] {};
\node (g3) at (1,0) [gauge,label=below:{$1$}] {};
\draw (g2)--(g3);
\draw [transform canvas={yshift=-2pt}] (g1)--(g2);
\draw [transform canvas={yshift=0pt}] (g1)--(g2);
\draw [transform canvas={yshift=2pt}] (g1)--(g2);
\draw (-0.6,0.2)--(-0.4,0)--(-0.6,-0.2);
\end{tikzpicture}}$ \\
\hline

$\begin{array}{c}
    \\
	d_{k+1}\\
	\\
	k\geq 4
	\end{array}$ &
$\raisebox{-.5\height}{\begin{tikzpicture}[x=1cm,y=.8cm]
\node (g1) at (-2,0) [gauge,label=below:{$1$}] {};
\node (g2) at (-1,0) [gauge,label=below:{$2$}] {};
\node (g3) at (0,0) {$\cdots$};
\node (g4) at (1,0) [gauge,label=below:{$2$}] {};
\node (g8) at (2,0) [gauge,label=below:{$1$}] {};
\draw (g2)--(g3)--(g4);
\draw[transform canvas={yshift=-1.5pt}] (g2)--(g1);
\draw[transform canvas={yshift=1.5pt}] (g2)--(g1);
\draw (-1.4,0.2)--(-1.6,0)--(-1.4,-0.2);
\draw[transform canvas={yshift=-1.5pt}] (g4)--(g8);
\draw[transform canvas={yshift=1.5pt}] (g4)--(g8);
\draw (1.4,0.2)--(1.6,0)--(1.4,-0.2);
\draw [decorate,decoration={brace,mirror,amplitude=6pt}] (-1,-0.8) --node[below=6pt] {$k-1$} (1,-0.8);
\end{tikzpicture}}$\\
\hline

$\begin{array}{c}
	e_{6}\\
	\\
	k\geq 2
	\end{array}$ &
$\raisebox{-.5\height}{\begin{tikzpicture}[x=1cm,y=.8cm]
\node (g1) at (-2,0) [gauge,label=below:{$2$}] {};
\node (g2) at (-1,0) [gauge,label=below:{$4$}] {};
\node (g3) at (0,0) [gauge,label=below:{$3$}] {};
\node (g4) at (1,0) [gauge,label=below:{$2$}] {};
\node (g8) at (2,0) [gauge,label=below:{$1$}] {};
\draw (g1)--(g2);
\draw (g3)--(g4)--(g8);
\draw[transform canvas={yshift=-1.5pt}] (g2)--(g3);
\draw[transform canvas={yshift=1.5pt}] (g2)--(g3);
\draw (-0.6,0.2)--(-0.4,0)--(-0.6,-0.2);
\end{tikzpicture}}$\\
\hline

\end{tabular}

 \caption{Quivers of the twisted affine Dynkin diagrams. The first column lists the Coulomb branch varieties, which are all the minimal nilpotent orbits $g_k$ of the semi-simple Lie groups $G_k$, of the quivers in the second column when ungauged on a long node. These quivers will be referred to as ``twisted affine $g_k$".}
    \label{tab:twisted affine Dds}
\end{table}

\clearpage

\noindent As said above, the Hilbert series and thus global symmetry is invariant regardless of which long node the ungauging takes place on. The reason for the discrepency in the power of $U(1)$ between the framed and unframed case in steps $2a$ and $2b$ respectively is because in the former case, the ungauging physically removes a gauge node and thus changes $k$ to $k-1$, as one gauge node becomes a flavour. This doesn't happen in the latter case, and thus this must be manually accounted for in (\ref{BGS unframed after ungauging}).\\

\noindent However, as noted in the introduction, this algorithm does not always yield the full global symmetry. To see this explicitly, consider the following example.

\paragraph{Example} Consider the following quiver
\begin{equation}\label{double e6 RHS example quiver}
    \begin{tikzpicture}[x=1cm,y=.8cm]
    \node (g0) at (-3,0) {$Q=$};
    \node (g2) at (-2,0) [gauge,label=below:{$2$}] {};
    \node (g3) at (-1,0) [gauge,label=below:{$4$}] {};
    \node (g4) at (0,0) [gauge,label=below:{$6$}] {};
    \node (g5) at (1,0) [gauge,label=below:{$8$}] {};
    \node (g6) at (2,0) [gauge, label=below:{$5$}] {};
    \node (g7) at (1,1) [gauge, label=left:{$5$}] {};
    \node (g8) at (1,2) [gauger, label=left:{$2$}] {};
    \node (g9) at (3,0) [gauger, label=below:{$2$}] {};
    \draw (g2)--(g3)--(g4)--(g5)--(g6)--(g9);
    \draw (g5)--(g7)--(g8);
    \end{tikzpicture}
\end{equation}
where the red indicates that the corresponding node is unbalanced. Here, after ungauging on one of the unbalanced nodes the quiver is still unframed, and there are two unbalanced nodes and one balanced sub-$D_6$ Dynkin diagram. Thus we read off the BGS
\begin{equation}
BGS_Q=\textcolor{blue}{SO(12)} \times U(1).
\end{equation}
However if we compute the refined Hilbert series to order $t^2$ for $Q$ we find that, after the appropriate fugacity map (see Appendix \ref{app:fugmap} for more details),
\begin{equation}
    HS_Q = 1 + (1+[0,1,0,0,0,0]_{B_6})t^2 + \mathcal{O}(t^4),
\end{equation}
which means that the global symmetry is actually enhanced to
\begin{equation}
EGS_Q=\textcolor{blue}{SO(13)} \times U(1).
\end{equation}
Blue is used to highlight the factor of the global symmetry that is enhanced. Here we see explicitly the failure of the BGS to give the full global symmetry. \hfill $\square$

\subsection{Hasse Diagrams} \label{sec: HDs}
One can also glean information about the non-Abelian part of the global symmetry of a quiver via studying its Hasse diagram \cite{1982:KP,Bourget:2019aer,Grimminger:2020dmg}. The Hasse diagram characterises the structure of a moduli space by visualising its stratification as a symplectic singularity into so-called symplectic leaves and transverse slices. Physically, a symplectic leaf corresponds to a certain set of massless states (or \textit{phase}) admitted by the theory on this section of the moduli space. A transverse slice is then the set of moduli one needs to tune to move from one phase of the theory to another. The symmetry of the lowest elementary slices of the Hasse diagram form a subgroup of the non-Abelian part of the global symmetry.\\

\noindent The Hasse diagrams for the Coulomb branch and Higgs branch are equivalently constructed in different ways: the former uses a ``bottom-up" approach, while the latter uses a ``top-down" one. The equivalence of these approaches can be seen from the brane picture or the $3d$ mirror of a theory, in the cases where these exist. The bottom-up approach is to use the Higgs mechanism to see what moduli need to be adjusted to give different combinations of massive and massless gauge bosons until the gauge group is maximally broken, and then use this to build up the picture of symplectic leaves and transverse slices. However this approach only works when the theory has a classical Lagrangian description. For theories where this is not the case, for instance due to a strong coupling limit, there is an alternative approach if a magnetic quiver\footnote{A magnetic quiver for an algebraic variety is a quiver whose Coulomb branch is equal to the variety. An electric quiver, or Higgs quiver, for an algebraic variety is a quiver whose Higgs branch is equal to the variety.} exists. If this is the case, the top-down method of finding the Coulomb branch Hasse diagram for the magnetic quiver can be used. The method for doing this is derived from the brane picture of the theory. Tuning moduli to move from a leaf with a certain set of massless states to another leaf with additional massless states corresponds to performing a Kraft-Processi transition \cite{Cabrera:2016vvv,Cabrera:2017njm}, where one alligns a set of branes. When a minimal number of moduli are tuned (or equivalently a minimal number of branes are alligned) to move to a different phase of the theory, the transverse slice corresponding to these moduli is called an elementary slice. From these minimal transitions, all composite ones can be obtained. The process of alligning branes has been translated into an operation on magnetic quivers called quiver subtraction \cite{Cabrera:2018ann}. Subtracting magnetic quivers of elementary slices from the magnetic quiver of a theory constructs its Hasse diagram. Appendix \ref{app:quiver subtraction} gives the current rules for quiver subtraction and an example of using it to compute a Hasse diagram.\\

\noindent In this work, quiver subtraction and its reverse process, quiver addition, are exploited to find quivers with enhanced Coulomb branch global symmetry. Exactly how this is done is explained in the next section. Although quiver subtraction was derived from brane systems, in the absence of a brane picture the rules for subtracting quivers can often still be implemented, although a physical motivation for doing so is absent and the Hasse diagram derived cannot be verified. We can however confirm that, for all quivers studied in this paper, at least the lowest slices in their conjectured Hasse diagrams are correct, as they agree with the correct global symmetries computed by the monopole formula.\\

\noindent This concludes the recap of the relevant topics needed to proceed with the rest of the paper. We now move on to describe the method by which we construct quivers with EGS.\\

\section{Constructing Quivers with Enhanced Global Symmetry} \label{sec: addition method}

The gauging of discrete symmetries in supersymmetric quiver gauge theories has received much attention in recent years. The notion of gauging continuous symmetries is a familiar one: introducing fabricated degrees of freedom to elicit mathematical trickery in order to simplify problems. Many theories also have global discrete symmetries, such as time reversal or charge conjugation, and so exploring the potential to gauge these symmetries should be of equal interest. In the context of the moduli spaces of supersymmetric quiver gauge theories, this has been explored in  \cite{Hanany:2018vph,Hanany:2018cgo,Hanany:2018dvd,Bourget:2020bxh}, for example. The outer automorphisms of the quiver itself are obvious discrete global symmetries of the theory, and thus we could choose to gauge them. In the string theory picture, this corresponds to moving branes on top of one another, inducing extra massless strings stretching between them. A result of this is the following conjecture:

\begin{equation}\label{bouquet to adjoint conjecture}
    \mathcal{C}\left( \raisebox{-.5\height}{\begin{tikzpicture}[x=1cm,y=.8cm]
    \node (g0) at (-2,0) [gauge,label=left:{$k$}] {};
    \node (g1) at (-2.5,0.8) [gauge,label=above:{$1$}] {};
    \node (g2) at (-2,0.8) {$\cdots$};
    \node (g4) at (-1.5,0.8) [gauge,label=above:{$1$}] {};
    \node (g5) at (-2,-1) {$\vdots$};
    \node (g6) at (-2,-1.3) {};
    \draw (g5)--(g0)--(g1);
    \draw (g0)--(g4);
    \draw [decorate,decoration={brace,amplitude=6pt}] (-2.5,1.6) --node[above=6pt] {$n$} (-1.5,1.6);
    \end{tikzpicture}}
    \right)\Big/S_n \, =\, \mathcal{C}\left(
    \raisebox{-.5\height}{\begin{tikzpicture}[x=1cm,y=.8cm]
    \node (f0) at (2,0) [gauge,label=left:{$k$}] {};
    \node (f1) at (2,1.1) [gauge,label=left:{$n$}] {};
    \node (f5) at (2,-1) {$\vdots$};
    \node (f6) at (2,-1.3) {};
    \draw (f5)--(f0)--(f1);
    \draw (f1) to [out=45,in=135,looseness=10] (f1);
    \end{tikzpicture}}
    \right),
\end{equation}
where $\mathcal{C}(Q)$ is used to denote the Coulomb branch of the quiver $Q$. This conjecture has been proven on the level of the Hilbert series \cite{Hanany:2018cgo}, but this is not sufficient to prove for good that the two Coulomb branches are the same.\\

\noindent Let's be explicit about what this tells us. The quiver on the left hand side, which we'll call $A$, has an $S_n$ outer automorphism permuting the bouquet of $n$ identical $U(1)$ nodes, and so this $S_n$ is a discrete global symmetry of the theory described by $A$. Upon gauging this symmetry, we add further constraining relations to the ring of gauge invariant chiral operators on the Coulomb branch, and thus the Coulomb branch of this gauged theory will be an $S_n$ quotient of that of $A$. Moreover, the magnetic quiver for this discretely gauged theory is found by compiling the bouquet in $A$ together into one single $U(n)$ node with a hypermultiplet in the adjoint representation, as shown in the quiver on the right hand side of (\ref{bouquet to adjoint conjecture}), which we'll call $B$. The intuition for this comes from the brane picture discussed above. From here on out, we refer to such a hypermultiplet (i.e. one that is in the adjoint representation of some gauge group) as an adjoint hypermultiplet, or more vaguely as adjoint matter. They are represented in the quiver by a loop going to and from the node that they are in the adjoint representation of. The conjecture (\ref{bouquet to adjoint conjecture}) has been shown to be true in many cases, and a counterexample has not yet been found, so it is believed to be true in general. For details on how to find an $S_n$ quotient of a moduli space in terms of its Hilbert series, see Appendix \ref{app:discrete projection} in which an example for $n=2$ is illustrated.\\

\noindent In Section \ref{sec:quivers with adj matter} we look at quivers with adjoint matter (whose Coulomb branches, by the above, can be seen as discrete quotients of quivers containing only regular matter\footnoteref{footnote:reg matter}) and see how to read their BGS, before constructing the full set of balanced quivers for which each node has at most one adjoint hypermultiplet. It is these quivers from which in the following sections we will construct, via quiver addition, quivers with only regular matter that enjoy an EGS. The process of this construction is outlined in Section \ref{sec: quiver addition}.

\subsection{Quivers with Adjoint Matter} \label{sec:quivers with adj matter}

The main result of this work is as follows: there are balanced quivers with adjoint matter which we call $Q_a$, listed in Table \ref{tab:adj hyper Qs},\footnote{Note that, as explained in the bullet points at the end of Section \ref{sec:quivers with adj matter}, the quiver with global symmetry $SU(2n+1)$ in Table \ref{tab:adj hyper Qs} is not part of the set of quivers $Q_a$ which we study in this note, as it is unclear how to perform quiver addition in this case.} that induce an EGS. By this, we mean that we have found many quivers, derived by performing quiver addition on $Q_a$ to ``absorb" the adjoint hypermultiplet (we will see what this means more concretely in Section \ref{sec: quiver addition}), which experience a symmetry that is enhanced from that predicted by the BGS. Since the adjoint matter is absorbed in this process, the resulting quivers are of type $Q_r$, as discussed at the beginning of Section \ref{sec: CB GS}. The Hasse diagram of such a quiver $Q_r$ will contain the Hasse diagram of the quiver $Q_a$ it was derived from, and so the global symmetry of $Q_a$ will be a subgroup of the global symmetry of $Q_r$, as explained in Section \ref{sec: HDs}. This motivates us to study the quivers $Q_a$, and in particular their global symmetry. However, the above prescription for reading the BGS does not specify what happens when there are gauge nodes with adjoint matter in the quiver. We will call such gauge nodes $a_i$, $i=1,...,l$, where $l \in \mathbb{N}$ is the number of nodes with adjoint matter. One might naively think to just ``ignore" the adjoint hypermultiplets (other than their contribution to the balance of $a_i$), and proceed with the algorithm given in Section \ref{sec:BGS}. However, as the example below shows, this leads to us identifying incorrect global symmetries. To correct this, we propose the following extension to the BGS algorithm for balanced quivers containing an adjoint hypermultiplet.

\level{4}*{Claim: BGS refinement (the simply laced adjoint hypermultiplet)}
\label{Claim: Alg: BGS refinement}
Let $Q_a$ be a balanced\footnote{This algorithm does also work for some unbalanced quivers with hypermultiplets as specified here. However it does not work for all such quivers, and hence we restrict the validity of the algorithm to only the balanced subset.} quiver which has $l$ nodes that each have a single adjoint hypermultiplet, and no nodes with more than one. Label these nodes $a_i$ and their ranks $r_{a_i}$, for $i=1,...,l$. Furthermore, impose that the nodes $a_i$ are only connected to other gauge nodes in $Q_a$ via simply laced edges. Then the global symmetry of $Q_a$ can be determined by performing the following steps:
\begin{enumerate}
    \item Replace the simply laced connections of the $a_i$ to other gauge node(s) of $Q_a$ by a non-simply laced edge of multiplicity $r_{a_i}$ such that $a_i$ are the short nodes.
    \item Remove the adjoint hypermultiplets attached to all nodes $a_i$.
    \item Set all $r_{a_i}=1$.
    \item  Call the resulting quiver after performing steps $1-3$ $\tilde{Q}_a$. The global symmetry of $Q_a$ is then given by implementing the previous BGS algorithm listed in Section \ref{sec:BGS} on $\tilde{Q}_a$.
\end{enumerate}

\vspace{2mm}

\noindent The justification for this claim is found in examining the Hilbert series. In the monopole formula for the quivers we consider, a topological fugacity $z_i$ is assigned to each gauge group to count the topological charge of the corresponding monopole operators. Under an appropriate fugacity map (see Appendix \ref{app:fugmap} for general details on fugacity maps), the $z_i$ can be mapped to so-called simple root fugacities $y_i$, such that the coefficients of $t^{2k}$ are characters of the global symmetry algebra in terms of its simple roots. In particular, the $t^2$ coefficient in the refined Hilbert series is the character for the adjoint representation of the global symmetry. In the cases of the quivers with adjoint hypermultiplets that we study in this work the global symmetry turns out to be that of $\tilde{Q}_a$, and under the topological-to-simple-root fugacity map, the topological fugacity for $a$ corresponds to the simple root fugacity for the short node in $\tilde{Q}_a$. This motivates the above algorithm, and indeed with all the simply laced adjoint quivers in Table \ref{tab:adj hyper Qs} that we study in this work it gives the correct global symmetry. We illustrate the need for and implementation of this claim below with the following example.

\paragraph{Example} Consider the quiver
\begin{equation} \label{nminB6}
    \begin{tikzpicture}
    \node (g0) at (-2,1.1) [gauge,label=right:{$2$},label={[red]left:{$a$}}] {};
    \node (g1) at (-2,0) [gauge,label=below:{$2$}] {};
    \node (g2) at (-1,0) [gauge,label=below:{$2$}] {};
    \node (g4) at (0,0) [gauge,label=below:{$2$}] {};
    \node (g5) at (1,0) [gauge,label=below:{$2$}] {};
    \node (g6) at (2,0) [gauge,label=below:{$1$}] {};
    \node (g7) at (1,1) [gauge,label=left:{$1$}] {};
    \draw (g0)--(g1)--(g2)--(g4)--(g5)--(g6);
    \draw (g5)--(g7);
    \draw (g0) to [out=45,in=135,looseness=10] (g0);
    \end{tikzpicture}
\end{equation}
This is a magnetic quiver of the Kostant Brylinski classification \cite{1992math......4227B}, as its Coulomb branch is the $\mathbb{Z}_2$ quotient of the minimal nilpotent orbit of $D_7$ \cite{Hanany:2018dvd}, which is the next to minimal nilpotent orbit of $B_6$ \cite{Bourget:2020bxh}, $\mathcal{O}_{(3,1^{10})}$. Thus the global symmetry of (\ref{nminB6}) is $B_6$. Let's check this against the initial BGS algorithm we gave in Section \ref{sec:BGS}, if we were to treat $a$ as we would any other node. Since all the nodes in this quiver are balanced, we would choose to ungauge on one of the rank one nodes, giving 
\begin{equation}
    \begin{tikzpicture}
    \node (g0) at (-2,1.1) [gauge,label=right:{$2$},label={[red]left:{$a$}}] {};
    \node (g1) at (-2,0) [gauge,label=below:{$2$}] {};
    \node (g2) at (-1,0) [gauge,label=below:{$2$}] {};
    \node (g4) at (0,0) [gauge,label=below:{$2$}] {};
    \node (g5) at (1,0) [gauge,label=below:{$2$}] {};
    \node (g6) at (1,1) [flavor,label=left:{$1$}] {};
    \node (g7) at (2,0) [gauge,label=below:{$1$}] {};
    \node (g8) at (2.5,0) {$.$};
    \draw (g0)--(g1)--(g2)--(g4)--(g5)--(g6);
    \draw (g5)--(g7);
    \draw (g0) to [out=45,in=135,looseness=10] (g0);
    \end{tikzpicture}
\end{equation}
Then since there are no unbalanced nodes, and among the gauge nodes there is just one balanced sub-Dynkin diagram which is in the shape of $A_6$, we would conclude that the global symmetry is $SU(7)$, which is incorrect. However, under the BGS refinement above, to correctly read the global symmetry, the connection of $a$ to its adjacent node is replaced by a non-simply laced edge of multiplicity two such that $a$ is the short node, its adjoint hypermultiplet is removed, and its rank is set to one. Then we proceed with the original BGS algorithm and ungauge on a (long) rank one node as above to give
\begin{equation}\label{eg b6}
    \begin{tikzpicture}
    \node (g0) at (-2.5,0) [gauge,label=below:{$1$},label={[red]left:{$a$}}] {};
    \node (g1) at (-1.5,0) [gauge,label=below:{$2$}] {};
    \node (g2) at (-0.5,0) [gauge,label=below:{$2$}] {};
    \node (g4) at (0.5,0) [gauge,label=below:{$2$}] {};
    \node (g5) at (1.5,0) [gauge,label=below:{$2$}] {};
    \node (g6) at (1.5,1) [flavor,label=left:{$1$}] {};
    \node (g7) at (2.5,0) [gauge,label=below:{$1$}] {};
    \node at (3,0) {$,$};
    \draw (g1)--(g2)--(g4)--(g5)--(g6);
    \draw (g5)--(g7);
    \draw[transform canvas={yshift=-1.5pt}] (g0)--(g1);
    \draw[transform canvas={yshift=1.5pt}] (g0)--(g1);
    \draw (-1.9,0.2)--(-2.1,0)--(-1.9,-0.2);
    \end{tikzpicture}
\end{equation}
which is a framed quiver, whose balanced subset of gauge nodes forms the Dynkin diagram of $B_6$, which is indeed the correct global symmetry. Note that although the quiver (\ref{eg b6}) has the same global symmetry as (\ref{nminB6}), it does not have the same moduli space. This algorithm should only be used to determine the global symmetry. \hfill $\square$ \\

\noindent It is worth noting that while this refinement for the BGS algorithm is nice, it does not fulfill the requirement of an all encompassing method for determining the global symmetry, as it still fails in many cases.\footnote{When we say the BGS fails, we mean that it fails to give the \textit{full} global symmetry. This is not technically a failure of the BGS as it only claims to give a subgroup of the full global symmetry, but it is a failure of the BGS as an algorithm to always give us the full global symmetry of a quiver, hence our use of the terminology.} See for instance the example of (\ref{double e6 RHS example quiver}), which is a quiver of type $Q_r$ and therefore has no adjoint hypermultiplet, so the previous BGS failure is not fixed by the above ammendment. \\

\noindent The reason we are interested in these quivers with adjoint matter is because, as is explained in Section \ref{sec: quiver addition}, when quiver addition is performed on them to absorb the node(s) with the adjoint hypermultiplet, the resulting quiver experiences an enhanced global symmetry.\\

\noindent There are many quivers with adjoint hypermultiplets which will elicit an EGS in quivers derived from them via quiver addition, but in this work we restrict ourselves to just a subset: the family of balanced quivers whose nodes each have at most one adjoint hypermultiplet. These are the quivers we refer to as $Q_a$. They can be completely classified, and the full classification is listed in Table \ref{tab:adj hyper Qs}, along with the Coulomb branch variety, enhanced global symmetry and highest weight generating function (HWG) \cite{Hanany:2014dia}.\footnote{The HWG is a way of simplifying the notation of the refined Hilbert series by encoding a representation of the global topological symmetry by fugacities graded by the highest weight label of that representation as opposed to detailing all weights.} It is sufficient to, for now, restrict our study to these quivers, as they can be used to construct a healthy range of examples for which the BGS fails. We proceed with proving the classification of these $Q_a$.

\begin{theorem}
The set of $3d$ $\mathcal{N}=4$ balanced unframed unitary quivers with $l \in \mathbb{N} \setminus \{0\}$ nodes $a_i, \ i=1,...,l$ that each have a single adjoint hypermultiplet but which otherwise only contain regular matter can be completely classified, and the classification is given by Table \ref{tab:adj hyper Qs}.
\end{theorem}

\begin{proof}
First start by considering a single node $a$ of rank $r_a$ with a single adjoint hypermultiplet, which all such quivers in the theroem must contain. Note that $r_a>1$ because the adjoint representation of $U(1)$ is trivial. To be balanced, $a$ must be attached to exactly two flavours. This limits us to: $(i)$ one rank two node attached via a simply laced edge; $(ii)$ one rank two node with an adjoint hypermultiplet attached via a simply laced edge; $(iii)$ two rank one nodes, each attached via a simply laced edge; $(iv)$ one short rank one node attached by a non-simply laced edge of multiplicity two; or $(v)$ one long rank two node attached by a non-simply laced edge of multiplicity two. Note that we cannot have one rank one node connected to $a$ via two simply laced hypermultiplets because for this node to be balanced it would force $r_a=1$. In the cases $(ii)$, $(iii)$, $(iv)$ and $(v)$, for these new nodes added to be balanced, and because $r_a \neq 1$,\footnote{This constraint is relevant to mention because if $r_a=1$ were allowed, the new node could have been balanced by attaching further nodes to it.} it must be that $r_a=2$. Furthermore, this must be the end of the chain, as these final nodes added are already balanced so we cannot possibly add anything more to them. In case $(i)$ we have three possibilities, as $r_a \leq 4$ otherwise the new node will be overbalanced, and we know $r_a \neq 1$. Call this new rank two node attached to $a$ in $(i)$ node $T$. If $r_a=4$, the chain ends here as $T$ is already balanced. If $r_a=3$, then $T$ must be connected to one flavour to be balanced, and this ends the chain. If $r_a=2$, then in order for $T$ to be balanced, it must connect to two flavours also. This will again have the same options $(i)$,$(ii)$, $(iii)$, $(iv)$ and $(v)$ as above, and again if option $(i)$ is chosen the chain will continue, until at some point it must end by choosing one of the other options. Running through all these possibilities gives the full list of quivers in Table \ref{tab:adj hyper Qs}.
\end{proof}

\noindent The balanced quivers with adjoint hypermultiplets in Table \ref{tab:adj hyper Qs}, except for the $n=2$ cases of $b_n/\mathbb{Z}_2$ and $a_{2n-1}/\mathbb{Z}_2$, could also be derived in an alternative way. Readers who are confident in the material may wish to skip this paragraph and proceed to the bullet-pointed remarks on Table \ref{tab:adj hyper Qs} below, as its main use is to clarify ideas already introduced. Recall from Section \ref{sec: HDs} that the lowest elementary slices of the Hasse diagram reveal a subgroup of the global symmetry. We therefore try to construct the most basic building block quivers with adjoint matter which induce an EGS.\footnote{Note that the Coulomb branches of these building block quivers are not elementary slices themselves, rather they are derived from them.} This can be done by finding the quivers $Q_a$ for whom $\tilde{Q}_a$, as named in the refined BGS alogrithm given at the top of this section, is an elementary slice. More complicated quivers with adjoint matter that induce an EGS can then be derived from these. At the time of writing, the full list of magnetic quivers of the known elementary slices is comprised of Table $1$ in \cite{Bourget:2021siw} and Table \ref{tab:twisted affine Dds} in this paper. Thus our building block quivers $Q_a$ are found by identifying these quivers that have short rank one nodes $a_i$ each connected to a single rank two node $b_i$ via a non-simply laced edge of multiplicity $r_{a_i}$, or to two rank one nodes $b_{i_1}$ and $b_{i_2}$ both via non simply laced edges both of multiplicity $r_{a_i}$, and replacing these short nodes by rank $r_{a_i}$ nodes with an adjoint hypermultiplet, connected to $b_i$ (or $b_{i_1}$ and $b_{i_2}$) via a simply laced edge. Looking at the elementary slices, the only ones whose magnetic quivers satisfy these criteria are the affine quivers $b_n$, $c_2$ and $g_2$, and the twisted affine quivers $a_{2k}$ (for all $k\geq 1$) and $d_{k+1}$ (for $k \geq 4$).\\

\begin{table}[htbp!]
    \centering
    
  \hspace*{-2.4cm}\begin{tabular}{|c|c|c|c|} 
 
 \hline
 
 $\renewcommand{\arraystretch}{1.5}
    \begin{array}{c}
    \text{Quiver with}\\
    \text{Adjoint Matter}
    \end{array}$ 
    &
    $\renewcommand{\arraystretch}{1.5}
    \begin{array}{c}
    \text{Coulomb Branch}\\
    \text{Variety}\\
    \end{array}$ 
    &
    $\renewcommand{\arraystretch}{1.5}
    \begin{array}{c}
    \text{Global}\\
    \text{Symmetry}\\
    \end{array}$ 
    &
    PL(HWG) \\

    \hline
    
     $\raisebox{-.5\height}{\begin{tikzpicture}[x=1cm,y=.8cm]
    \node (g1) at (-1.5,1.1) [gauge,label=left:{$2$}] {};
    \node (g2) at (-1.5,0) [gauge,label=below:{$2$}] {};
    \node (g3) at (-0.5,0) {$\cdots$};
    \node (g4) at (0.5,0) [gauge,label=below:{$2$}] {};
    \node (g5) at (1.3,-0.8) [gauge,label=right:{$1$}] {};
    \node (g6) at (1.3,0.8) [gauge, label=right:{$1$}] {};
    \draw (g1)--(g2)--(g3)--(g4)--(g5);
    \draw (g4)--(g6);
    \draw (g1) to [out=45,in=135,looseness=10] (g1);
    \draw [decorate,decoration={brace,mirror,amplitude=6pt}] (-1.5,-0.8) --node[below=6pt] {$n-2$} (0.5,-0.8);
    \end{tikzpicture}}$ 
    & 
    $\renewcommand{\arraystretch}{1.5}
    \begin{array}{c}
    d_{n+1}/\mathbb{Z}_2\\
    =\\
    \text{n.min} \ B_n \\
    \end{array}$
    &
    $\renewcommand{\arraystretch}{1.5}
    \begin{array}{c}
    \textcolor{blue}{SO(2n+1)}\\
    \end{array}$
    &
    $\renewcommand{\arraystretch}{1.5}
    \begin{array}{c}
    \mu_2^2 t^2 + \mu_1^2 t^4, \ \ n=2\\
    \mu_2 t^2 + \mu_1^2 t^4, \ \ n \geq 3\\
    \end{array}$ \\
    
    \hline
    
     $\raisebox{-.5\height}{\begin{tikzpicture}[x=1cm,y=.8cm]
        \node (g0) at (-0.5,1.1) [gauge, label=left:{$3$}] {};
        \node (g1) at (-0.5,0) [gauge,label=below:{$2$}] {};
        \node (g2) at (0.5,0) [gauge, label=below:{$1$}] {};
        \draw (g0)--(g1)--(g2);
        \draw (g0) to [out=45,in=135,looseness=10] (g0);
    \end{tikzpicture}}$ 
    & 
    $\renewcommand{\arraystretch}{1.5}
    \begin{array}{c}
    d_4/S_3\\
    =\\
    \text{sub-regular} \ G_2 \\
    \end{array}$
    &
    $\renewcommand{\arraystretch}{1.5}
    \begin{array}{c}
    \textcolor{blue}{G_2}\\
    \end{array}$ 
    &
    $\renewcommand{\arraystretch}{1.5}
    \begin{array}{c}
    \mu_1 t^2 + \mu_2^2 t^4 + \mu_2^3 t^6 + \mu_1^2 t^8 +\mu_1 \mu_2^3 t^{10} - \mu_1^2 \mu_2^6 t^{20}\\
    \end{array}$\\
    
    \hline
    
     $\raisebox{-.5\height}{ \begin{tikzpicture}[x=1cm,y=.8cm]
    \node (g1) at (-1.5,1.1) [gauge,label=left:{$2$}] {};
    \node (g2) at (-1.5,0) [gauge,label=below:{$2$}] {};
    \node (g3) at (-0.5,0) {$\cdots$};
    \node (g4) at (0.5,0) [gauge,label=below:{$2$}] {};
    \node (g5) at (1.5,0) [gauge,label=below:{$1$}] {};
    \draw (g1)--(g2)--(g3)--(g4);
    \draw[transform canvas = {yshift=-1.5pt}] (g5)--(g4);
    \draw[transform canvas = {yshift=1.5pt}] (g5)--(g4);
    \draw (0.9,0.2)--(1.1,0)--(0.9,-0.2);
    \draw (g1) to [out=45,in=135,looseness=10] (g1);
    \draw [decorate,decoration={brace,mirror,amplitude=6pt}] (-1.5,-0.8) --node[below=6pt] {$n-2$} (0.5,-0.8);
    \end{tikzpicture}}$
    & 
    $\renewcommand{\arraystretch}{1.5}
    \begin{array}{c}
    b_n/\mathbb{Z}_2\\
    =\\
    \text{n.min} \ D_n \\
    \end{array}$
    &
    $\renewcommand{\arraystretch}{1.5}
    \begin{array}{c}
    \textcolor{blue}{SO(2n)}\\
    \end{array}$
    &
    $\renewcommand{\arraystretch}{1.5}
    \begin{array}{c}
    (\mu_1^2 + \mu_2^2) t^2, \ \ n=2\\
    \mu_2 \mu_3 t^2 + \mu_1^2 t^4, \ \ n=3\\
    \mu_2 t^2 + \mu_1^2 t^4, \ \ n \geq 4\\
    \end{array}$\\
    
    \hline
        
        $\raisebox{-.5\height}{\begin{tikzpicture}[x=1cm,y=.8cm]
    \node (g1) at (-1.5,1.1) [gauge,label=left:{$2$}] {};
    \node (g2) at (-1.5,0) [gauge,label=below:{$2$}] {};
    \node (g3) at (-0.5,0) {$\cdots$};
    \node (g4) at (0.5,0) [gauge,label=below:{$2$}] {};
    \node (g5) at (0.5,1.1) [gauge,label=left:{$2$}] {};
    \draw (g1)--(g2)--(g3)--(g4)--(g5);
    \draw (g1) to [out=45,in=135,looseness=10] (g1);
    \draw (g5) to [out=45,in=135,looseness=10] (g5);
    \draw [decorate,decoration={brace,mirror,amplitude=6pt}] (-1.5,-0.8) --node[below=6pt] {$n-2$} (0.5,-0.8);
    \end{tikzpicture}}$ 
    & 
    $d_{n+1}/(\mathbb{Z}_2 \times \mathbb{Z}_2)$
    &
    $\renewcommand{\arraystretch}{1.5}
    \begin{array}{c}
    \textcolor{blue}{SO(2n)}\\
    \end{array}$ 
    &
    $\renewcommand{\arraystretch}{1.5}
    \begin{array}{c}
    (\mu_1^2 +\mu_2^2)t^2 + (1+\mu_1^2 \mu_2^2)t^4 +\mu_1^2 \mu_2^2 t^6 - \mu_1^4 \mu_2^4 t^{12}, \ n=2\\
    \mu_2 \mu_3 t^2 + (1+2\mu_1^2)t^4+\mu_1^2t^6-\mu_1^4t^{12}, \ n=3\\
    \mu_2 t^2 + (1+2\mu_1^2)t^4 +\mu_1^2t^6-\mu_1^4 t^{12}, \ n\geq 4
    \end{array}$ \\
    \hline
    
    $\raisebox{-.5\height}{\begin{tikzpicture}[x=1cm,y=.8cm]
        \node (g1) at (0,1.1) [gauge,label=left:{$4$}] {};
        \node (g2) at (0,0) [gauge,label=below:{$2$}] {};
        \draw (g1)--(g2);
        \draw (g1) to [out=45,in=135,looseness=10] (g1);
        \end{tikzpicture}}$ 
    &
    $d_4/S_4$
    &
    $\renewcommand{\arraystretch}{1.5}
    \begin{array}{c}
    \textcolor{blue}{SU(3)}\\
    \end{array}$
    &
     $\renewcommand{\arraystretch}{1.5}
    \begin{array}{c}
    \mu_1 \mu_2t^2 + (1+\mu_1^2 + \mu_1 \mu_2 + \mu_2^2)t^4 +\\
    (1+\mu_1^3+\mu_2^2+\mu_1\mu_2^2 + \mu_2^3 + \mu_1^2 + \mu_1^2 \mu_2)t^6 + \mathcal{O}(t^{8})\\
    \end{array}$\\
    
    \hline
    
     $\raisebox{-.5\height}{ \begin{tikzpicture}[x=1cm,y=.8cm]
    \node (g1) at (-1.5,1.1) [gauge,label=left:{$2$}] {};
    \node (g2) at (-1.5,0) [gauge,label=below:{$2$}] {};
    \node (g3) at (-0.5,0) {$\cdots$};
    \node (g4) at (0.5,0) [gauge,label=below:{$2$}] {};
    \node (g5) at (1.5,0) [gauge,label=below:{$2$}] {};
    \draw (g1)--(g2)--(g3)--(g4);
    \draw[transform canvas = {yshift=-1.5pt}] (g5)--(g4);
    \draw[transform canvas = {yshift=1.5pt}] (g5)--(g4);
    \draw (1.1,0.2)--(0.9,0)--(1.1,-0.2);
    \draw (g1) to [out=45,in=135,looseness=10] (g1);
    \draw [decorate,decoration={brace,mirror,amplitude=6pt}] (-1.5,-0.8) --node[below=6pt] {$n-2$} (0.5,-0.8);
    \end{tikzpicture}}$
    &
    $a_{2n-1}/\mathbb{Z}_2$
    &
    $\renewcommand{\arraystretch}{1.5}
    \begin{array}{c}
    \textcolor{blue}{SU(2n-1)}\\
    \textcolor{blue}{\times} \\
    \textcolor{blue}{U(1)}\\
    \end{array}$ 
    &
    $\renewcommand{\arraystretch}{1.5}
    \begin{array}{c}
    (1+\mu_1 \mu_{2n-2})t^2 + (q^{4n} \mu_1^2 +q^{-4n} \mu_{2n-2}^2)t^4 - \mu_1^2 \mu_{2n-2}^2 t^8\\
    \end{array}$ \\
    \hline

\end{tabular}

 \caption{The full list of balanced unframed unitary $3d$ $\mathcal{N}=4$ quivers with $l \in \mathbb{N}\ \{0\}$ nodes that each have a single adjoint hypermultiplet but which otherwise contain only regular matter. Upon performing quiver addition to remove the adjoint hypermultiplet(s), the resulting quiver will elicit an enhanced global symmetry. For all quivers that depend on the parameter $n$ the quiver is valid for $n \geq 2$. N.min is shorthand for the (closure of the) next to minimal nilpotent orbit of the corresponding algebra. The $\mu_i, \, q$ in the PL(HWG) column are the Dynkin label fugacities of the corresponding Lie group and Abelian factors of the global symmetry group respectively.}
    \label{tab:adj hyper Qs}
\end{table}

\clearpage

\noindent There are several things worth noting regarding the quivers $Q_a$ in Table \ref{tab:adj hyper Qs}:
\begin{itemize}
    \item The quivers are ordered in such a way to best aid an illustration of the results in this paper. In particular the first quiver yields the best illustration of the most basic quiver addition needed, and the final quiver will not appear at all as it is unclear if it is possible to use quiver addition to absorb the adjoint hypermultiplet on this node.
    
    \item We have been unable to find the fully refined Hilbert series, and therefore HWG, in terms of the global symmetry fugacities for the $b_n/\mathbb{Z}_2$, $d_{n+1}/(\mathbb{Z}_2 \times \mathbb{Z}_2)$, $d_4/S_4$ and $a_{2n-1}/\mathbb{Z}_2$ quivers in Table \ref{tab:adj hyper Qs}. The reason is that in these cases, there are not sufficient topological fugacities to apply the normal tricks we use to find fugacity maps (such maps are discussed in more detail in Appendix \ref{app:fugmap}). However the \textit{partially refined} Hilbert series can be found, which is the Hilbert series in terms of the characters of some subgroup of the global symmetry. These partially refined Hilbert series have been computed for all the quivers under discussion in this bullet point, but they are not listed because we were able to find the fully refined HWG for the quiver theories in question via other methods. We discuss these in turn, labelling the quiver we're addressing in each case by its Coulomb branch variety:
        \begin{itemize}
        
            \item $b_n/\mathbb{Z}_2$: Here the fully refined HWG can be obtained by noticing that the variety of $b_n/\mathbb{Z}_2$ is the next to minimal nilpotent orbit of $D_n$ \cite{1992math......4227B}, and so the Coulomb branch HWG can be computed by applying the monopole formula to the magnetic quiver for this variety, which is given in \cite{Hanany:2016gbz}.
        
            \item $d_{n+1}/(\mathbb{Z}_2 \times \mathbb{Z}_2)$: Here we have found the HWG in two different ways. Firstly, we use conjecture (\ref{bouquet to adjoint conjecture}) to realise the Coulomb branch of this quiver as the $\mathbb{Z}_2$ projection of the next to minimal nilpotent orbit of $B_n$. The fully refined HWG for n.min $B_n$ can be computed exactly by applying the monopole formula to the top quiver in Table \ref{tab:adj hyper Qs}. The $\mathbb{Z}_2$ projection of this can then be computed: see Appendix \ref{app:discrete projection} for an example of how to perform such a projection on a HWG. We then $(a)$ unrefined and $(b)$ partially refined the HWG that resulted from this projection and confirmed that it agreed exactly with the unrefined and partially refined Hilbert series for the magnetic quiver for $d_{n+1}/(\mathbb{Z}_2 \times \mathbb{Z}_2)$ that is given in Table \ref{tab:adj hyper Qs}. The second way to check this HWG was to compute the Higgs branch of the electric quiver for $d_{n+1}/(\mathbb{Z}_2 \times \mathbb{Z}_2)$, given in Figure 17 of \cite{Hanany:2018dvd}. This has been done for several $n$, and the result matches that obtained by the first method, providing as much evidence for the validity of this result as is possible using the monopole formula.
        
            \item $d_4/S_4$: Again using conjecture (\ref{bouquet to adjoint conjecture}), the HWG for this variety can be found by applying an $S_4$ projection to the minimal nilpotent orbit of $D_4$. This has been done by Siyul Lee and does indeed have a closed form, but it is rather lengthy and not particularly illuminating to give here, so we list it to order $6$. Again, the result of (\ref{bouquet to adjoint conjecture}) in this case has been checked by unrefining and partially refining the HWG computed for $d_4 / S_4$ and comparing with the unrefined and partially refined Hilbert series for the magnetic quiver for $d_4/S_4$ that is given in Table \ref{tab:adj hyper Qs}.
        
            \item $a_{2n-1}/\mathbb{Z}_2$: Here (\ref{bouquet to adjoint conjecture}) tells us that this quiver is the $\mathbb{Z}_2$ quotient of the twisted affine Dynkin diagram for $A_{2n-1}$ (see Table \ref{tab:twisted affine Dds}). This cannot be fully refined however, and so we exploit the fact that the twisted affine and affine quiver for the same algebra both give the minimal nilpotent orbit of this algebra, and so we can instead use the HWG computed from the affine quiver (which does fully refine) to perform the $\mathbb{Z}_2$ projection on. When unrefined, this result matches the unrefined Hilbert series for the magnetic quiver for $a_{2n-1}/\mathbb{Z}_2$ listed in Table \ref{tab:adj hyper Qs}.
            
        \end{itemize}

     \item This is not an exhaustive list of quivers with adjoint matter that induce an EGS, but it is an exhaustive list of all \textit{balanced} quivers with adjoint matter that induce EGS, and as such is rich enough to provide us with many examples for which the BGS algorithm of Section \ref{sec:BGS} fails.
\end{itemize}
We now move on to discussing how to use quiver addition on the quivers in Table \ref{tab:adj hyper Qs} to construct quivers with only regular matter $Q_r$ (remember that among other things this means with\textit{out} adjoint matter) whose global symmetry is enhanced from that predicted by the BGS.

\subsection{Quiver Addition} \label{sec: quiver addition}

Recall that our goal is to find all possible quivers whose Hasse diagram\footnote{In most cases that we will see in this paper, the Hasse diagram we find will only be a conjecture. This is because, as mentioned in the introduction, since we are unaware of a brane system, we have no way to physically motivate this process of quiver subtraction or addition.} has as a lowest leaf the Coulomb branch of one of the quivers in Table \ref{tab:adj hyper Qs}, because this leads to an enhancement of the BGS predicted by the algorithm outlined in Section \ref{sec:BGS}. In order to do this, quiver subtraction must be reverse engineered to ``absorb" the nodes with adjoint hypermultiplets appearing in Table \ref{sec:BGS}. One notion of quiver addition is discussed in \cite{Bourget:2021siw}, but here we will develop a new algorithm for the case when we add to absorb nodes with adjoint hypermultiplets from a quiver. This will be done based on the quiver subtraction rule conjectured in Appendix C of \cite{Bourget:2020mez}. For a review on this rule, and all other necessary information on quiver subtraction for the present work, see Appendix \ref{app:quiver subtraction}. To illustrate how the process of quiver addition was developed by reverse engineering quiver subtraction, we use the example of the quiver $Q$ of (\ref{double e6 RHS example quiver}), which we recall was used to highlight the inadequacy of the BGS algorithm in Section \ref{sec:BGS}. Performing double $e_6$ quiver subtraction\footnote{Generally, when we say a $g_k$ quiver subtraction, we mean the act of subtracting the magnetic quiver for the $g_k$ variety (the minimal nilpotent orbit of $G_k$). In this particular example, we are taking $G_k=E_6$.} on $Q$ gives

\clearpage

\begin{equation} \label{double e6 subtraction eg}
    \raisebox{-.5\height}{\begin{tikzpicture}[x=1cm,y=.8cm]
    \node (g2) at (-2,0) [gauge,label=below:{$2$},label={[red]left:{$b$}}] {};
    \node (g3) at (-1,0) [gauge,label=below:{$4$}] {};
    \node (g4) at (0,0) [gauge,label=below:{$6$}] {};
    \node (g5) at (1,0) [gauge,label=below:{$8$}] {};
    \node (g6) at (2,0) [gauge, label=below:{$5$}] {};
    \node (g7) at (1,1) [gauge, label=left:{$5$}] {};
    \node (g8) at (1,2) [gauger, label=left:{$2$}] {};
    \node (g9) at (3,0) [gauger, label=below:{$2$}] {};
    \draw (g2)--(g3)--(g4)--(g5)--(g6)--(g9);
    \draw (g5)--(g7)--(g8);
\begin{scope}[shift={(0,-3.5)}]
\node at (-4,0) {$-$};
\node (g2) at (-1,0) [gauge,label=below:{$1$}] {};
\node (g3) at (0,0) [gauge,label=below:{$2$}] {};
\node (g4) at (1,0) [gauge,label=below:{$3$}] {};
\node (g5) at (2,0) [gauge,label=below:{$2$}] {};
\node (g6) at (3,0) [gauge,label=below:{$1$}] {};
\node (g7) at (1,1) [gauge,label=left:{$2$}] {};
\node (g8) at (1,2) [gauge,label=left:{$1$}] {};
\draw (g2)--(g3)--(g4)--(g5)--(g6);
\draw (g4)--(g7)--(g8);
\end{scope}
\begin{scope}[shift={(0,-7)}]
    \node (g1) at (-2,1) [gaugeb,label=left:{$1$}] {};
    \node (g2) at (-2,0) [gauge,label=below:{$2$},label={[red]left:{$b$}}] {};
    \node (g3) at (-1,0) [gauge,label=below:{$3$}] {};
    \node (g4) at (0,0) [gauge,label=below:{$4$}] {};
    \node (g5) at (1,0) [gauge,label=below:{$5$}] {};
    \node (g6) at (2,0) [gauge, label=below:{$3$}] {};
    \node (g7) at (1,1) [gauge, label=left:{$3$}] {};
    \node (g8) at (1,2) [gauger, label=left:{$1$}] {};
    \node (g9) at (3,0) [gauger, label=below:{$1$}] {};
    \draw (g1)--(g2)--(g3)--(g4)--(g5)--(g6)--(g9);
    \draw (g5)--(g7)--(g8);
\end{scope}
\begin{scope}[shift={(0,-10.5)}]
    \node at (-4,0) {$-$};
    \node (g2) at (-1,0) [gauge,label=below:{$1$}] {};
    \node (g3) at (0,0) [gauge,label=below:{$2$}] {};
    \node (g4) at (1,0) [gauge,label=below:{$3$}] {};
    \node (g5) at (2,0) [gauge,label=below:{$2$}] {};
    \node (g6) at (3,0) [gauge,label=below:{$1$}] {};
    \node (g7) at (1,1) [gauge,label=left:{$2$}] {};
    \node (g8) at (1,2) [gauge,label=left:{$1$}] {};
    \draw (g2)--(g3)--(g4)--(g5)--(g6);
    \draw (g4)--(g7)--(g8);
\end{scope}
\begin{scope}[shift={(0,-14)}]
    \node (g0) at (-2,1.1) [gaugeb,label=right:{$2$},label={[red]left:{$a$}}] {};
    \node (g1) at (-2,0) [gauge,label=below:{$2$},label={[red]left:{$b$}}] {};
    \node (g2) at (-1,0) [gauge,label=below:{$2$}] {};
    \node (g4) at (0,0) [gauge,label=below:{$2$}] {};
    \node (g5) at (1,0) [gauge,label=below:{$2$}] {};
    \node (g6) at (2,0) [gauge,label=below:{$1$}] {};
    \node (g7) at (1,1) [gauge,label=left:{$1$}] {};
    \draw (g0)--(g1)--(g2)--(g4)--(g5)--(g6);
    \draw (g5)--(g7);
    \draw (g0) to [out=45,in=135,looseness=10] (g0);
\end{scope}
    \draw[->] (4,0) to [out=300,in=60,looseness=1] (4,-6.8);
    \draw[->] (4,-7.2) to [out=300,in=60,looseness=1] (4,-13.8);
\end{tikzpicture}}
\end{equation}
where blue nodes indicate those that were introduced in the rebalancing stage of quiver subtraction; red nodes indicate unbalanced nodes; $b$ is the node being rebalanced as a result of the subtractions (the reason for this labelling is to make contact with the notation we introduce shortly in (\ref{h,b,a notation quiver})); and as before $a$ is the name of the node with the adjoint hypermultiplet (that has arisen through multiple same slice subtractions). From Table \ref{tab:adj hyper Qs}, we see that this quiver has the next to minimal nilpotent orbit of $B_6$ as a leaf in its Hasse diagram. The Hasse diagram of this orbit is
\begin{equation} \label{HD nminB6}
    \begin{tikzpicture}
    \node (1) [hasse] at (0,0) {};
    \node (2) [hasse] at (0,-1) {};
    \node (3) [hasse] at (0,-2) {};
    \node (4) at (0.5,-1) {$,$};
    \draw (1) edge [] node[label=left:$A_1$] {} (2);
    \draw (2) edge [] node[label=left:$b_6$] {} (3);
    \end{tikzpicture}
\end{equation}
which can be derived using quiver subtraction on the magnetic quiver for the next to minimal nilpotent orbit of $B_6$, which can be found in \cite{Hanany:2016gbz}. As explained in Section \ref{sec:BGS}, the BGS of $Q$ is $SO(12) \times U(1)$, but the Hilbert series shows the true global symmetry, the EGS, to be $SO(13) \times U(1)$. The quiver subtraction above illustrates why this enhancement occurs: the Hasse diagram of $Q$ contains $b_6$ among its bottom-most elementary slices, and so this must be a subset of the global symmetry (as described in Section \ref{sec: HDs}). The lowest slice is $b_6$ and not $d_6$ (which is what was expected from the BGS) because of the appearance of the adjoint hypermultiplet in quiver subtraction as opposed to an additional rebalancing $U(1)$ node. \\

\noindent Note that although in this case these $e_6$ subtractions are the only possible ones that could be performed, in other quivers that we find in this work multiple subtractions will be possible, leading to bifurcations in the Hasse diagram. However, it is only subtracting the same slice twice in a row that leads to the arisal of a node with an adjoint hypermultiplet, and thus the enhancement of a factor of the global symmetry, and so it is only the ``same slice twice" subtractions and additions we focus on, as the other factors can be read accurately from the BGS algorithm.\\

\noindent From this example, we can see that the quiver subtraction algorithm which brings about these adjoint hypermultiplets can be reverse engineered to find quivers with only regular matter which have an EGS. In order to explain how this works, we will refer to the node with the adjoint hypermultiplet as $a$ of rank $r_a$, the node adjacent to it $b$ of rank $r_b$, and any generic node adjacent to \textit{this} $c$ of rank $r_c$:
\begin{equation} \label{h,b,a notation quiver}
    \begin{tikzpicture}[x=1cm,y=.8cm]
    \node (g-1) at (0,-1) {$\textcolor{red}{b}$};
    \node (g-2) at (1,-1) {$\textcolor{red}{c}$};
    \node (g1) at (-1,0) {$\cdots$};
    \node (g2) at (0,0) [gauge,label=below:{$r_b$}] {};
    \node (g3) at (0,1.1) [gauge,label=left:{$r_a$},label={[red]right:{$a$}}] {};
    \node (g4) at (1,0) [gauge,label=below:{$r_c$}] {};
    \node (g5) at (2,0) {$\cdots$};
    \draw (g1)--(g2)--(g4)--(g5);
    \draw (g2)--(g3);
    \draw (g3) to [out=45,in=135,looseness=10] (g3);
    \end{tikzpicture}.
\end{equation}
These will be the names of the nodes \textit{after} quiver additions also: note in particular that this means that after the penultimate quiver addition, $a$ will not have an adjoint hypermultiplet attached, and after the final addition $a$ will no longer exist, as explained below.\\

\noindent In the ``forward" process of quiver subtraction, an example of which was shown in (\ref{double e6 subtraction eg}), the adjoint hypermultiplet arose as a result of rebalancing node $b$ after multiple subsequent same slice subtractions. Thus to ``absorb" the adjoint hypermultiplet in the reverse process, we will have to add the same elementary slice mutliple subsequent times to node $c$ such that node $b$ becomes overbalanced, and thus the rebalancing process will involve removing a rank from $a$. With all this in mind, we are ready to construct the algorithm for quiver addition.

\level{4}*{Quiver addition algorithm}
\label{Alg: Quiver addition alg}
Let $Q_a$ be a quiver with a single node $a$ that has a single adjoint hypermultiplet, such that $a$ is connected to the remainder of the quiver via a single simply laced edge to node $b$, and with all other nodes being linked by regular matter\footnoteref{footnote:reg matter} only. Let $Q_\sigma$ be the magnetic quiver for a generic elementary slice $\sigma$.\footnote{Recall that the full list of these to date can be found by comprising Table 1 of \cite{Bourget:2021siw} and Table $2$ of this note.} Then adding some $Q_\sigma$ to $Q_a$, \textit{that is going to be added multiple subsequent times} can be performed as follows:
\begin{enumerate}
    \item First ensure that $Q_\sigma$ is balanced and contains as a subset the run of nodes connecting to and including $c$ but excluding $a$ and $b$ (i.e. $c$ and the ``$\cdots$" next to it in (\ref{h,b,a notation quiver})) of $Q_a$.
    \item Line up $Q_\sigma$ with $Q_a$ so that one of the long rank $1$ nodes of $Q_\sigma$ is superimposed upon $c$, and the subset of the rest of the nodes of $Q_\sigma$ which is of the same form as the nodes connecting to and including $c$ of $Q_a$ (as described in Step $1$) is aligned with these nodes. 
    \item Add the ranks of the nodes in $Q_\sigma$ to those in $Q_a$ that they line up with.
    \item Reduce the rank of $a$ by one. Note that if this reduces the rank of $a$ to one, the adjoint hypermultiplet can be eliminated as the adjoint representation of $U(1)$ is trivial, and if it reduces the rank of $a$ to zero, $a$ itself can just be eliminated as it is now an ``empty node".
\end{enumerate}

\vspace{3mm}

\noindent There are a few things to note here with regards to this algorithm:
\begin{itemize} 

\item Firstly, the condition on $Q_\sigma$ of balance and needing a long rank one node restricts the possible slices we can add to be just the affine quivers displayed in Table \ref{tab:affine Dds}.

\item Secondly, $c$ can be an existing node, as pictured in (\ref{h,b,a notation quiver}), or it could be an ``empty node". By taking $c$ as an empty node, we mean that the quiver being added is superimposed on top of currently non-existent nodes, but such that it is linked to node $b$. If we take $c$ to be an existing node, we call such an addition \textit{adding to existing nodes}. If we take $c$ to be an empty node, we call such an addition \textit{adding to empty nodes}. We call all possible options for the node $c$ the \textit{c-nodes} of the quiver. 

\item Thirdly, during such a process of quiver addition, extra unbalanced nodes that were not in the original quiver may appear, but this is not a problem as long as throughout both additions no nodes undergo a \textit{change} in balance.

\item Fourthly, this algorithm is modified for the $d_3/\mathbb{Z}_2$, $b_2/\mathbb{Z}_2$ and $d_{n+1}/(\mathbb{Z}_2 \times \mathbb{Z}_2)$ quivers in Table \ref{tab:adj hyper Qs}. In the $d_3/\mathbb{Z}_2$ case this is because we have two possible ``$b$-nodes" which each need a $c$-node, and so, among other modifications, any added slice must have \textit{two} rank one long nodes. In the $b_n/\mathbb{Z}_2$ case it is because $a$ is attached to $b$ via a non-simply laced edge of multiplicity two. In the $d_{n+1}/(\mathbb{Z}_2 \times \mathbb{Z}_2)$ case it is due to the fact that there are two nodes with adjoint hypermultiplets, and so we need to be careful where we can add. These modifications will be explained within the relevant sections of the paper: Section \ref{sec:dn to bn}, Section \ref{sec: bn-1 to dn} and Section \ref{sec: sun u1 to so2n} respectively.

\item Finally, this is only the quiver addition process for the specific purpose outlined in this paper, and not a general algorithm. When not adding the same slice twice to the same node, there are often many possible ways to add a slice to a quiver. In Section \ref{A type to B type section} ``single additions" (just adding one slice once) will be performed, so an example of the general algorithm for this type of addition will be illustrated here.

\end{itemize}

\vspace{2mm}

\noindent For complete clarity regarding the execution of this algorithm, we provide here an example of an \textit{incorrect} addition (i.e. adding a slice $\sigma$ to a quiver $Q_1$ to give a quiver $Q_2$ from which performing the reverse processs of subtraction does not give back $Q_1$: $Q_2 - \sigma \neq Q_1$), followed by the correct way to execute it.

\paragraph{Example} Consider trying to add $d_4$ to the \textit{existing node} in n.min $B_5$.\footnote{To clarify the terminology one last time, here what we are actually saying is ``consider trying to add the quiver whose Coulomb branch is the closure of the minimal nilpotent orbit of $D_4$, which we call $d_4$, to the existing node of the quiver whose Coulomb branch is the closure of the next to minimal nilpotent orbit of $B_5$, which we call n.min $B_5$.} The magnetic quiver for n.min $B_5$ is
\begin{equation} \label{n.min B5}
    \begin{tikzpicture}[x=1cm,y=.8cm]
    \node (g-1) at (-1.5,-1) {$\textcolor{red}{b}$};
    \node (g-2) at (-0.5,-1.05) {$\textcolor{red}{c}$};
    \node (g2) at (-1.5,0) [gauge,label=below:{$2$}] {};
    \node (g3) at (-1.5,1.1) [gaugeb,label=left:{$2$},label={[red]right:{$a$}}] {};
    \node (g4) at (-0.5,0) [gauge,label=below:{$2$}] {};
    \node (g5) at (0.5,0) [gauge,label=below:{$2$}] {};
    \node (g6) at (1.5,0) [gauge,label=below:{$1$}] {};
    \node (g7) at (0.5,1) [gauge,label=left:{$1$}] {};
    \node (g9) at (2,0) {$,$};
    \draw (g3)--(g2)--(g4)--(g5)--(g6);
    \draw (g5)--(g7);
    \draw (g3) to [out=45,in=135,looseness=10] (g3);
    \end{tikzpicture}
\end{equation}
where $c$ has been labelled as such as it is the only possible existing node we could add to.\footnote{Note that we could have taken $c$ to be the empty nodes: the space to the left of $b$ in (\ref{n.min B5}), but for this case we wouldn't have been able to demonstrate this type of incorrect addition.} This addition fits in with Step $1$ in the algorithm, as $d_4$ contains as a subset the nodes to the right of and including $c$ in (\ref{n.min B5}). So adding this slice \textit{could} be valid, but we will perform the addition in a way that violates step $2$:
\clearpage

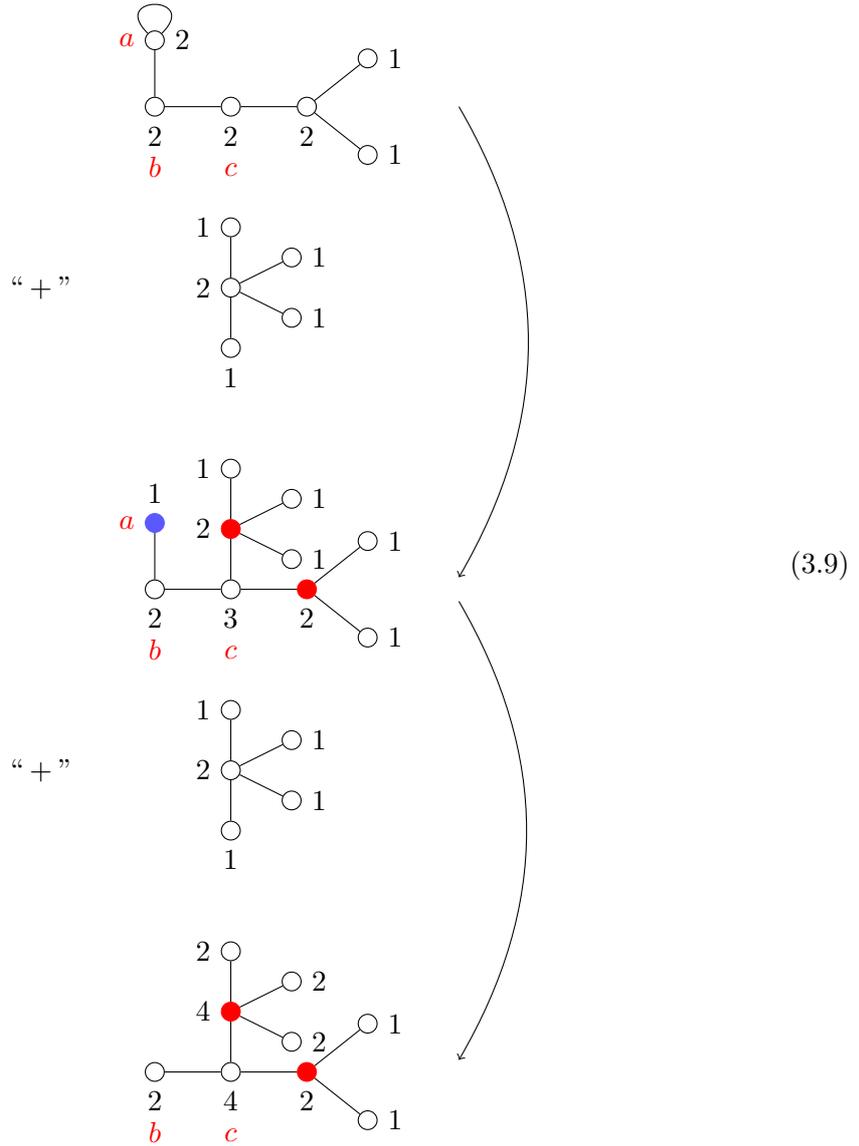
\begin{equation}
\raisebox{-.5\height}{\begin{tikzpicture}[x=1cm,y=.8cm]
    \node (g1) at (-1.5,1.1) [gauge,label=right:{$2$},label={[red]left:{$a$}}] {};
    \node (g2) at (-1.5,0) [gauge,label=below:{$2$}] {};
    \node (g2a) at (-1.5,-1) {$\textcolor{red}{b}$};
    \node (g3) at (-0.5,0) [gauge,label=below:{$2$}] {};
    \node (g3a) at (-0.5,-1.05) {$\textcolor{red}{c}$};
    \node (g4) at (0.5,0) [gauge,label=below:{$2$}] {};
    \node (g5) at (1.3,0.8) [gauge,label=right:{$1$}] {};
    \node (g6) at (1.3,-0.8) [gauge, label=right:{$1$}] {};
    \draw (g1)--(g2)--(g3)--(g4)--(g5);
    \draw (g4)--(g6);
    \draw (g1) to [out=45,in=135,looseness=10] (g1);

\begin{scope}[shift={(0,-4)}]
    \node at (-3,1) {$``+"$};
    \node (g2) at (-0.5,0) [gauge,label=below:{$1$}] {};
    \node (g3) at (-0.5,1) [gauge,label=left:{$2$}] {};
    \node (g4) at (-0.5,2) [gauge,label=left:{$1$}] {};
    \node (g5) at (0.3,1.5) [gauge,label=right:{$1$}] {};
    \node (g6) at (0.3,0.5) [gauge,label=right:{$1$}] {};
    \draw (g2)--(g3)--(g4);
    \draw (g5)--(g3);
    \draw (g6)--(g3);
\end{scope}

\begin{scope}[shift={(0,-8)}]
    \node (g1) at (-1.5,1.1) [gaugeb,label=above:{$1$},label={[red]left:{$a$}}] {};
    \node (g2) at (-1.5,0) [gauge,label=below:{$2$}] {};
    \node (g2a) at (-1.5,-1) {$\textcolor{red}{b}$};
    \node (g3) at (-0.5,0) [gauge,label=below:{$3$}] {};
    \node (g3a) at (-0.5,-1.05) {$\textcolor{red}{c}$};
    \node (g4) at (0.5,0) [gauger,label=below:{$2$}] {};
    \node (g5) at (1.3,0.8) [gauge,label=right:{$1$}] {};
    \node (g6) at (1.3,-0.8) [gauge, label=right:{$1$}] {};
    \node (g8) at (-0.5,1) [gauger,label=left:{$2$}] {};
    \node (g9) at (-0.5,2) [gauge,label=left:{$1$}] {};
    \node (g10) at (0.3,1.5) [gauge,label=right:{$1$}] {};
    \node (g11) at (0.3,0.5) [gauge,label=right:{$1$}] {};
    \draw (g3)--(g8)--(g9);
    \draw (g10)--(g8);
    \draw (g11)--(g8);
    \draw (g1)--(g2)--(g3)--(g4)--(g5);
    \draw (g4)--(g6);
\end{scope}

\begin{scope}[shift={(0,-12)}]
   \node at (-3,1) {$``+"$};
    \node (g2) at (-0.5,0) [gauge,label=below:{$1$}] {};
    \node (g3) at (-0.5,1) [gauge,label=left:{$2$}] {};
    \node (g4) at (-0.5,2) [gauge,label=left:{$1$}] {};
    \node (g5) at (0.3,1.5) [gauge,label=right:{$1$}] {};
    \node (g6) at (0.3,0.5) [gauge,label=right:{$1$}] {};
    \draw (g2)--(g3)--(g4);
    \draw (g5)--(g3);
    \draw (g6)--(g3);
\end{scope}

\begin{scope}[shift={(0,-16)}]
    \node (g2) at (-1.5,0) [gauge,label=below:{$2$}] {};
    \node (g2a) at (-1.5,-1) {$\textcolor{red}{b}$};
    \node (g3) at (-0.5,0) [gauge,label=below:{$4$}] {};
    \node (g3a) at (-0.5,-1.05) {$\textcolor{red}{c}$};
    \node (g4) at (0.5,0) [gauger,label=below:{$2$}] {};
    \node (g5) at (1.3,0.8) [gauge,label=right:{$1$}] {};
    \node (g6) at (1.3,-0.8) [gauge, label=right:{$1$}] {};
    \node (g8) at (-0.5,1) [gauger,label=left:{$4$}] {};
    \node (g9) at (-0.5,2) [gauge,label=left:{$2$}] {};
    \node (g10) at (0.3,1.5) [gauge,label=right:{$2$}] {};
    \node (g11) at (0.3,0.5) [gauge,label=right:{$2$}] {};
    \draw (g3)--(g8)--(g9);
    \draw (g10)--(g8);
    \draw (g11)--(g8);
    \draw (g2)--(g3)--(g4)--(g5);
    \draw (g4)--(g6);
\end{scope}

    \draw[->] (2.5,0) to [out=300,in=60,looseness=1] (2.5,-7.8);
    \draw[->] (2.5,-8.2) to [out=300,in=60,looseness=1] (2.5,-15.8);
\end{tikzpicture}}
\end{equation}
Here the $d_4$ slice that was added was not superimposed correctly on top of the nodes next to and including $c$ in n.min $B_5$ that it ``matched to" (i.e. that formed a subset of the $d_4$'s nodes). We can immediately see why this is an incorrect addition, because the final rank two node in n.min $B_5$ has changed balance during this process, and we know quiver subtraction always preserves the balance of nodes. Thus if we subtract $d_4$ from our result here, we see that both $a$ and this node undergo a change in balance, and so the rank one node added to rebalance must be attached to both. When the subtraction is performed again, this node will then become a rank two node with an adjoint hypermultiplet, and we will not have the quiver n.min $B_5$, but instead
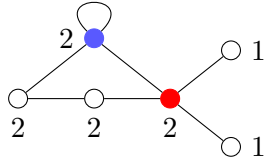
\begin{equation}
    \begin{tikzpicture}[x=1cm,y=.8cm]
    \node (g1) at (-0.5,1) [gaugeb,label=left:{$2$}] {};
    \node (g2) at (-1.5,0) [gauge,label=below:{$2$}] {};
    \node (g3) at (-0.5,0) [gauge,label=below:{$2$}] {};
    \node (g4) at (0.5,0) [gauger,label=below:{$2$}] {};
    \node (g5) at (1.3,-0.8) [gauge,label=right:{$1$}] {};
    \node (g6) at (1.3,0.8) [gauge, label=right:{$1$}] {};
    \draw (g1)--(g2)--(g3)--(g4)--(g1);
    \draw (g4)--(g5);
    \draw (g4)--(g6);
    \draw (g1) to [out=45,in=135,looseness=10] (g1);
    \end{tikzpicture}
\end{equation}
and so the subtraction is invalid. The \textit{correct} way to add $d_4$ on to n.min $B_5$ would be as follows:

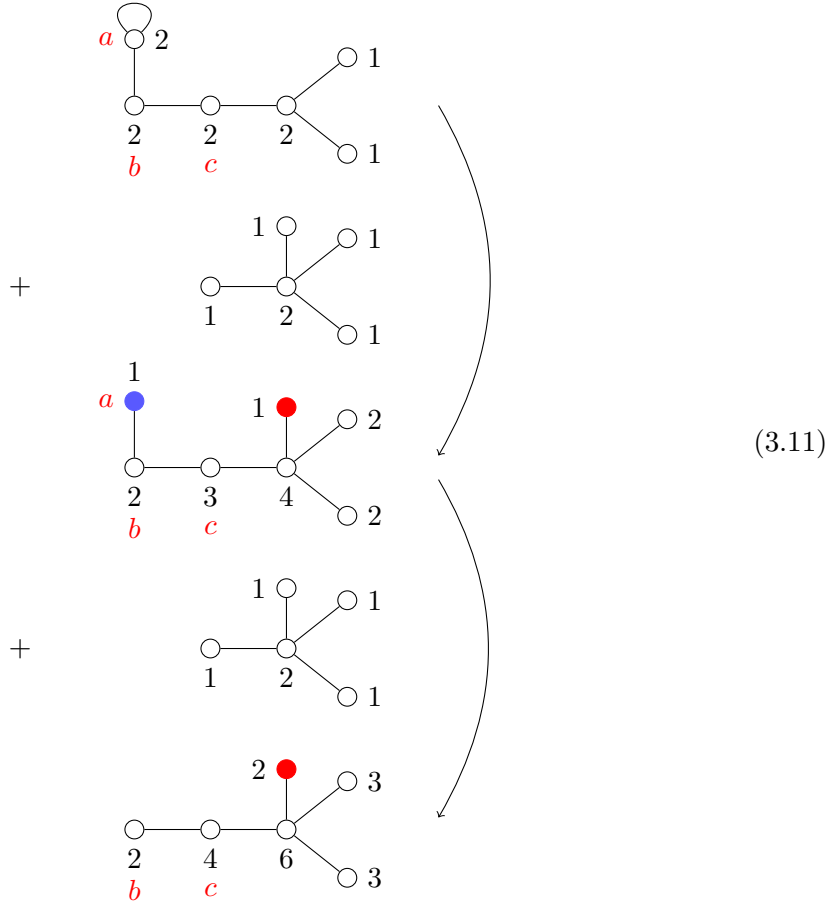
\begin{equation} \label{correct d4 addition to nminB5}
\raisebox{-.5\height}{\begin{tikzpicture}[x=1cm,y=.8cm]
    \node (g1) at (-1.5,1.1) [gauge,label=right:{$2$},label={[red]left:{$a$}}] {};
    \node (g2) at (-1.5,0) [gauge,label=below:{$2$}] {};
    \node (g2a) at (-1.5,-1) {$\textcolor{red}{b}$};
    \node (g3) at (-0.5,0) [gauge,label=below:{$2$}] {};
    \node (g3a) at (-0.5,-1) {$\textcolor{red}{c}$};
    \node (g4) at (0.5,0) [gauge,label=below:{$2$}] {};
    \node (g5) at (1.3,0.8) [gauge,label=right:{$1$}] {};
    \node (g6) at (1.3,-0.8) [gauge, label=right:{$1$}] {};
    \draw (g1)--(g2)--(g3)--(g4)--(g5);
    \draw (g4)--(g6);
    \draw (g1) to [out=45,in=135,looseness=10] (g1);

\begin{scope}[shift={(0,-3)}]
    \node at (-3,0) {$+$};
    \node (g2) at (-0.5,0) [gauge,label=below:{$1$}] {};
    \node (g3) at (0.5,0) [gauge,label=below:{$2$}] {};
    \node (g4) at (0.5,1) [gauge,label=left:{$1$}] {};
    \node (g5) at (1.3,0.8) [gauge,label=right:{$1$}] {};
    \node (g6) at (1.3,-0.8) [gauge,label=right:{$1$}] {};
    \draw (g2)--(g3)--(g4);
    \draw (g5)--(g3);
    \draw (g6)--(g3);
\end{scope}

\begin{scope}[shift={(0,-6)}]
    \node (g1) at (-1.5,1.1) [gaugeb,label=above:{$1$},label={[red]left:{$a$}}] {};
    \node (g2) at (-1.5,0) [gauge,label=below:{$2$}] {};
    \node (g2a) at (-1.5,-1) {$\textcolor{red}{b}$};
    \node (g3) at (-0.5,0) [gauge,label=below:{$3$}] {};
    \node (g3a) at (-0.5,-1) {$\textcolor{red}{c}$};
    \node (g4) at (0.5,0) [gauge,label=below:{$4$}] {};
    \node (g5) at (1.3,0.8) [gauge,label=right:{$2$}] {};
    \node (g6) at (1.3,-0.8) [gauge, label=right:{$2$}] {};
    \node (g8) at (0.5,1) [gauger,label=left:{$1$}] {};
    \draw (g4)--(g8);
    \draw (g1)--(g2)--(g3)--(g4)--(g5);
    \draw (g4)--(g6);
\end{scope}

\begin{scope}[shift={(0,-9)}]
   \node at (-3,0) {$+$};
    \node (g2) at (-0.5,0) [gauge,label=below:{$1$}] {};
    \node (g3) at (0.5,0) [gauge,label=below:{$2$}] {};
    \node (g4) at (0.5,1) [gauge,label=left:{$1$}] {};
    \node (g5) at (1.3,0.8) [gauge,label=right:{$1$}] {};
    \node (g6) at (1.3,-0.8) [gauge,label=right:{$1$}] {};
    \draw (g2)--(g3)--(g4);
    \draw (g5)--(g3);
    \draw (g6)--(g3);
\end{scope}

\begin{scope}[shift={(0,-12)}]
    \node (g2) at (-1.5,0) [gauge,label=below:{$2$}] {};
    \node (g2a) at (-1.5,-1) {$\textcolor{red}{b}$};
    \node (g3) at (-0.5,0) [gauge,label=below:{$4$}] {};
    \node (g3a) at (-0.5,-1) {$\textcolor{red}{c}$};
    \node (g4) at (0.5,0) [gauge,label=below:{$6$}] {};
    \node (g5) at (1.3,0.8) [gauge,label=right:{$3$}] {};
    \node (g6) at (1.3,-0.8) [gauge, label=right:{$3$}] {};
    \node (g8) at (0.5,1) [gauger,label=left:{$2$}] {};
    \draw (g4)--(g8);
    \draw (g2)--(g3)--(g4)--(g5);
    \draw (g4)--(g6);
\end{scope}

    \draw[->] (2.5,0) to [out=300,in=60,looseness=1] (2.5,-5.8);
    \draw[->] (2.5,-6.2) to [out=300,in=60,looseness=1] (2.5,-11.8);
\end{tikzpicture}}
\end{equation}
One can check that subtracting $d_4$ twice from the final quiver in (\ref{correct d4 addition to nminB5}) will indeed give n.min $B_5$ (\ref{n.min B5}) as desired. The final quiver in (\ref{correct d4 addition to nminB5}), following the algorithm in Section \ref{sec:BGS}, has BGS $SO(10)$. However, we have seen that the Hasse diagram contains $SO(11)$, and thus this must be a subgroup of the global symmetry. Indeed, upon Hilbert series computation, $SO(11)$ is the confirmed EGS. \hfill $\square$\\

\noindent This concludes the methodology needed for the results of this paper. Throughout the paper the results of quiver additions will be labelled with two parameters, $n$ and $k$, pertaining to the Coulomb branch variety of the base quiver from Table \ref{tab:adj hyper Qs} we add to (or equivalently the rank of the global symmetry factor which the BGS is enhanced to) and the rank of the Coulomb branch variety of the slice which we are adding twice respectively. In some cases, the only valid way to add a slice will restrict the allowed values of $n$ and or $k$. We now proceed to performing the quiver addition method discussed in this section on the quivers from Table \ref{tab:adj hyper Qs} one by one (excluding the $a_{2n-1}/\mathbb{Z}_2$ case, as discussed in the bullet points following Table \ref{tab:adj hyper Qs}), to derive quivers with an EGS.

\section{\texorpdfstring{Enhancement to $SO(2n+1)$}{Enhancement to Bn}}
\label{sec:enhancement to Bn}

\noindent In this section, we focus on deriving quivers with enhanced global symmetry from the first quiver in Table \ref{tab:adj hyper Qs}, the next to minimal nilpotent orbit of $B_n$:
\begin{equation}\label{nminBn}
    \begin{tikzpicture}[x=1cm,y=.8cm]
    \node (g1) at (-1.5,1.1) [gauge,label=right:{$2$},label={[red]left:{$a$}}] {};
    \node (g2) at (-1.5,0) [gauge,label=below:{$2$},label={[red]left:{$b$}}] {};
    \node (g3) at (-0.5,0) {$\cdots$};
    \node (g4) at (0.5,0) [gauge,label=below:{$2$}] {};
    \node (g5) at (1.3,-0.8) [gauge,label=right:{$1$}] {};
    \node (g6) at (1.3,0.8) [gauge, label=right:{$1$}] {};
    \draw (g1)--(g2)--(g3)--(g4)--(g5);
    \draw (g4)--(g6);
    \draw (g1) to [out=45,in=135,looseness=10] (g1);
    \draw [decorate,decoration={brace,mirror,amplitude=6pt}] (-1.5,-0.8) --node[below=6pt] {$n-2$} (0.5,-0.8);
    \end{tikzpicture}
\end{equation}
Using the refinement of the BGS algorithm given in Section \ref{sec:quivers with adj matter}, one can see that the BGS of (\ref{nminBn}) is $SO(2n+1)$. Indeed, the Hasse diagram of (\ref{nminBn}) is
\begin{equation} \label{HD nminBn}
    \begin{tikzpicture}
    \node (1) [hasse] at (0,0) {};
    \node (2) [hasse] at (0,-1) {};
    \node (3) [hasse] at (0,-2) {};
    \node (4) at (0.5,-1) {$,$};
    \draw (1) edge [] node[label=left:$A_1$] {} (2);
    \draw (2) edge [] node[label=left:$b_n$] {} (3);
    \end{tikzpicture}
\end{equation}
which agrees with this, and this global symmetry is confirmed upon Hilbert series computation. The quivers derived from using quiver addition on (\ref{nminBn}) to absorb $a$ will therefore experience a symmetry enhancement of this type.  After considering all possible additions following the addition algorithm listed in Section \ref{sec: quiver addition}, two types of enhancement are found: $SO(2n) \rightarrow SO(2n+1)$ and $SU(n) \times U(1) \rightarrow SO(2n+1)$. In all cases listed, except those where an $e_8$ is added twice, the EGS has been verified via Hilbert series computation. In the cases where adding a slice gives a family of quivers depending on one or two parameters,\footnote{These parameters are typically $n$ to specify which n.min $B_n$ we are adding on to (as in (\ref{nminBn})), and $k$ to specify the affine slice $g_k$ being added.} this EGS has been verified via Hilbert series computation for several low values of the parameters. The results for all quivers not depending on parameters have been confirmed too, except the $e_8$ cases which have proved too computationally complex to verify, so the EGS in these cases remain as conjectures. Brane systems for a selection of the $e_8$ cases listed both in this section and in subsequent sections are known and the construction is discussed in \cite{Cabrera:2019izd}. The result of the $e_8$ enhancement shown in Table \ref{tab:adding to RHS nmin Bn} has been previously found in \cite{Bhardwaj:2019jtr} through an F-theory construction, and in \cite{Zafrir:2015rga} using $5d$ brane webs. Several of the quivers appearing in this work also appear in the construction of so-called minimally unbalanced quivers in \cite{Cabrera:2018uvz}.\\

\subsection{\texorpdfstring{Enhancement from $SO(2n)$ to $SO(2n+1)$}{Enhancement from SO(2n) to SO(2nadd1)}} \label{sec:dn to bn}
In this section, we see quivers whose BGS has as one of its factors a $D_n$ type symmetry which undergoes enhancements to $B_n$. That is,
\begin{equation}
    \prod_{i} G_i \times \textcolor{blue}{SO(2n)} \rightarrow \prod_{i} G_i \times \textcolor{blue}{SO(2n+1)}
\end{equation}
for some semi-simple Lie groups $G_i$. Here we use blue to highlight the symmetry factor which experiences enhancement, and this will continue to be used throughout the paper. There are two possible ways we can construct these quivers by using quiver addition to absorb $a$ of (\ref{nminBn}): by adding to its existing nodes or its empty nodes (as discussed in the bullet points following the quiver addition algorithm of Section \ref{sec: quiver addition}). We explore each of these in turn now.

\subsubsection{\texorpdfstring{Adding to existing nodes of n.min $B_n$}{Adding to existing nodes of n.min Bn}} \label{sec:existing nodes nminbn}
Let's first investigate adding on to the existing nodes in (\ref{nminBn}). The case of $n=3$ will be treated separately. For the case of $n=2$, there are two $b$-nodes and no possible existing option for the node $c$ (taking $b$ and $c$ nodes as defined in the quiver addition algorithm of Section \ref{sec: quiver addition}), so we cannot add on to existing nodes here. Note however that we can add on to empty nodes, and this will be discussed in Section \ref{lhsnminBnsection}.\\

\noindent For adding to the existing nodes of (\ref{nminBn}) for $n \geq 4$, our $c$-node will be given by
\begin{equation}\label{cnodenminBngeq4}
    \begin{tikzpicture}[x=1cm,y=.8cm]
    \node (g1) at (-2,1.1) [gauge,label=right:{$2$},label={[red]left:{$a$}}] {};
    \node (g2) at (-2,0) [gauge,label=below:{$2$},label={[red]left:{$b$}}] {};
    \node (g2a) at (-1,0) [gauge,label=below:{$2$},label={[red]above:{$c$}}] {};
    \node (g3) at (0,0) {$\cdots$};
    \node (g4) at (1,0) [gauge,label=below:{$2$}] {};
    \node (g5) at (1.8,-0.8) [gauge,label=right:{$1$}] {};
    \node (g6) at (1.8,0.8) [gauge, label=right:{$1$}] {};
    \draw (g1)--(g2)--(g2a)--(g3)--(g4)--(g5);
    \draw (g4)--(g6);
    \draw (g1) to [out=45,in=135,looseness=10] (g1);
    \draw [decorate,decoration={brace,mirror,amplitude=6pt}] (-2,-0.8) --node[below=6pt] {$n-2$} (1,-0.8);
    \end{tikzpicture}
\end{equation}
The only slices that can be validly added to (\ref{cnodenminBngeq4}) according to the quiver addition algorithm in Section \ref{sec: quiver addition} are $a_k$, $d_k$, $e_6$, $e_7$ and $e_8$. The method of addition for all slices can be easily extended from the $d_4$ example shown in (\ref{correct d4 addition to nminB5}). Adding each of these slices ($a_k$, $d_k$, $e_6$, $e_7$ and $e_8$) will fix $n$ in (\ref{nminBn}) to be a particular value, in order for the addition to obey Step $2$ of our quiver addition algorithm. The results of these additions are listed in Table \ref{tab:adding to RHS nmin Bn}. The second column gives the quiver arrived at after twice adding the affine quiver corresponding to the slice listed in the first column of (\ref{nminBn}) to the quiver (\ref{nminBn}) for the values of $n$ listed. The third column shows the enhancement of the BGS to the EGS, with the factor(s) that experience enhancement shown in blue. This will also be the format of all future tables documenting the results in this paper.\\

\begin{table}[hbt!]
    \centering
    
 \hspace*{-1.1cm}\begin{tabular}{|c|c|c|} \hline
 
Added Slice & Quiver & Global Symmetry \\

\hline

$\begin{array}{c}
	a_3\\
	\\
	n=4
	\end{array}$ &
$\raisebox{-.5\height}{\begin{tikzpicture}[x=1cm,y=.8cm]
\node (g1) at (-1,0) [gauge,label=below:{$2$}] {};
\node (g2) at (0,0) [gauge,label=below:{$4$}] {};
\node (g3) at (1,1) [gauge,label=right:{$3$}] {};
\node (g4) at (1,0) [gauger,label=right:{$2$}] {};
\node (g7) at (1,-1) [gauge,label=right:{$3$}] {};
\draw (g1)--(g2)--(g3)--(g4)--(g7);
\draw (g7)--(g2);
\end{tikzpicture}}$ & 
$\renewcommand{\arraystretch}{1.5}
\begin{array}{c}
\textcolor{blue}{SO(8)} \\
\downarrow \\
\textcolor{blue}{SO(9)}
\end{array}$ \\
\hline

$\begin{array}{c}
	a_k\\
	\\
	k\geq 4, \ n=4
	\end{array}$ &
$\raisebox{-.5\height}{\begin{tikzpicture}[x=1cm,y=.8cm]
\node (g1) at (-1,0) [gauge,label=below:{$2$}] {};
\node (g2) at (0,0) [gauge,label=below:{$4$}] {};
\node (g3) at (1,3) [gauge,label=right:{$3$}] {};
\node (g4) at (1,2) [gauger,label=right:{$2$}] {};
\node (g8) at (1,1) [gauge,label=right:{$2$}] {};
\node (g5) at (1,0) {$\vdots$};
\node (g9) at (1,-1) [gauge,label=right:{$2$}] {};
\node (g6) at (1,-2) [gauger,label=right:{$2$}] {};
\node (g7) at (1,-3) [gauge,label=right:{$3$}] {};
\draw (g1)--(g2)--(g3)--(g4)--(g8)--(g5)--(g9)--(g6)--(g7);
\draw (g7)--(g2);
\draw [decorate,decoration={brace,amplitude=6pt}] (1.6,1) --node[right=6pt] {$k-4$} (1.6,-1);
\end{tikzpicture}}$ & 
$\renewcommand{\arraystretch}{1.5}
\begin{array}{c}
\textcolor{blue}{SO(8)} \times U(k-3) \\
\downarrow \\
\textcolor{blue}{SO(9)} \times U(k-3)
\end{array}$ \\
\hline

$\begin{array}{c}
	d_k\\
	\\
	k\geq 4, \ n=k+1
	\end{array}$ &
$\raisebox{-.5\height}{\begin{tikzpicture}[x=1cm,y=.8cm]
\node (g1) at (-2.5,0) [gauge,label=below:{$2$}] {};
\node (g2) at (-1.5,0) [gauge,label=below:{$4$}] {};
\node (g3) at (-0.5,0) [gauge,label=below:{$6$}] {};
\node (g4) at (0.5,0) {$\cdots$};
\node (g5) at (1.5,0) [gauge,label=below:{$6$}] {};
\node (g6) at (2.3,0.8) [gauge,label=right:{$3$}] {};
\node (g7) at (2.3,-0.8) [gauge,label=right:{$3$}] {};
\node (g8) at (-0.5,1) [gauger,label=above:{$2$}] {};
\draw (g1)--(g2)--(g3)--(g4)--(g5)--(g6);
\draw (g5)--(g7);
\draw (g3)--(g8);
\draw [decorate,decoration={brace,mirror,amplitude=6pt}] (-0.5,-0.8) --node[below=6pt] {$k-3 = n-4$} (1.5,-0.8);
\end{tikzpicture}}$ & 
$\renewcommand{\arraystretch}{1.5}
\begin{array}{c}
\textcolor{blue}{SO(2n)} \\
\downarrow \\
\textcolor{blue}{SO(2n+1)}
\end{array}$ \\
\hline

$\begin{array}{c}
	e_6\\
	\\
	n=6
	\end{array}$ &
$\raisebox{-.5\height}{\begin{tikzpicture}[x=1cm,y=.8cm]
\node (g1) at (-2.5,0) [gauge,label=below:{$2$}] {};
\node (g2) at (-1.5,0) [gauge,label=below:{$4$}] {};
\node (g3) at (-0.5,0) [gauge,label=below:{$6$}] {};
\node (g4) at (0.5,0) [gauge,label=below:{$8$}] {};
\node (g5) at (1.5,0) [gauge,label=below:{$5$}] {};
\node (g6) at (2.5,0) [gauger,label=below:{$2$}] {};
\node (g7) at (0.5,1) [gauge,label=right:{$5$}] {};
\node (g8) at (0.5,2) [gauger,label=right:{$2$}] {};
\draw (g1)--(g2)--(g3)--(g4)--(g5)--(g6);
\draw (g4)--(g7)--(g8);
\end{tikzpicture}}$ & 
$\renewcommand{\arraystretch}{1.5}
\begin{array}{c}
\textcolor{blue}{SO(12)} \times U(1) \\
\downarrow \\
\textcolor{blue}{SO(13)} \times U(1)
\end{array}$ \\
\hline

$\begin{array}{c}
	e_7\\
	\\
	n=7
	\end{array}$ &
$\raisebox{-.5\height}{\begin{tikzpicture}[x=1cm,y=.8cm]
\node (g1) at (-3.5,0) [gauge,label=below:{$2$}] {};
\node (g2) at (-2.5,0) [gauge,label=below:{$4$}] {};
\node (g3) at (-1.5,0) [gauge,label=below:{$6$}] {};
\node (g4) at (-0.5,0) [gauge,label=below:{$8$}] {};
\node (g5) at (0.5,0) [gauge,label=below:{$10$}] {};
\node (g6) at (1.5,0) [gauge,label=below:{$7$}] {};
\node (g7) at (2.5,0) [gauger,label=below:{$4$}] {};
\node (g8) at (3.5,0) [gauge,label=below:{$2$}] {};
\node (g9) at (0.5,1) [gauge,label=right:{$5$}] {};
\draw (g1)--(g2)--(g3)--(g4)--(g5)--(g6)--(g7)--(g8);
\draw (g5)--(g9);
\end{tikzpicture}}$ & 
$\renewcommand{\arraystretch}{1.5}
\begin{array}{c}
\textcolor{blue}{SO(14)} \times SU(2) \\
\downarrow \\
\textcolor{blue}{SO(15)} \times SU(2)
\end{array}$ \\
\hline

$\begin{array}{c}
	e_8\\
	\\
	n=9
	\end{array}$ &
$\raisebox{-.5\height}{\begin{tikzpicture}[x=1cm,y=.8cm]
\node (g1) at (-4,0) [gauge,label=below:{$2$}] {};
\node (g2) at (-3,0) [gauge,label=below:{$4$}] {};
\node (g3) at (-2,0) [gauge,label=below:{$6$}] {};
\node (g4) at (-1,0) [gauge,label=below:{$8$}] {};
\node (g5) at (0,0) [gauge,label=below:{$10$}] {};
\node (g6) at (1,0) [gauge,label=below:{$12$}] {};
\node (g7) at (2,0) [gauge,label=below:{$14$}] {};
\node (g8) at (3,0) [gauge,label=below:{$9$}] {};
\node (g9) at (4,0) [gauger,label=below:{$4$}] {};
\node (g10) at (2,1) [gauge,label=right:{$7$}] {};
\draw (g1)--(g2)--(g3)--(g4)--(g5)--(g6)--(g7)--(g8)--(g9);
\draw (g7)--(g10);
\end{tikzpicture}}$ & 
$\renewcommand{\arraystretch}{1.5}
\begin{array}{c}
\textcolor{blue}{SO(18)} \\
\downarrow \\
\textcolor{blue}{SO(19)}
\end{array}$ \\
\hline

\end{tabular}

 \caption{Quivers resulting from adding all possible elementary slices to the existing nodes of n.min $B_n$ (\ref{nminBn}) to absorb $a$, $n \geq 4$.}
    \label{tab:adding to RHS nmin Bn}
\end{table}

\clearpage 

\noindent For the case of $n=3$, there are two possible $c$-nodes
\begin{equation}\label{cnodenminBn3}
    \begin{tikzpicture}[x=1cm,y=.8cm]
    \node (g1) at (-0.5,1.1) [gauge,label=right:{$2$},label={[red]left:{$a$}}] {};
    \node (g2) at (-0.5,0) [gauge,label=below:{$2$},label={[red]left:{$b$}}] {};
    \node (g5a) at (1,-0.8) {\textcolor{red}{$c_1$}};
    \node (g5) at (0.3,-0.8) [gauge,label=right:{$1$}] {};
    \node (g6) at (0.3,0.8) [gauge, label=right:{$1$}] {};
    \node (g6a) at (1,0.8) {\textcolor{red}{$c_2$}};
    \draw (g1)--(g2)--(g5);
    \draw (g2)--(g6);
    \draw (g1) to [out=45,in=135,looseness=10] (g1);
    \end{tikzpicture}
\end{equation}
which we call $c_1$ and $c_2$. However due to the $\mathbb{Z}_2$ outer automorphism of (\ref{cnodenminBn3}), adding to either $c$-node gives the same result, and so only one need be considered. Here there are more allowed quivers that we can add to the $c$-node of (\ref{cnodenminBn3}) than there were for (\ref{cnodenminBngeq4}), because (\ref{cnodenminBn3}) has no nodes connected to the $c$-node that we have to ensure the added slice contains as a subset. As a result, we may add any affine quiver to the existing nodes of n.min $B_3$. The resulting quivers are listed in Tables \ref{tab:sl addition nmin B3} and \ref{tab:nsl addition nmin B3}.

\clearpage

\begin{table}[htbp!]
    \centering
    
 \hspace*{-1.35cm}\begin{tabular}{|c|c|c|} \hline
 
Added Slice & Quiver & Global Symmetry \\

\hline

$\begin{array}{c}
	a_1\\
	\end{array}$ &
$\raisebox{-.5\height}{\begin{tikzpicture}[x=1cm,y=.8cm]
\node (g1) at (-1.5,0) [gauge,label=below:{$1$}] {};
\node (g2) at (-0.5,0) [gauge,label=below:{$2$}] {};
\node (g3) at (0.5,0) [gauge,label=below:{$3$}] {};
\node (g4) at (1.5,0) [gauger,label=below:{$2$}] {};
\draw (g1)--(g2)--(g3);
\draw[transform canvas ={yshift=-2pt}] (g3)--(g4);
\draw[transform canvas ={yshift=2pt}] (g3)--(g4);
\end{tikzpicture}}$ & 
$\renewcommand{\arraystretch}{1.5}
\begin{array}{c}
\textcolor{blue}{SU(4)} \\
\downarrow \\
\textcolor{blue}{SO(7)}
\end{array}$ \\
\hline

$\begin{array}{c}
	a_k\\
	\\
	k\geq 2
	\end{array}$ &
$\raisebox{-.5\height}{\begin{tikzpicture}[x=1cm,y=.8cm]
\node (g1) at (-3,0) [gauge,label=below:{$1$}] {};
\node (g2) at (-2,0) [gauge,label=below:{$2$}] {};
\node (g3) at (-1,0) [gauge,label=below:{$3$}] {};
\node (g4) at (0,0) [gauger,label=below:{$2$}] {};
\node (g8) at (1,0) [gauge,label=below:{$2$}] {};
\node (g5) at (2,0) {$\cdots$};
\node (g6) at (3,0) [gauge,label=below:{$2$}] {};
\node (g7) at (1,1) [gauger,label=above:{$2$}] {};
\draw (g1)--(g2)--(g3)--(g4)--(g8)--(g5)--(g6)--(g7);
\draw (g3)--(g7);
\draw [decorate,decoration={brace,mirror,amplitude=6pt}] (1,-0.8) --node[below=6pt] {$k-2$} (3,-0.8);
\end{tikzpicture}}$ & 
$\renewcommand{\arraystretch}{1.5}
\begin{array}{c}
\textcolor{blue}{SU(4)} \times U(k-1) \\
\downarrow \\
\textcolor{blue}{SO(7)} \times U(k-1)
\end{array}$ \\
\hline

$\begin{array}{c}
	d_k\\
	\\
	k\geq 4
	\end{array}$ &
$\raisebox{-.5\height}{\begin{tikzpicture}[x=1cm,y=.8cm]
\node (g0) at (-3,0) [gauge,label=below:{$1$}] {};
\node (g1) at (-2,0) [gauge,label=below:{$2$}] {};
\node (g2) at (-1,0) [gauge,label=below:{$3$}] {};
\node (g3) at (0,0) [gauger,label=below:{$4$}] {};
\node (g4) at (1,0) {$\cdots$};
\node (g5) at (2,0) [gauge,label=below:{$4$}] {};
\node (g12) at (3,0) [gauge,label=below:{$2$}] {};
\node (g13) at (2,1) [gauge,label=above:{$2$}] {};
\node (g11) at (0,1) [gauge,label=above:{$2$}] {};
\draw (g0)--(g1)--(g2)--(g3)--(g4)--(g5)--(g12);
\draw (g13)--(g5);
\draw (g11)--(g3);
\draw [decorate,decoration={brace,mirror,amplitude=6pt}] (0,-0.8) --node[below=6pt] {$k-3$} (2,-0.8);
\end{tikzpicture}}$ & 
$\renewcommand{\arraystretch}{1.5}
\begin{array}{c}
\textcolor{blue}{SU(4)} \times SU(2) \times SO(2k-4) \\
\downarrow \\
\textcolor{blue}{SO(7)} \times SU(2) \times SO(2k-4)
\end{array}$ \\
\hline

$\begin{array}{c}
	e_6
	\end{array}$ &
$\raisebox{-.5\height}{\begin{tikzpicture}[x=1cm,y=.8cm]
\node (g0) at (-3,0) [gauge,label=below:{$1$}] {};
\node (g1) at (-2,0) [gauge,label=below:{$2$}] {};
\node (g2) at (-1,0) [gauge,label=below:{$3$}] {};
\node (g3) at (0,0) [gauger,label=below:{$4$}] {};
\node (g4) at (1,0) [gauge,label=below:{$6$}] {};
\node (g5) at (2,0) [gauge,label=below:{$4$}] {};
\node (g12) at (3,0) [gauge,label=below:{$2$}] {};
\node (g10) at (1,1) [gauge,label=right:{$4$}] {};
\node (g11) at (1,2) [gauge,label=right:{$2$}] {};
\draw (g0)--(g1)--(g2)--(g3)--(g4)--(g5)--(g12);
\draw (g4)--(g10)--(g11);
\end{tikzpicture}}$ & 
$\renewcommand{\arraystretch}{1.5}
\begin{array}{c}
\textcolor{blue}{SU(4)} \times SU(6) \\
\downarrow \\
\textcolor{blue}{SO(7)} \times SU(6)
\end{array}$ \\
\hline

$\begin{array}{c}
	e_7
	\end{array}$ &
$\raisebox{-.5\height}{\begin{tikzpicture}[x=1cm,y=.8cm]
\node (g1) at (-4,0) [gauge,label=below:{$1$}] {};
\node (g2) at (-3,0) [gauge,label=below:{$2$}] {};
\node (g3) at (-2,0) [gauge,label=below:{$3$}] {};
\node (g4) at (-1,0) [gauger,label=below:{$4$}] {};
\node (g13) at (0,0) [gauge,label=below:{$6$}] {};
\node (g5) at (1,0) [gauge,label=below:{$8$}] {};
\node (g12) at (2,0) [gauge,label=below:{$6$}] {};
\node (g6) at (3,0) [gauge,label=below:{$4$}] {};
\node (g10) at (4,0) [gauge,label=below:{$2$}] {};
\node (g7) at (1,1) [gauge,label=above:{$4$}] {};
\draw (g1)--(g2)--(g3)--(g4)--(g13)--(g5)--(g12)--(g6)--(g10);
\draw (g7)--(g5);
\end{tikzpicture}}$ & 
$\renewcommand{\arraystretch}{1.5}
\begin{array}{c}
\textcolor{blue}{SU(4)} \times SO(12) \\
\downarrow \\
\textcolor{blue}{SO(7)} \times SO(12)
\end{array}$ \\
\hline

$\begin{array}{c}
	e_8
	\end{array}$ &
$\raisebox{-.5\height}{\begin{tikzpicture}[x=1cm,y=.8cm]
\node (g1) at (-4.5,0) [gauge,label=below:{$1$}] {};
\node (g2) at (-3.5,0) [gauge,label=below:{$2$}] {};
\node (g3) at (-2.5,0) [gauge,label=below:{$3$}] {};
\node (g4) at (-1.5,0) [gauger,label=below:{$4$}] {};
\node (g5) at (-0.5,0) [gauge,label=below:{$6$}] {};
\node (g6) at (0.5,0) [gauge,label=below:{$8$}] {};
\node (g7) at (1.5,0) [gauge,label=below:{$10$}] {};
\node (g8) at (2.5,0) [gauge,label=below:{$12$}] {};
\node (g9) at (3.5,0) [gauge,label=below:{$8$}] {};
\node (g10) at (4.5,0) [gauge,label=below:{$4$}] {};
\node (g11) at (2.5,1) [gauge,label=above:{$6$}] {};
\draw (g1)--(g2)--(g3)--(g4)--(g5)--(g6)--(g7)--(g8)--(g9)--(g10);
\draw (g8)--(g11);
\end{tikzpicture}}$ & 
$\renewcommand{\arraystretch}{1.5}
\begin{array}{c}
\textcolor{blue}{SU(4)} \times E_7 \\
\downarrow \\
\textcolor{blue}{SO(7)} \times E_7
\end{array}$ \\
\hline

\end{tabular}
 \caption{Quivers resulting from adding the possible simply laced elementary slices to n.min $B_3$ (\ref{cnodenminBn3}) to absorb $a$.}
    \label{tab:sl addition nmin B3}
\end{table}

\begin{table}[htbp!]
    \centering
    
 \hspace*{-0.05cm}\begin{tabular}{|c|c|c|} \hline
 
Added Slice & Quiver & Global Symmetry \\

\hline

$\begin{array}{c}
	b_k\\
	\\
	k \geq 3
	\end{array}$ &
$\raisebox{-.5\height}{\begin{tikzpicture}[x=1cm,y=.8cm]
\node (g0) at (-3,0) [gauge,label=below:{$1$}] {};
\node (g1) at (-2,0) [gauge,label=below:{$2$}] {};
\node (g2) at (-1,0) [gauge,label=below:{$3$}] {};
\node (g3) at (0,0) [gauger,label=below:{$4$}] {};
\node (g4) at (1,0) {$\cdots$};
\node (g5) at (2,0) [gauge,label=below:{$4$}] {};
\node (g12) at (3,0) [gauge,label=below:{$2$}] {};
\node (g11) at (0,1) [gauge,label=above:{$2$}] {};
\draw (g0)--(g1)--(g2)--(g3)--(g4)--(g5);
\draw (g3)--(g11);
\draw[transform canvas={yshift=-1.5pt}] (g5)--(g12);
\draw[transform canvas={yshift=1.5pt}] (g5)--(g12);
\draw (2.4,-0.2)--(2.6,0)--(2.4,0.2);
\draw [decorate,decoration={brace,mirror,amplitude=6pt}] (0,-0.8) --node[below=6pt] {$k-2$} (2,-0.8);
\end{tikzpicture}}$ & 
$\renewcommand{\arraystretch}{1.5}
\begin{array}{c}
\textcolor{blue}{SU(4)} \times SU(2) \times SO(2k-3) \\
\downarrow \\
\textcolor{blue}{SO(7)} \times SU(2) \times SO(2k-3)
\end{array}$ \\
\hline

$\begin{array}{c}
	c_k\\
	\\
	k \geq 2
	\end{array}$ &
$\raisebox{-.5\height}{\begin{tikzpicture}[x=1cm,y=.8cm]
\node (g0) at (-3,0) [gauge,label=below:{$1$}] {};
\node (g1) at (-2,0) [gauge,label=below:{$2$}] {};
\node (g2) at (-1,0) [gauge,label=below:{$3$}] {};
\node (g3) at (0,0) [gauger,label=below:{$2$}] {};
\node (g4) at (1,0) {$\cdots$};
\node (g5) at (2,0) [gauge,label=below:{$2$}] {};
\node (g12) at (3,0) [gauge,label=below:{$2$}] {};
\draw (g0)--(g1)--(g2);
\draw (g3)--(g4)--(g5);
\draw[transform canvas={yshift=-1.5pt}] (g2)--(g3);
\draw[transform canvas={yshift=1.5pt}] (g2)--(g3);
\draw (-0.6,-0.2)--(-0.4,0)--(-0.6,0.2);
\draw[transform canvas={yshift=-1.5pt}] (g5)--(g12);
\draw[transform canvas={yshift=1.5pt}] (g5)--(g12);
\draw (2.6,-0.2)--(2.4,0)--(2.6,0.2);
\draw [decorate,decoration={brace,mirror,amplitude=6pt}] (0,-0.8) --node[below=6pt] {$k-1$} (2,-0.8);
\end{tikzpicture}}$ & 
$\renewcommand{\arraystretch}{1.5}
\begin{array}{c}
\textcolor{blue}{SU(4)} \times Sp(k-1) \\
\downarrow \\
\textcolor{blue}{SO(7)} \times Sp(k-1)
\end{array}$ \\
\hline

$\begin{array}{c}
	f_4
	\end{array}$ &
$\raisebox{-.5\height}{\begin{tikzpicture}[x=1cm,y=.8cm]
\node (g0) at (-3,0) [gauge,label=below:{$1$}] {};
\node (g1) at (-2,0) [gauge,label=below:{$2$}] {};
\node (g2) at (-1,0) [gauge,label=below:{$3$}] {};
\node (g3) at (0,0) [gauger,label=below:{$4$}] {};
\node (g4) at (1,0) [gauge,label=below:{$6$}] {};
\node (g5) at (2,0) [gauge,label=below:{$4$}] {};
\node (g12) at (3,0) [gauge,label=below:{$2$}] {};
\draw (g0)--(g1)--(g2)--(g3)--(g4);
\draw (g12)--(g5);
\draw[transform canvas={yshift=-1.5pt}] (g4)--(g5);
\draw[transform canvas={yshift=1.5pt}] (g4)--(g5);
\draw (1.4,-0.2)--(1.6,0)--(1.4,0.2);
\end{tikzpicture}}$ & 
$\renewcommand{\arraystretch}{1.5}
\begin{array}{c}
\textcolor{blue}{SU(4)} \times Sp(3) \\
\downarrow \\
\textcolor{blue}{SO(7)} \times Sp(3)
\end{array}$ \\
\hline

$\begin{array}{c}
    g_2
    \end{array}$ &
$\raisebox{-.5\height}{\begin{tikzpicture}[x=1cm,y=.8cm]
\node (g1) at (-2,0) [gauge,label=below:{$1$}] {};
\node (g2) at (-1,0) [gauge,label=below:{$2$}] {};
\node (g3) at (0,0) [gauge,label=below:{$3$}] {};
\node (g4) at (1,0) [gauger,label=below:{$4$}] {};
\node (g5) at (2,0) [gauge,label=below:{$2$}] {};
\draw (g1)--(g2)--(g3)--(g4);
\draw[transform canvas = {yshift=-2pt}] (g4)--(g5);
\draw[transform canvas = {yshift=0pt}] (g4)--(g5);
\draw[transform canvas = {yshift =2pt}] (g4)--(g5);
\draw (1.4,0.2)--(1.6,0)--(1.4,-0.2);
\end{tikzpicture}}$ &
$\renewcommand{\arraystretch}{1.5}
\begin{array}{c}
\textcolor{blue}{SU(4)} \times SU(2)\\
\downarrow \\
\textcolor{blue}{SO(7)} \times SU(2)
\end{array}$\\
\hline

\end{tabular}
 \caption{Quivers resulting from adding the possible non-simply laced elementary slices to n.min $B_3$ (\ref{cnodenminBn3}) to absorb $a$.}
    \label{tab:nsl addition nmin B3}
\end{table}

\clearpage

\subsubsection{\texorpdfstring{Adding to empty nodes of n.min $B_n$}{Adding to the empty nodes of n.min Bn}} \label{lhsnminBnsection}
We now turn our attention to adding on to the empty nodes (as defined in the quiver addition algorithm of Section \ref{sec: quiver addition}) of (\ref{nminBn}). That is, we take the $c$-node to be empty:
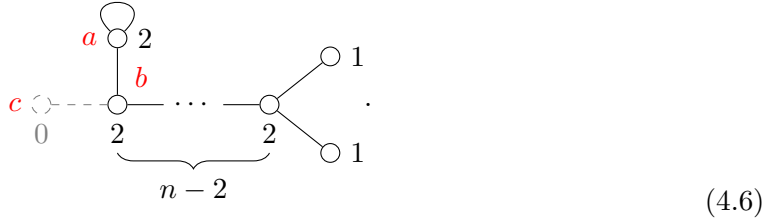
\begin{equation}\label{cnodeemptynminBn}
    \begin{tikzpicture}[x=1cm,y=.8cm]
    \node (g1) at (-1.5,1.1) [gauge,label=right:{$2$},label={[red]left:{$a$}}] {};
    \node (g2) at (-1.5,0) [gauge,label=below:{$2$},label={[red]45:{$b$}}] {};
    \node (g2a) at (-2.5,0) [emptygauge,label={[gray]below:{$0$}},label={[red]left:{$c$}}] {};
    \node (g3) at (-0.5,0) {$\cdots$};
    \node (g4) at (0.5,0) [gauge,label=below:{$2$}] {};
    \node (g5) at (1.3,-0.8) [gauge,label=right:{$1$}] {};
    \node (g6) at (1.3,0.8) [gauge, label=right:{$1$}] {};
    \node (g7) at (1.8,0) {$.$};
    \draw (g1)--(g2)--(g3)--(g4)--(g5);
    \draw (g4)--(g6);
    \draw [dashed, gray] (g2a)--(g2);
    \draw (g1) to [out=45,in=135,looseness=10] (g1);
    \draw [decorate,decoration={brace,mirror,amplitude=6pt}] (-1.5,-0.8) --node[below=6pt] {$n-2$} (0.5,-0.8);
    \end{tikzpicture}
\end{equation}
The dashed grey line indicates that $c$ is not really there, reinforced by its vanishing rank. Note that in the $n=2$ case there are two possible $b$-nodes, which we call $b_1$ and $b_2$, and their corresponding empty $c$-nodes will be called $c_1$ and $c_2$:
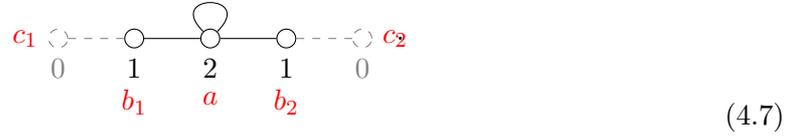
\begin{equation} \label{cnodesnminB2}
    \begin{tikzpicture}[x=1cm,y=.8cm]
\node (g2) at (-1,0) [gauge,label=below:{$1$}] {};
\node (g2a) at (-1,-1.05) {$\textcolor{red}{b_1}$};
\node (g2b) at (-2,0) [emptygauge,label={[gray]below:{$0$}},label={[red]left:{$c_1$}}] {};
\node (g3) at (0,0) [gauge,label=below:{$2$}] {};
\node (g2a) at (0,-1) {$\textcolor{red}{a}$};
\node (g4) at (1,0) [gauge,label=below:{$1$}] {};
\node (g4a) at (1,-1.05) {$\textcolor{red}{b_2}$};
\node (g4b) at (2,0) [emptygauge,label={[gray]below:{$0$}},label={[red]right:{$c_2$}}] {};
\node (g5) at (2.5,0) {$.$};
\draw (g2)--(g3)--(g4);
\draw [dashed, gray] (g4)--(g4b);
\draw [dashed, gray] (g2)--(g2b);
\draw (g3) to [out=45,in=135,looseness=10] (g3);
\end{tikzpicture}
\end{equation}
In order for the adjoint hypermultiplet to appear on the node $a$, which is connected to both $b_1$ and $b_2$, during quiver subtraction, $c$-nodes $c_1$ and $c_2$ must coincide so that an added slice is being ``attached" on to both $b_1$ and $b_2$. As a result, because the quiver addition algorithm of Section \ref{sec: quiver addition} requires any slice which is added to have two long rank one nodes, the allowed slices which can be validly added to \ref{cnodesnminB2} to absorb $a$ are restricted to $a_k$, $d_k$, $e_6$ and $e_7$. The results of these additions are given in Table \ref{tab:addition nmin B2}.\\

\begin{table}[htbp!]
    \centering
    
 \hspace*{-0.8cm}\begin{tabular}{|c|c|c|} \hline
 
Added Slice & Quiver & Global Symmetry \\

\hline

$\begin{array}{c}
	a_1
	\end{array}$ &
$\raisebox{-.5\height}{\begin{tikzpicture}[x=1cm,y=.8cm]
\node (g1) at (-1,0) [gauge,label=below:{$1$}] {};
\node (g2) at (0,0) [gauger,label=below:{$2$}] {};
\node (g4) at (0,1) [gauge,label=left:{$2$}] {};
\node (g5) at (1,0) [gauge,label=below:{$1$}] {};
\draw (g1)--(g2)--(g5);
\draw[transform canvas={xshift=-2pt}] (g2)--(g4);
\draw[transform canvas={xshift=2pt}] (g2)--(g4);
\end{tikzpicture}}$ & 
$\renewcommand{\arraystretch}{1.5}
\begin{array}{c}
\textcolor{blue}{SU(2)^2} \times SU(2) \\
\downarrow \\
\textcolor{blue}{SO(5)} \times SU(2)
\end{array}$ \\
\hline

$\begin{array}{c}
	a_k\\
	\\
	k\geq 2
	\end{array}$ &
$\raisebox{-.5\height}{\begin{tikzpicture}[x=1cm,y=.8cm]
\node (g1) at (-2,0) [gauge,label=below:{$1$}] {};
\node (g2) at (-1,0) [gauger,label=below:{$2$}] {};
\node (g3) at (0,0) {$\cdots$};
\node (g4) at (1,0) [gauger,label=below:{$2$}] {};
\node (g5) at (2,0) [gauge,label=below:{$1$}] {};
\node (g6) at (0,1) [gauge,label=above:{$2$}] {};
\draw (g1)--(g2)--(g3)--(g4)--(g5);
\draw (g6)--(g2);
\draw (g6)--(g4);
\draw [decorate,decoration={brace,mirror,amplitude=6pt}] (-1,-0.8) --node[below=6pt] {$k$} (1,-0.8);
\end{tikzpicture}}$ & 
$\renewcommand{\arraystretch}{1.5}
\begin{array}{c}
\textcolor{blue}{SU(2)^2} \times SU(2) \times U(k-1) \\
\downarrow \\
\textcolor{blue}{SO(5)} \times SU(2) \times U(k-1)
\end{array}$ \\
\hline

$\begin{array}{c}
	d_k\\
	\\
	k\geq 4
	\end{array}$ &
$\raisebox{-.5\height}{\begin{tikzpicture}[x=1cm,y=.8cm]
\node (g1) at (-3,0) [gauge,label=below:{$1$}] {};
\node (g2) at (-2,0) [gauger,label=below:{$2$}] {};
\node (g3) at (-1,0) [gauge,label=below:{$4$}] {};
\node (g4) at (0,0) {$\cdots$};
\node (g5) at (1,0) [gauge,label=below:{$4$}] {};
\node (g12) at (2,0) [gauger,label=below:{$2$}] {};
\node (g6) at (3,0) [gauge,label=below:{$1$}] {};
\node (g13) at (1,1) [gauge,label=above:{$2$}] {};
\node (g11) at (-1,1) [gauge,label=above:{$2$}] {};
\draw (g1)--(g2)--(g3)--(g4)--(g5)--(g12)--(g6);
\draw (g13)--(g5);
\draw (g11)--(g3);
\draw [decorate,decoration={brace,mirror,amplitude=6pt}] (-1,-0.8) --node[below=6pt] {$k-3$} (1,-0.8);
\end{tikzpicture}}$ & 
$\renewcommand{\arraystretch}{1.5}
\begin{array}{c}
\textcolor{blue}{SU(2)^2} \times U(k) \\
\downarrow \\
\textcolor{blue}{SO(5)} \times U(k)
\end{array}$ \\
\hline

$\begin{array}{c}
	e_6
	\end{array}$ &
$\raisebox{-.5\height}{\begin{tikzpicture}[x=1cm,y=.8cm]
\node (g0) at (-3,0) [gauge,label=below:{$1$}] {};
\node (g1) at (-2,0) [gauger,label=below:{$2$}] {};
\node (g2) at (-1,0) [gauge,label=below:{$4$}] {};
\node (g3) at (0,0) [gauge,label=below:{$6$}] {};
\node (g4) at (1,0) [gauge,label=below:{$4$}] {};
\node (g5) at (2,0) [gauger,label=below:{$2$}] {};
\node (g12) at (3,0) [gauge,label=below:{$1$}] {};
\node (g10) at (0,1) [gauge,label=right:{$4$}] {};
\node (g11) at (0,2) [gauge,label=right:{$2$}] {};
\draw (g0)--(g1)--(g2)--(g3)--(g4)--(g5)--(g12);
\draw (g3)--(g10)--(g11);
\end{tikzpicture}}$ & 
$\renewcommand{\arraystretch}{1.5}
\begin{array}{c}
\textcolor{blue}{SU(2)^2} \times SO(10) \times U(1) \\
\downarrow \\
\textcolor{blue}{SO(5)} \times SO(10) \times U(1)
\end{array}$ \\
\hline

$\begin{array}{c}
	e_7
	\end{array}$ &
$\raisebox{-.5\height}{\begin{tikzpicture}[x=1cm,y=.8cm]
\node (g1) at (-4,0) [gauge,label=below:{$1$}] {};
\node (g2) at (-3,0) [gauger,label=below:{$2$}] {};
\node (g3) at (-2,0) [gauge,label=below:{$4$}] {};
\node (g4) at (-1,0) [gauge,label=below:{$6$}] {};
\node (g13) at (0,0) [gauge,label=below:{$8$}] {};
\node (g5) at (1,0) [gauge,label=below:{$6$}] {};
\node (g12) at (2,0) [gauge,label=below:{$4$}] {};
\node (g6) at (3,0) [gauger,label=below:{$2$}] {};
\node (g7) at (4,0) [gauge,label=below:{$1$}] {};
\node (g10) at (0,1) [gauge,label=right:{$4$}] {};
\draw (g1)--(g2)--(g3)--(g4)--(g13)--(g5)--(g12)--(g6)--(g7);
\draw (g13)--(g10);
\end{tikzpicture}}$ & 
$\renewcommand{\arraystretch}{1.5}
\begin{array}{c}
\textcolor{blue}{SU(2)^2} \times E_6 \times U(1) \\
\downarrow \\
\textcolor{blue}{SO(5)} \times E_6 \times U(1)
\end{array}$ \\
\hline

\end{tabular}
 \caption{Quivers resulting from adding all possible elementary slices to n.min $B_2$ (\ref{cnodesnminB2}) to absorb $a$.}
    \label{tab:addition nmin B2}
\end{table}

\noindent In the case of $n \geq 3$ there is only one $b$-node, so the above problem does not arise. Here $c$ is still an empty node, and so the requirement in Step $1$ of the quiver addition algorithm in Section \ref{sec: quiver addition} that the added slice must contain the nodes in (\ref{cnodeemptynminBn}) connected to and including $c$ becomes trivial, and so we may add any affine slice. The results of adding these slices are given in Tables \ref{tab:lhs sl addition nmin Bn} and \ref{tab:lhs nsl addition nmin Bn}.

\begin{table}[hbt!]
    \centering
    
 \hspace*{-2.5cm}\begin{tabular}{|c|c|c|} \hline
 
Added Slice & Quiver & Global Symmetry \\

\hline

$\begin{array}{c}
	a_k\\
	\\
	k\geq 1, n \geq 3
	\end{array}$ &
$\raisebox{-.5\height}{\begin{tikzpicture}[x=1cm,y=.8cm]
\node (g1) at (-3,0) [gauge,label=below:{$2$}] {};
\node (g2) at (-2,0) {$\cdots$};
\node (g3) at (-1,0) [gauger,label=below:{$2$}] {};
\node (g4) at (0,0) [gauge,label=below:{$2$}] {};
\node (g5) at (1,0) {$\cdots$};
\node (g6) at (2,0) [gauge,label=below:{$2$}] {};
\node (g7) at (2.8,0.8) [gauge,label=right:{$1$}] {};
\node (g8) at (2.8,-0.8) [gauge,label=right:{$1$}] {};
\node (g9) at (-2,1) [gauge,label=above:{$2$}] {};
\draw (g1)--(g2)--(g3)--(g4)--(g5)--(g6)--(g7);
\draw (g6)--(g8);
\draw (g1)--(g9);
\draw (g3)--(g9);
\draw [decorate,decoration={brace,mirror,amplitude=6pt}] (-3,-0.8) --node[below=6pt] {$k$} (-1,-0.8);
\draw [decorate,decoration={brace,mirror,amplitude=6pt}] (0,-0.8) --node[below=6pt] {$n-2$} (2,-0.8);
\end{tikzpicture}}$ & 
$\renewcommand{\arraystretch}{1.5}
\begin{array}{c}
SU(k+1) \times \textcolor{blue}{SO(2n)}  \\
\downarrow \\
SU(k+1) \times \textcolor{blue}{SO(2n+1)}
\end{array}$ \\
\hline

$\begin{array}{c}
	d_k\\
	\\
	k\geq 4, n=k+1
	\end{array}$ &
$\raisebox{-.5\height}{\begin{tikzpicture}[x=1cm,y=.8cm]
\node (g1) at (-4,0) [gauge,label=below:{$2$}] {};
\node (g2) at (-3,0) [gauge,label=below:{$4$}] {};
\node (g3) at (-2,0) {$\cdots$};
\node (g4) at (-1,0) [gauge,label=below:{$4$}] {};
\node (g5) at (0,0) [gauger,label=below:{$2$}] {};
\node (g12) at (1,0) [gauge,label=below:{$2$}] {};
\node (g6) at (2,0) {$\cdots$};
\node (g7) at (3,0) [gauge,label=below:{$2$}] {};
\node (g8) at (3.8,0.8) [gauge,label=right:{$1$}] {};
\node (g9) at (3.8,-0.8) [gauge,label=right:{$1$}] {};
\node (g10) at (-3,1) [gauge,label=above:{$2$}] {};
\node (g11) at (-1,1) [gauge,label=above:{$2$}] {};
\draw (g1)--(g2)--(g3)--(g4)--(g5)--(g12)--(g6)--(g7)--(g8);
\draw (g7)--(g9);
\draw (g10)--(g2);
\draw (g11)--(g4);
\draw [decorate,decoration={brace,mirror,amplitude=6pt}] (-3,-0.8) --node[below=6pt] {$k-3$} (-1,-0.8);
\draw [decorate,decoration={brace,mirror,amplitude=6pt}] (1,-0.8) --node[below=6pt] {$n-2$} (3,-0.8);
\end{tikzpicture}}$ & 
$\renewcommand{\arraystretch}{1.5}
\begin{array}{c}
SO(2k) \times \textcolor{blue}{SO(2n)}  \\
\downarrow \\
SO(2k) \times \textcolor{blue}{SO(2n+1)}
\end{array}$ \\
\hline

$\begin{array}{c}
	e_6\\
	\\
	n \geq 3
	\end{array}$ &
$\raisebox{-.5\height}{\begin{tikzpicture}[x=1cm,y=.8cm]
\node (g1) at (-4,0) [gauge,label=below:{$2$}] {};
\node (g2) at (-3,0) [gauge,label=below:{$4$}] {};
\node (g3) at (-2,0) [gauge,label=below:{$6$}] {};
\node (g4) at (-1,0) [gauge,label=below:{$4$}] {};
\node (g5) at (0,0) [gauger,label=below:{$2$}] {};
\node (g12) at (1,0) [gauge,label=below:{$2$}] {};
\node (g6) at (2,0) {$\cdots$};
\node (g7) at (3,0) [gauge,label=below:{$2$}] {};
\node (g8) at (3.8,0.8) [gauge,label=right:{$1$}] {};
\node (g9) at (3.8,-0.8) [gauge,label=right:{$1$}] {};
\node (g10) at (-2,1) [gauge,label=right:{$4$}] {};
\node (g11) at (-2,2) [gauge,label=right:{$2$}] {};
\draw (g1)--(g2)--(g3)--(g4)--(g5)--(g12)--(g6)--(g7)--(g8);
\draw (g7)--(g9);
\draw (g3)--(g10)--(g11);
\draw [decorate,decoration={brace,mirror,amplitude=6pt}] (1,-0.8) --node[below=6pt] {$n-2$} (3,-0.8);
\end{tikzpicture}}$ & 
$\renewcommand{\arraystretch}{1.5}
\begin{array}{c}
E_6 \times \textcolor{blue}{SO(2n)}  \\
\downarrow \\
E_6 \times \textcolor{blue}{SO(2n+1)}
\end{array}$ \\
\hline

$\begin{array}{c}
	e_7\\
	n \geq 3
	\end{array}$ &
$\raisebox{-.5\height}{\begin{tikzpicture}[x=1cm,y=.8cm]
\node (g1) at (-5,0) [gauge,label=below:{$2$}] {};
\node (g2) at (-4,0) [gauge,label=below:{$4$}] {};
\node (g3) at (-3,0) [gauge,label=below:{$6$}] {};
\node (g4) at (-2,0) [gauge,label=below:{$8$}] {};
\node (g13) at (-1,0) [gauge,label=below:{$6$}] {};
\node (g5) at (-0,0) [gauge,label=below:{$4$}] {};
\node (g12) at (1,0) [gauger,label=below:{$2$}] {};
\node (g6) at (2,0) [gauge,label=below:{$2$}] {};
\node (g7) at (3,0) {$\cdots$};
\node (g14) at (4,0) [gauge,label=below:{$2$}] {};
\node (g8) at (4.8,0.8) [gauge,label=right:{$1$}] {};
\node (g9) at (4.8,-0.8) [gauge,label=right:{$1$}] {};
\node (g10) at (-2,1) [gauge,label=right:{$4$}] {};
\draw (g1)--(g2)--(g3)--(g4)--(g13)--(g5)--(g12)--(g6)--(g7)--(g14)--(g8);
\draw (g14)--(g9);
\draw (g4)--(g10);
\draw [decorate,decoration={brace,mirror,amplitude=6pt}] (2,-0.8) --node[below=6pt] {$n-2$} (4,-0.8);
\end{tikzpicture}}$ & 
$\renewcommand{\arraystretch}{1.5}
\begin{array}{c}
E_7 \times \textcolor{blue}{SO(2n)}  \\
\downarrow \\
E_7 \times \textcolor{blue}{SO(2n+1)}
\end{array}$ \\
\hline

$\begin{array}{c}
	e_8\\
	n \geq 3
	\end{array}$ &
$\raisebox{-.5\height}{\begin{tikzpicture}[x=1cm,y=.8cm]
\node (g1) at (-5.5,0) [gauge,label=below:{$4$}] {};
\node (g2) at (-4.5,0) [gauge,label=below:{$8$}] {};
\node (g3) at (-3.5,0) [gauge,label=below:{$12$}] {};
\node (g4) at (-2.5,0) [gauge,label=below:{$10$}] {};
\node (g13) at (-1.5,0) [gauge,label=below:{$8$}] {};
\node (g5) at (-0.5,0) [gauge,label=below:{$6$}] {};
\node (g15) at (0.5,0) [gauge,label=below:{$4$}] {};
\node (g12) at (1.5,0) [gauger,label=below:{$2$}] {};
\node (g6) at (2.5,0) [gauge,label=below:{$2$}] {};
\node (g7) at (3.5,0) {$\cdots$};
\node (g14) at (4.5,0) [gauge,label=below:{$2$}] {};
\node (g8) at (5.3,0.8) [gauge,label=right:{$1$}] {};
\node (g9) at (5.3,-0.8) [gauge,label=right:{$1$}] {};
\node (g10) at (-3.5,1) [gauge,label=right:{$6$}] {};
\draw (g1)--(g2)--(g3)--(g4)--(g13)--(g5)--(g15)--(g12)--(g6)--(g7)--(g14)--(g8);
\draw (g14)--(g9);
\draw (g3)--(g10);
\draw [decorate,decoration={brace,mirror,amplitude=6pt}] (2.5,-0.8) --node[below=6pt] {$n-2$} (4.5,-0.8);
\end{tikzpicture}}$ & 
$\renewcommand{\arraystretch}{1.5}
\begin{array}{c}
E_8 \times \textcolor{blue}{SO(2n)}  \\
\downarrow \\
E_8 \times \textcolor{blue}{SO(2n+1)}
\end{array}$ \\
\hline

\end{tabular}
 \caption{Quivers resulting from adding all possible simply laced elementary slices to the empty nodes of n.min $B_n$ (\ref{nminBn}) to absorb $a$, $n \geq 3$.}
    \label{tab:lhs sl addition nmin Bn}
\end{table}

\begin{table}[htbp!]
    \centering
    
 \hspace*{-0.75cm}\begin{tabular}{|c|c|c|} \hline
 
Added Slice & Quiver & Global Symmetry \\

\hline

$\begin{array}{c}
	b_k\\
	\\
	k\geq 3, n \geq 3
	\end{array}$ &
$\raisebox{-.5\height}{\begin{tikzpicture}[x=1cm,y=.8cm]
\node (g1) at (-4,0) [gauge,label=below:{$2$}] {};
\node (g2) at (-3,0) [gauge,label=below:{$4$}] {};
\node (g3) at (-2,0) {$\cdots$};
\node (g4) at (-1,0) [gauge,label=below:{$4$}] {};
\node (g5) at (0,0) [gauger,label=below:{$2$}] {};
\node (g12) at (1,0) [gauge,label=below:{$2$}] {};
\node (g6) at (2,0) {$\cdots$};
\node (g7) at (3,0) [gauge,label=below:{$2$}] {};
\node (g8) at (3.8,0.8) [gauge,label=right:{$1$}] {};
\node (g9) at (3.8,-0.8) [gauge,label=right:{$1$}] {};
\node (g11) at (-1,1) [gauge,label=above:{$2$}] {};
\draw (g2)--(g3)--(g4)--(g5)--(g12)--(g6)--(g7)--(g8);
\draw (g7)--(g9);
\draw (g11)--(g4);
\draw[transform canvas={yshift=-1.5pt}] (g1)--(g2);
\draw[transform canvas={yshift=1.5pt}] (g1)--(g2);
\draw (-3.4,0.2)--(-3.6,0)--(-3.4,-0.2);
\draw [decorate,decoration={brace,mirror,amplitude=6pt}] (-3,-0.8) --node[below=6pt] {$k-2$} (-1,-0.8);
\draw [decorate,decoration={brace,mirror,amplitude=6pt}] (1,-0.8) --node[below=6pt] {$n-2$} (3,-0.8);
\end{tikzpicture}}$ & 
$\renewcommand{\arraystretch}{1.5}
\begin{array}{c}
SO(2k+1) \times \textcolor{blue}{SO(2n)}  \\
\downarrow \\
SO(2k+1) \times \textcolor{blue}{SO(2n+1)}
\end{array}$ \\
\hline

$\begin{array}{c}
	c_k\\
	\\
	k\geq 2, n \geq 3
	\end{array}$ &
$\raisebox{-.5\height}{\begin{tikzpicture}[x=1cm,y=.8cm]
\node (g1) at (-4,0) [gauge,label=below:{$2$}] {};
\node (g2) at (-3,0) [gauge,label=below:{$2$}] {};
\node (g3) at (-2,0) {$\cdots$};
\node (g4) at (-1,0) [gauge,label=below:{$2$}] {};
\node (g5) at (0,0) [gauger,label=below:{$2$}] {};
\node (g12) at (1,0) [gauge,label=below:{$2$}] {};
\node (g6) at (2,0) {$\cdots$};
\node (g7) at (3,0) [gauge,label=below:{$2$}] {};
\node (g8) at (3.8,0.8) [gauge,label=right:{$1$}] {};
\node (g9) at (3.8,-0.8) [gauge,label=right:{$1$}] {};
\draw (g2)--(g3)--(g4);
\draw (g5)--(g12)--(g6)--(g7)--(g8);
\draw (g7)--(g9);
\draw[transform canvas={yshift=-1.5pt}] (g1)--(g2);
\draw[transform canvas={yshift=1.5pt}] (g1)--(g2);
\draw (-3.6,0.2)--(-3.4,0)--(-3.6,-0.2);
\draw[transform canvas={yshift=-1.5pt}] (g4)--(g5);
\draw[transform canvas={yshift=1.5pt}] (g4)--(g5);
\draw (-0.4,0.2)--(-0.6,0)--(-0.4,-0.2);
\draw [decorate,decoration={brace,mirror,amplitude=6pt}] (-3,-0.8) --node[below=6pt] {$k-1$} (-1,-0.8);
\draw [decorate,decoration={brace,mirror,amplitude=6pt}] (1,-0.8) --node[below=6pt] {$n-2$} (3,-0.8);
\end{tikzpicture}}$ & 
$\renewcommand{\arraystretch}{1.5}
\begin{array}{c}
Sp(k) \times \textcolor{blue}{SO(2n)}  \\
\downarrow \\
Sp(k) \times \textcolor{blue}{SO(2n+1)}
\end{array}$ \\
\hline

$\begin{array}{c}
	f_4\\
	\\
	n \geq 3
	\end{array}$ &
$\raisebox{-.5\height}{\begin{tikzpicture}[x=1cm,y=.8cm]
\node (g1) at (-4,0) [gauge,label=below:{$2$}] {};
\node (g2) at (-3,0) [gauge,label=below:{$4$}] {};
\node (g3) at (-2,0) [gauge,label=below:{$6$}] {};
\node (g4) at (-1,0) [gauge,label=below:{$4$}] {};
\node (g5) at (0,0) [gauger,label=below:{$2$}] {};
\node (g12) at (1,0) [gauge,label=below:{$2$}] {};
\node (g6) at (2,0) {$\cdots$};
\node (g7) at (3,0) [gauge,label=below:{$2$}] {};
\node (g8) at (3.8,0.8) [gauge,label=right:{$1$}] {};
\node (g9) at (3.8,-0.8) [gauge,label=right:{$1$}] {};
\draw (g1)--(g2);
\draw (g3)--(g4)--(g5)--(g12)--(g6)--(g7)--(g8);
\draw (g7)--(g9);
\draw[transform canvas={yshift=-1.5pt}] (g3)--(g2);
\draw[transform canvas={yshift=1.5pt}] (g3)--(g2);
\draw (-2.4,0.2)--(-2.6,0)--(-2.4,-0.2);
\draw [decorate,decoration={brace,mirror,amplitude=6pt}] (1,-0.8) --node[below=6pt] {$n-2$} (3,-0.8);
\end{tikzpicture}}$ & 
$\renewcommand{\arraystretch}{1.5}
\begin{array}{c}
F_4 \times \textcolor{blue}{SO(2n)}  \\
\downarrow \\
F_4 \times \textcolor{blue}{SO(2n+1)}
\end{array}$ \\
\hline

$\begin{array}{c}
	g_2\\
	\\
	n \geq 3
	\end{array}$ &
$\raisebox{-.5\height}{\begin{tikzpicture}[x=1cm,y=.8cm]
\node (g3) at (-3,0) [gauge,label=below:{$2$}] {};
\node (g4) at (-2,0) [gauge,label=below:{$4$}] {};
\node (g5) at (-1,0) [gauger,label=below:{$2$}] {};
\node (g12) at (0,0) [gauge,label=below:{$2$}] {};
\node (g6) at (1,0) {$\cdots$};
\node (g7) at (2,0) [gauge,label=below:{$2$}] {};
\node (g8) at (2.8,0.8) [gauge,label=right:{$1$}] {};
\node (g9) at (2.8,-0.8) [gauge,label=right:{$1$}] {};
\draw (g4)--(g5)--(g12)--(g6)--(g7)--(g8);
\draw (g7)--(g9);
\draw[transform canvas={yshift=-2pt}] (g3)--(g4);
    \draw[transform canvas={yshift=0pt}] (g3)--(g4);
    \draw[transform canvas={yshift=2pt}] (g3)--(g4);
\draw (-2.4,0.2)--(-2.6,0)--(-2.4,-0.2);
\draw [decorate,decoration={brace,mirror,amplitude=6pt}] (0,-0.8) --node[below=6pt] {$n-2$} (2,-0.8);
\end{tikzpicture}}$ & 
$\renewcommand{\arraystretch}{1.5}
\begin{array}{c}
G_2 \times \textcolor{blue}{SO(2n)}  \\
\downarrow \\
G_2 \times \textcolor{blue}{SO(2n+1)}
\end{array}$ \\
\hline

\end{tabular}
 \caption{Quivers resulting from adding all possible non-simply laced elementary slices to the empty nodes of n.min $B_n$ (\ref{nminBn}) to absorb $a$, $n \geq 3$.}
    \label{tab:lhs nsl addition nmin Bn}
\end{table}

\clearpage

\subsection{\texorpdfstring{Enhancement from $SU(n) \times U(1)$ to $SO(2n+1)$}{Enhancement from SU(n) times U(1) to SO(2n+1)}} 
\label{A type to B type section}
The results of Section \ref{sec:dn to bn} exhaust all possibilities of adding to (\ref{nminBn}) to immediately absorb $a$. However we could also have first performed a ``different sort" of quiver addition to absorb one of the two rank one nodes of (\ref{nminBn}), and \textit{then} have added to the resulting quiver from this to absorb $a$. This ``different" type of quiver addition does not quite follow the algorithm of Section \ref{sec: quiver addition}: it is actually a little simpler. Here, instead of ``undoing" multiple same slice quiver subtractions that result in an adjoint hypermultiplet of increasing rank, we instead ``undo" just a single quiver subtraction. We will not list this slightly different algorithm here, but for those interested it is derivable from reverse engineering the basic quiver subtraction process that is given in following Steps $1$ and $2)a)$ of Appendix \ref{app:quiver subtraction} with $E_s$ being empty. An example of the result of such an addition will also be given below. A key feature of this quiver addition worth mentioning is that, as before, in order for it to work we must be adding a balanced elementary slice via a long rank one node, restricting us to adding only the affine slices (whose magnetic quivers live in Table \ref{tab:affine Dds}). Such an addition is of interest because, if the quiver addition algorithm of Section \ref{sec: quiver addition} is then used to absorb their remaining adjoint hypermultiplet, we find quivers which experience a different type of BGS enhancement:
\begin{equation} \label{eq:A to B enhancement}
    \prod_{i} G_i \times \textcolor{blue}{SU(n)} \times \textcolor{blue}{U(1)} \rightarrow \prod_{i} G_i \times \textcolor{blue}{SO(2n+1)}
\end{equation}
for some semi-simple Lie groups $G_i$. \\

\noindent Let's see an example of how this works. We choose to take the case of adding an $A_1 = a_1$ slice to absorb one of the rank one nodes of (\ref{nminBn}). The resulting quiver will be
\begin{equation} \label{nmin Bn add A1}
     \begin{tikzpicture}[x=1cm,y=.8cm]
    \node (g1) at (-1.5,1.1) [gauge,label=right:{$2$},label={[red]left:{$a$}}] {};
    \node (g2) at (-1.5,0) [gauge,label=below:{$2$},label={[red]left:{$b$}}] {};
    \node (g3) at (-0.5,0) {$\cdots$};
    \node (g4) at (0.5,0) [gauge,label=below:{$2$}] {};
    \node (g5) at (1.3,0) [gauger,label=below:{$1$}] {};
    \node (g6) at (1.8,0) {$.$};
    \draw (g1)--(g2)--(g3)--(g4);
    \draw[transform canvas = {yshift=-1pt}] (g4)--(g5);
    \draw[transform canvas = {yshift=1pt}] (g4)--(g5);
    \draw (g1) to [out=45,in=135,looseness=10] (g1);
    \draw [decorate,decoration={brace,mirror,amplitude=6pt}] (-1.5,-0.8) --node[below=6pt] {$n-1$} (0.5,-0.8);
    \end{tikzpicture}
\end{equation}
One can verify this is a valid addition by subtracting $A_1$ (see Appendix \ref{app:quiver subtraction} for the quiver subtraction algorithm) from the two rightmost nodes of (\ref{nmin Bn add A1}) to find (\ref{nminBn}).\\

\noindent Working from (\ref{nmin Bn add A1}) we can then perform the adjoint-hypermultiplet-absorbing quiver addition of Section \ref{sec: quiver addition}, as we did before on (\ref{nminBn}) in Section \ref{sec:dn to bn}, to absorb $a$. At first glance, it seems as though we can add to either existing or empty nodes to absorb $a$, but actually doing the former doesn't yield the enhancement above. The reason for this is that the additions possible on the existing nodes are restricted to the $n=3$ case, to which we may only add an $A_1$ to unbalance the node $b$ in the necessary way to absorb $a$:
\begin{equation}
    \begin{tikzpicture}[x=1cm,y=.8cm]
    \node (g2) at (-1,0) [gauge,label=below:{$2$}] {};
    \node (g3) at (0,0) [gauge,label=below:{$4$}] {};
    \node (g4) at (1,0) [gauger,label=below:{$3$}] {};
    \node (g5) at (-1,-1) {$\textcolor{red}{b}$};
    \node (g7) at (0,-1) {$\textcolor{red}{c}$};
    \node (g6) at (1.5,0) {$.$};
    \draw (g2)--(g3);
    \draw[transform canvas = {yshift=-1pt}] (g4)--(g3);
    \draw[transform canvas = {yshift=1pt}] (g4)--(g3);
    \end{tikzpicture}
\end{equation}
But if we perform quiver subtraction on this quiver to find the Hasse diagram, we will find ourselves performing three $A_1$ subtractions before arriving at the magnetic quiver for the sub-regular nilpotent orbit of $G_2$, the second quiver in Table \ref{tab:adj hyper Qs}. As a result, here we actually expect a factor of the BGS to be enhanced to $G_2$, rather than $SO(2n+1)$. This is indeed confirmed upon computation of the Hilbert series, and actually this quiver appears later in our classification, in Table \ref{tab:RHS sl sr G2} of Section \ref{sec:enhancement to G2}. This result is actually more general: 
if $Q_{g_k}$ is obtained by adding some affine slice $g_k$ to (\ref{nminBn}) to absorb one of the rank one nodes, as shown above for $g_k=A_1$, then any quiver derived from this by using quiver addition on its existing nodes to absorb $a$ will exhibit an EGS of type $G_2$ as opposed to $SO(2n+1)$. This is because the only way to add to the existing nodes of $Q_{g_k}$ to absorb $a$ is by adding $g_k$ when $n=3$, so upon subtraction, we will be subtracting $g_k$ \textit{three} times, rather than two, and so an adjoint hypermultiplet of rank three will arise from rebalancing. This explains why we are left with the sub-regular nilpotent orbit of $G_2$, and why we have enhancement to $G_2$ symmetry from an apparent $SU(3)$.\\

\noindent Although adding to existing nodes gives us nothing new, adding to the empty nodes on the left hand side of $b$ in (\ref{nmin Bn add A1}) does yield quivers who experience the enhancement (\ref{eq:A to B enhancement}). Again the slices that can be validly added are affine slices via a rank one long node. The resulting quivers for this process are listed in Tables \ref{tab:LHS sl nmin Bn add A1} and \ref{tab:LHS nsl nmin Bn add A1}. \\

\begin{table}[htbp!]
    \centering
    
 \hspace*{-2.25cm}\begin{tabular}{|c|c|c|} \hline
 
Added Slice & Quiver & Global Symmetry \\

\hline

$\begin{array}{c}
	a_k\\
	\\
	k \geq 1
	\end{array}$ &
$\raisebox{-.5\height}{\begin{tikzpicture}[x=1cm,y=.8cm]
    \node (g7) at (-3,0) [gauge,label=below:{$2$}] {};
    \node (g6) at (-2,0) {$\cdots$};
    \node (g1) at (-1,0) [gauger,label=below:{$2$}] {};
    \node (g8) at (-2,1) [gauge,label=above:{$2$}] {};
    \node (g2) at (0,0) [gauge,label=below:{$2$}] {};
    \node (g3) at (1,0) {$\cdots$};
    \node (g4) at (2,0) [gauge,label=below:{$2$}] {};
    \node (g5) at (3,0) [gauger,label=below:{$1$}] {};
    \draw (g8)--(g7)--(g6)--(g1)--(g2)--(g3)--(g4);
    \draw (g8)--(g1);
    \draw[transform canvas = {yshift=-1pt}] (g4)--(g5);
    \draw[transform canvas = {yshift=1pt}] (g4)--(g5);
    \draw [decorate,decoration={brace,mirror,amplitude=6pt}] (0,-0.8) --node[below=6pt] {$n-1$} (2,-0.8);
    \draw [decorate,decoration={brace,mirror,amplitude=6pt}] (-3,-0.8) --node[below=6pt] {$k$} (-1,-0.8);
    \end{tikzpicture}}$ & 
$\renewcommand{\arraystretch}{1.5}
\begin{array}{c}
SU(k+1) \times \textcolor{blue}{SU(n) \times U(1)} \\
\downarrow \\
SU(k+1) \times \textcolor{blue}{SO(2n+1)}
\end{array}$ \\

\hline

$\begin{array}{c}
	d_k\\
	\\
	k \geq 4
	\end{array}$ &
$\raisebox{-.5\height}{\begin{tikzpicture}[x=1cm,y=.8cm]
    \node (g9) at (-4,0) [gauge,label=below:{$2$}] {};
    \node (g7) at (-3,0) [gauge,label=below:{$4$}] {};
    \node (g10) at (-3,1) [gauge,label=above:{$2$}] {};
    \node (g6) at (-2,0) {$\cdots$};
    \node (g1) at (-1,0) [gauge,label=below:{$4$}] {};
    \node (g11) at (-1,1) [gauge,label=above:{$2$}] {};
    \node (g8) at (0,0) [gauger,label=below:{$2$}] {};
    \node (g2) at (1,0) [gauge,label=below:{$2$}] {};
    \node (g3) at (2,0) {$\cdots$};
    \node (g4) at (3,0) [gauge,label=below:{$2$}] {};
    \node (g5) at (4,0) [gauger,label=below:{$1$}] {};
    \draw (g9)--(g7)--(g6)--(g1)--(g8)--(g2)--(g3)--(g4);
    \draw (g10)--(g7);
    \draw (g11)--(g1);
    \draw[transform canvas = {yshift=-1pt}] (g4)--(g5);
    \draw[transform canvas = {yshift=1pt}] (g4)--(g5);
    \draw [decorate,decoration={brace,mirror,amplitude=6pt}] (1,-0.8) --node[below=6pt] {$n-1$} (3,-0.8);
    \draw [decorate,decoration={brace,mirror,amplitude=6pt}] (-3,-0.8) --node[below=6pt] {$k-3$} (-1,-0.8);
    \end{tikzpicture}}$ & 
$\renewcommand{\arraystretch}{1.5}
\begin{array}{c}
SO(2k) \times \textcolor{blue}{SU(n) \times U(1)} \\
\downarrow \\
SO(2k) \times \textcolor{blue}{SO(2n+1)}
\end{array}$ \\

\hline

$\begin{array}{c}
	e_6
	\end{array}$ &
$\raisebox{-.5\height}{\begin{tikzpicture}[x=1cm,y=.8cm]
    \node (g9) at (-4,0) [gauge,label=below:{$2$}] {};
    \node (g7) at (-3,0) [gauge,label=below:{$4$}] {};
    \node (g10) at (-2,1) [gauge,label=left:{$4$}] {};
    \node (g6) at (-2,0) [gauge,label=below:{$6$}] {};
    \node (g1) at (-1,0) [gauge,label=below:{$4$}] {};
    \node (g11) at (-2,2) [gauge,label=left:{$2$}] {};
    \node (g8) at (0,0) [gauger,label=below:{$2$}] {};
    \node (g2) at (1,0) [gauge,label=below:{$2$}] {};
    \node (g3) at (2,0) {$\cdots$};
    \node (g4) at (3,0) [gauge,label=below:{$2$}] {};
    \node (g5) at (4,0) [gauger,label=below:{$1$}] {};
    \draw (g9)--(g7)--(g6)--(g1)--(g8)--(g2)--(g3)--(g4);
    \draw (g11)--(g10)--(g6);
    \draw[transform canvas = {yshift=-1pt}] (g4)--(g5);
    \draw[transform canvas = {yshift=1pt}] (g4)--(g5);
    \draw [decorate,decoration={brace,mirror,amplitude=6pt}] (1,-0.8) --node[below=6pt] {$n-1$} (3,-0.8);
    \end{tikzpicture}}$ & 
$\renewcommand{\arraystretch}{1.5}
\begin{array}{c}
E_6 \times \textcolor{blue}{SU(n) \times U(1)} \\
\downarrow \\
E_6 \times \textcolor{blue}{SO(2n+1)}
\end{array}$ \\

\hline

$\begin{array}{c}
	e_7
	\end{array}$ &
$\raisebox{-.5\height}{\begin{tikzpicture}[x=1cm,y=.8cm]
    \node (g12) at (-5,0) [gauge,label=below:{$2$}] {};
    \node (g9) at (-4,0) [gauge,label=below:{$4$}] {};
    \node (g7) at (-3,0) [gauge,label=below:{$6$}] {};
    \node (g10) at (-2,1) [gauge,label=above:{$4$}] {};
    \node (g6) at (-2,0) [gauge,label=below:{$8$}] {};
    \node (g1) at (-1,0) [gauge,label=below:{$6$}] {};
    \node (g11) at (0,0) [gauge,label=below:{$4$}] {};
    \node (g8) at (1,0) [gauger,label=below:{$2$}] {};
    \node (g2) at (2,0) [gauge,label=below:{$2$}] {};
    \node (g3) at (3,0) {$\cdots$};
    \node (g4) at (4,0) [gauge,label=below:{$2$}] {};
    \node (g5) at (5,0) [gauger,label=below:{$1$}] {};
    \draw (g12)--(g9)--(g7)--(g6)--(g1)--(g11)--(g8)--(g2)--(g3)--(g4);
    \draw (g10)--(g6);
    \draw[transform canvas = {yshift=-1pt}] (g4)--(g5);
    \draw[transform canvas = {yshift=1pt}] (g4)--(g5);
    \draw [decorate,decoration={brace,mirror,amplitude=6pt}] (2,-0.8) --node[below=6pt] {$n-1$} (4,-0.8);
    \end{tikzpicture}}$ & 
$\renewcommand{\arraystretch}{1.5}
\begin{array}{c}
E_7 \times \textcolor{blue}{SU(n) \times U(1)} \\
\downarrow \\
E_7 \times \textcolor{blue}{SO(2n+1)}
\end{array}$ \\

\hline

$\begin{array}{c}
	e_8
	\end{array}$ &
$\raisebox{-.5\height}{\begin{tikzpicture}[x=1cm,y=.8cm]
    \node (g12) at (-5.5,0) [gauge,label=below:{$4$}] {};
    \node (g9) at (-4.5,0) [gauge,label=below:{$8$}] {};
    \node (g7) at (-3.5,0) [gauge,label=below:{$12$}] {};
    \node (g10) at (-3.5,1) [gauge,label=above:{$6$}] {};
    \node (g6) at (-2.5,0) [gauge,label=below:{$10$}] {};
    \node (g1) at (-1.5,0) [gauge,label=below:{$8$}] {};
    \node (g11) at (-0.5,0) [gauge,label=below:{$6$}] {};
    \node (g13) at (0.5,0) [gauge,label=below:{$4$}] {};
    \node (g8) at (1.5,0) [gauger,label=below:{$2$}] {};
    \node (g2) at (2.5,0) [gauge,label=below:{$2$}] {};
    \node (g3) at (3.5,0) {$\cdots$};
    \node (g4) at (4.5,0) [gauge,label=below:{$2$}] {};
    \node (g5) at (5.5,0) [gauger,label=below:{$1$}] {};
    \draw (g12)--(g9)--(g7)--(g6)--(g1)--(g11)--(g13)--(g8)--(g2)--(g3)--(g4);
    \draw (g10)--(g7);
    \draw[transform canvas = {yshift=-1pt}] (g4)--(g5);
    \draw[transform canvas = {yshift=1pt}] (g4)--(g5);
    \draw [decorate,decoration={brace,mirror,amplitude=6pt}] (2.5,-0.8) --node[below=6pt] {$n-1$} (4.5,-0.8);
    \end{tikzpicture}}$ & 
$\renewcommand{\arraystretch}{1.5}
\begin{array}{c}
E_8 \times \textcolor{blue}{SU(n) \times U(1)} \\
\downarrow \\
E_8 \times \textcolor{blue}{SO(2n+1)}
\end{array}$ \\

\hline

\end{tabular}
 \caption{Quivers resulting from adding all possible simply laced elementary slices to the empty nodes of (\ref{nmin Bn add A1}) to absorb $a$.}
    \label{tab:LHS sl nmin Bn add A1}
\end{table}

\begin{table}[htbp!]
    \centering
    
 \hspace*{-0.75cm}\begin{tabular}{|c|c|c|} \hline
 
Added Slice & Quiver & Global Symmetry \\

\hline

$\begin{array}{c}
	b_k\\
	\\
	k \geq 3
	\end{array}$ &
$\raisebox{-.5\height}{\begin{tikzpicture}[x=1cm,y=.8cm]
    \node (g9) at (-4,0) [gauge,label=below:{$2$}] {};
    \node (g7) at (-3,0) [gauge,label=below:{$4$}] {};
    \node (g6) at (-2,0) {$\cdots$};
    \node (g1) at (-1,0) [gauge,label=below:{$4$}] {};
    \node (g10) at (-1,1) [gauge,label=above:{$2$}] {};
    \node (g8) at (0,0) [gauger,label=below:{$2$}] {};
    \node (g2) at (1,0) [gauge,label=below:{$2$}] {};
    \node (g3) at (2,0) {$\cdots$};
    \node (g4) at (3,0) [gauge,label=below:{$2$}] {};
    \node (g5) at (4,0) [gauger,label=below:{$1$}] {};
    \draw (g7)--(g6)--(g1)--(g8)--(g2)--(g3)--(g4);
    \draw (g10)--(g1);
    \draw[transform canvas = {yshift=-1pt}] (g9)--(g7);
    \draw[transform canvas = {yshift=1pt}] (g9)--(g7);
    \draw (-3.4,0.2)--(-3.6,0)--(-3.4,-0.2);
    \draw[transform canvas = {yshift=-1pt}] (g4)--(g5);
    \draw[transform canvas = {yshift=1pt}] (g4)--(g5);
    \draw [decorate,decoration={brace,mirror,amplitude=6pt}] (1,-0.8) --node[below=6pt] {$n-1$} (3,-0.8);
    \draw [decorate,decoration={brace,mirror,amplitude=6pt}] (-3,-0.8) --node[below=6pt] {$k-2$} (-1,-0.8);
    \end{tikzpicture}}$ & 
$\renewcommand{\arraystretch}{1.5}
\begin{array}{c}
SO(2k+1) \times \textcolor{blue}{SU(n) \times U(1)} \\
\downarrow \\
SO(2k+1) \times \textcolor{blue}{SO(2n+1)}
\end{array}$ \\

\hline

$\begin{array}{c}
	c_k\\
	\\
	k \geq 2
	\end{array}$ &
$\raisebox{-.5\height}{\begin{tikzpicture}[x=1cm,y=.8cm]
    \node (g9) at (-4,0) [gauge,label=below:{$2$}] {};
    \node (g7) at (-3,0) [gauge,label=below:{$2$}] {};
    \node (g6) at (-2,0) {$\cdots$};
    \node (g1) at (-1,0) [gauge,label=below:{$2$}] {};
    \node (g8) at (0,0) [gauger,label=below:{$2$}] {};
    \node (g2) at (1,0) [gauge,label=below:{$2$}] {};
    \node (g3) at (2,0) {$\cdots$};
    \node (g4) at (3,0) [gauge,label=below:{$2$}] {};
    \node (g5) at (4,0) [gauger,label=below:{$1$}] {};
    \draw (g7)--(g6)--(g1);
    \draw (g8)--(g2)--(g3)--(g4);
    \draw[transform canvas = {yshift=-1pt}] (g9)--(g7);
    \draw[transform canvas = {yshift=1pt}] (g9)--(g7);
    \draw (-3.6,0.2)--(-3.4,0)--(-3.6,-0.2);
    \draw[transform canvas = {yshift=-1pt}] (g1)--(g8);
    \draw[transform canvas = {yshift=1pt}] (g1)--(g8);
    \draw (-0.4,0.2)--(-0.6,0)--(-0.4,-0.2);
    \draw[transform canvas = {yshift=-1pt}] (g4)--(g5);
    \draw[transform canvas = {yshift=1pt}] (g4)--(g5);
    \draw [decorate,decoration={brace,mirror,amplitude=6pt}] (1,-0.8) --node[below=6pt] {$n-1$} (3,-0.8);
    \draw [decorate,decoration={brace,mirror,amplitude=6pt}] (-3,-0.8) --node[below=6pt] {$k-1$} (-1,-0.8);
    \end{tikzpicture}}$ & 
$\renewcommand{\arraystretch}{1.5}
\begin{array}{c}
Sp(k) \times \textcolor{blue}{SU(n) \times U(1)} \\
\downarrow \\
Sp(k) \times \textcolor{blue}{SO(2n+1)}
\end{array}$ \\

\hline

$\begin{array}{c}
	f_4
	\end{array}$ &
$\raisebox{-.5\height}{\begin{tikzpicture}[x=1cm,y=.8cm]
    \node (g9) at (-4,0) [gauge,label=below:{$2$}] {};
    \node (g7) at (-3,0) [gauge,label=below:{$4$}] {};
    \node (g6) at (-2,0) [gauge,label=below:{$6$}] {};
    \node (g1) at (-1,0) [gauge,label=below:{$4$}] {};
    \node (g8) at (0,0) [gauger,label=below:{$2$}] {};
    \node (g2) at (1,0) [gauge,label=below:{$2$}] {};
    \node (g3) at (2,0) {$\cdots$};
    \node (g4) at (3,0) [gauge,label=below:{$2$}] {};
    \node (g5) at (4,0) [gauger,label=below:{$1$}] {};
    \draw (g9)--(g7);
    \draw (g6)--(g1)--(g8)--(g2)--(g3)--(g4);
    \draw[transform canvas = {yshift=-1pt}] (g7)--(g6);
    \draw[transform canvas = {yshift=1pt}] (g7)--(g6);
    \draw (-2.4,0.2)--(-2.6,0)--(-2.4,-0.2);
    \draw[transform canvas = {yshift=-1pt}] (g4)--(g5);
    \draw[transform canvas = {yshift=1pt}] (g4)--(g5);
    \draw [decorate,decoration={brace,mirror,amplitude=6pt}] (1,-0.8) --node[below=6pt] {$n-1$} (3,-0.8);
    \end{tikzpicture}}$ & 
$\renewcommand{\arraystretch}{1.5}
\begin{array}{c}
F_4 \times \textcolor{blue}{SU(n) \times U(1)} \\
\downarrow \\
F_4 \times \textcolor{blue}{SO(2n+1)}
\end{array}$ \\

\hline

$\begin{array}{c}
	g_2\\
	\end{array}$ &
$\raisebox{-.5\height}{\begin{tikzpicture}[x=1cm,y=.8cm]
    \node (g7) at (-3,0) [gauge,label=below:{$2$}] {};
    \node (g1) at (-2,0) [gauge,label=below:{$4$}] {};
    \node (g8) at (-1,0) [gauger,label=below:{$2$}] {};
    \node (g2) at (0,0) [gauge,label=below:{$2$}] {};
    \node (g3) at (1,0) {$\cdots$};
    \node (g4) at (2,0) [gauge,label=below:{$2$}] {};
    \node (g5) at (3,0) [gauger,label=below:{$1$}] {};
    \draw (g1)--(g8)--(g2)--(g3)--(g4);
    \draw[transform canvas = {yshift=0pt}] (g7)--(g1);
    \draw[transform canvas = {yshift=2pt}] (g7)--(g1);
    \draw[transform canvas = {yshift=-2pt}] (g7)--(g1);
    \draw (-2.4,0.2)--(-2.6,0)--(-2.4,-0.2);
    \draw[transform canvas = {yshift=-1pt}] (g4)--(g5);
    \draw[transform canvas = {yshift=1pt}] (g4)--(g5);
    \draw [decorate,decoration={brace,mirror,amplitude=6pt}] (0,-0.8) --node[below=6pt] {$n-1$} (2,-0.8);
    \end{tikzpicture}}$ & 
$\renewcommand{\arraystretch}{1.5}
\begin{array}{c}
G_2 \times \textcolor{blue}{SU(n) \times U(1)} \\
\downarrow \\
G_2 \times \textcolor{blue}{SO(2n+1)}
\end{array}$ \\

\hline

\end{tabular}
 \caption{Quivers resulting from adding all possible non-simply laced elementary slice to the empty nodes of (\ref{nmin Bn add A1}) to absorb $a$.}
    \label{tab:LHS nsl nmin Bn add A1}
\end{table}

\noindent Recall that this was just one example of performing a single addition on (\ref{nminBn}) (that of adding an $A_1$) to find quivers with a new type of symmetry enhancement. This enhancement is actually exhibited by infinitely many quivers: those obtained by performing further additions on the quiver (\ref{nmin Bn add A1}), or those constructed by adding a single slice other than $A_1$ to (\ref{nminBn}), and then adding to this and so on. We restrict the list of examples provided to just the $A_1$ case, as we (obviously!) cannot list them all, and these illustrate the enhancement. In conclusion, these infinitely many quivers, combined with those in Tables \ref{tab:adding to RHS nmin Bn}, \ref{tab:sl addition nmin B3}, \ref{tab:nsl addition nmin B3}, \ref{tab:addition nmin B2}, \ref{tab:lhs sl addition nmin Bn}, \ref{tab:lhs nsl addition nmin Bn}, \ref{tab:LHS sl nmin Bn add A1} and \ref{tab:LHS nsl nmin Bn add A1} concludes the families of quivers derived from quiver addition on the next to minimal nilpotent orbit of $B_n$ (\ref{nminBn}) whose global symmetry contains a factor which is enhanced to $SO(2n+1)$ from that predicted by the BGS.

\clearpage

\section{\texorpdfstring{Enhancement to $G_2$}{Enhancement to G2}}
\label{sec:enhancement to G2}
We move on to focus on the second quiver in Table \ref{tab:adj hyper Qs}, the sub-regular nilpotent orbit of $G_2$
\begin{equation}\label{sr G2}
    \begin{tikzpicture}[x=1cm,y=.8cm]
        \node (g0) at (-0.5,1.1) [gauge, label=right:{$3$},label={[red]left:{$a$}}] {};
        \node (g1) at (-0.5,0) [gauge,label=below:{$2$},label={[red]left:{$b$}}] {};
        \node (g2) at (0.5,0) [gauge, label=below:{$1$}] {};
        \node (g3) at (1,0) {$.$};
        \draw (g0)--(g1)--(g2);
        \draw (g0) to [out=45,in=135,looseness=10] (g0);
    \end{tikzpicture}
\end{equation}
The Hasse diagram of this quiver can be inferred from the Hasse diagram of the $G_2$ nilpotent cone:
\begin{equation} \label{HD srG2}
    \begin{tikzpicture}
    \node (1) [hasse] at (0,0) {};
    \node (2) [hasse] at (0,-1) {};
    \node (3) [hasse] at (0,-2) {};
    \node (4) [hasse] at (0,-3) {};
    \node (5) at (0.5,-1.5) {$.$};
    \draw (1) edge [] node[label=left:$d_4$] {} (2);
    \draw (2) edge [] node[label=left:$m$] {} (3);
    \draw (3) edge [] node[label=left:$A_1$] {} (4);
    \end{tikzpicture}
\end{equation}
It involves a non-normal slice $m$, which at present is not fully understood. In all quivers listed in this section, which are a result of adding the same elementary slice three subsequent times to the same node of (\ref{sr G2}), the BGS enhances as
\begin{equation}
    \prod_{i} G_i \times \textcolor{blue}{SU(3)} \rightarrow \prod_{i} G_i \times \textcolor{blue}{G_2},
\end{equation}
for some semi-simple Lie groups $G_i$. The enhancements in global symmetry of the quivers in this section were also found via a different approach in \cite{Mekareeya:2017jgc}.\\

\noindent As before, following the quiver addition algorithm from Section \ref{sec: quiver addition}, the only possible slices we can add to the node $c$ are the affine quivers via a rank one long node. Here, we can see that $c$ could either be the existing rank one node to the right of $b$ in (\ref{sr G2})
\begin{equation}\label{cnode existing sr G2}
    \begin{tikzpicture}[x=1cm,y=.8cm]
        \node (g0) at (-0.5,1.1) [gauge, label=right:{$3$},label={[red]left:{$a$}}] {};
        \node (g1) at (-0.5,0) [gauge,label=below:{$2$},label={[red]left:{$b$}}] {};
        \node (g2) at (0.5,0) [gauge, label=below:{$1$},label={[red]right:{$c$}}] {};
        \node (g3) at (1.25,0) {$,$};
        \draw (g0)--(g1)--(g2);
        \draw (g0) to [out=45,in=135,looseness=10] (g0);
    \end{tikzpicture}
\end{equation}
or the empty node to its left
\begin{equation}\label{cnode empty sr G2}
    \begin{tikzpicture}[x=1cm,y=.8cm]
        \node (g0) at (-0.5,1.1) [gauge, label=right:{$3$},label={[red]left:{$a$}}] {};
        \node (g1) at (-0.5,0) [gauge,label=below:{$2$},label={[red]45:{$b$}}] {};
        \node (g2) at (0.5,0) [gauge, label=below:{$1$}] {};
        \node (g4) at (-1.5,0) [emptygauge,label={[gray]below:{$0$}},label={[red]left:{$c$}}] {};
        \node (g3) at (1,0) {$.$};
        \draw (g0)--(g1)--(g2);
        \draw [dashed, gray] (g4)--(g1);
        \draw (g0) to [out=45,in=135,looseness=10] (g0);
    \end{tikzpicture}
\end{equation}
As before in Section \ref{sec:enhancement to Bn}, we call these ways to absorb $a$ in (\ref{sr G2}) ``adding to existing nodes" or ``adding to empty nodes" respectively. In these cases, since the set of nodes connected to $c$ is either $\{b\}$ or the empty set, the condition that the affine quivers we add must contain $c$ and its connected nodes as a subset is trivial, and so all affine quivers are allowed. The results for these additions are given in Tables \ref{tab:RHS sl sr G2} and \ref{tab:RHS nsl sr G2} in the case where $c$ is taken to be an existing node, as in (\ref{cnode existing sr G2}), and Tables \ref{tab:LHS sl sr G2} and \ref{tab:LHS nsl sr G2} in the case where $c$ is taken to be an empty node, as in (\ref{cnode empty sr G2}). As in Section \ref{sec:enhancement to Bn}, in the cases where the quivers listed depend on a parameter, the EGS has been verified via Hilbert series computation for small values of the parameter. In the cases with no parameter dependence, every case except for that of $e_8$ has been verified.\\

\clearpage

\begin{table}[htbp!]
    \centering
    
 \hspace*{-1cm}\begin{tabular}{|c|c|c|} \hline
 
Added Slice & Quiver & Global Symmetry \\

\hline

$\begin{array}{c}
	a_k\\
	\\
	k\geq 1
	\end{array}$ &
$\raisebox{-.5\height}{\begin{tikzpicture}[x=1cm,y=.8cm]
\node (g1) at (-1,0) [gauge,label=below:{$2$}] {};
\node (g2) at (0,0) [gauge,label=below:{$4$}] {};
\node (g4) at (1,1) [gauger,label=right:{$3$}] {};
\node (g5) at (1,0) {$\vdots$};
\node (g6) at (1,-1) [gauger,label=right:{$3$}] {};
\draw (g1)--(g2)--(g4)--(g5)--(g6)--(g2);
\draw [decorate,decoration={brace,amplitude=6pt}] (1.6,1) --node[right=6pt] {$k$} (1.6,-1);
\end{tikzpicture}}$ & 
$\renewcommand{\arraystretch}{1.5}
\begin{array}{c}
\textcolor{blue}{SU(3)} \times U(1)^l \times SU(k-1) \\
\downarrow \\
\textcolor{blue}{G_2} \times U(1)^l \times SU(k-1),\\
l=
\begin{cases}
    1, & k\geq 2\\
    0, & k=1
\end{cases}
\end{array}$ \\
\hline

$\begin{array}{c}
	d_k\\
	\\
	k\geq 4
	\end{array}$ &
$\raisebox{-.5\height}{\begin{tikzpicture}[x=1cm,y=.8cm]
\node (g0) at (-2.5,0) [gauge,label=below:{$2$}] {};
\node (g1) at (-1.5,0) [gauge,label=below:{$4$}] {};
\node (g2) at (-0.5,0) [gauger,label=below:{$6$}] {};
\node (g3) at (0.5,0) {$\cdots$};
\node (g4) at (1.5,0) [gauge,label=below:{$6$}] {};
\node (g5) at (2.3,0.8) [gauge,label=below:{$3$}] {};
\node (g12) at (2.3,-0.8) [gauge,label=below:{$3$}] {};
\node (g13) at (-0.5,1) [gauge,label=above:{$3$}] {};
\draw (g0)--(g1)--(g2)--(g3)--(g4)--(g5);
\draw (g12)--(g4);
\draw (g13)--(g2);
\draw [decorate,decoration={brace,mirror,amplitude=6pt}] (-0.5,-0.8) --node[below=6pt] {$k-3$} (1.5,-0.8);
\end{tikzpicture}}$ & 
$\renewcommand{\arraystretch}{1.5}
\begin{array}{c}
\textcolor{blue}{SU(3)} \times SU(2) \times SO(2k-4) \\
\downarrow \\
\textcolor{blue}{G_2} \times SU(2) \times SO(2k-4)
\end{array}$ \\
\hline

$\begin{array}{c}
	e_6\\
	\end{array}$ &
$\raisebox{-.5\height}{\begin{tikzpicture}[x=1cm,y=.8cm]
\node (g0) at (-2.5,0) [gauge,label=below:{$2$}] {};
\node (g1) at (-1.5,0) [gauge,label=below:{$4$}] {};
\node (g2) at (-0.5,0) [gauger,label=below:{$6$}] {};
\node (g3) at (0.5,0) [gauge,label=below:{$9$}] {};
\node (g4) at (1.5,0) [gauge,label=below:{$6$}] {};
\node (g5) at (2.5,0) [gauge,label=below:{$3$}] {};
\node (g12) at (0.5,1) [gauge,label=left:{$6$}] {};
\node (g13) at (0.5,2) [gauge,label=left:{$3$}] {};
\draw (g0)--(g1)--(g2)--(g3)--(g4)--(g5);
\draw (g13)--(g12)--(g3);
\end{tikzpicture}}$ & 
$\renewcommand{\arraystretch}{1.5}
\begin{array}{c}
\textcolor{blue}{SU(3)} \times SU(6) \\
\downarrow \\
\textcolor{blue}{G_2} \times SU(6)
\end{array}$ \\
\hline

$\begin{array}{c}
	e_7\\
	\end{array}$ &
$\raisebox{-.5\height}{\begin{tikzpicture}[x=1cm,y=.8cm]
\node (g1) at (-3.5,0) [gauge,label=below:{$2$}] {};
\node (g2) at (-2.5,0) [gauge,label=below:{$4$}] {};
\node (g3) at (-1.5,0) [gauger,label=below:{$6$}] {};
\node (g4) at (-0.5,0) [gauge,label=below:{$9$}] {};
\node (g5) at (0.5,0) [gauge,label=below:{$12$}] {};
\node (g6) at (1.5,0) [gauge,label=below:{$9$}] {};
\node (g7) at (2.5,0) [gauge,label=below:{$6$}] {};
\node (g8) at (3.5,0) [gauge,label=below:{$3$}] {};
\node (g9) at (0.5,1) [gauge,label=right:{$6$}] {};
\draw (g1)--(g2)--(g3)--(g4)--(g5)--(g6)--(g7)--(g8);
\draw (g5)--(g9);
\end{tikzpicture}}$ & 
$\renewcommand{\arraystretch}{1.5}
\begin{array}{c}
\textcolor{blue}{SU(3)} \times SO(12) \\
\downarrow \\
\textcolor{blue}{G_2} \times SO(12)
\end{array}$ \\
\hline

$\begin{array}{c}
	e_8\\
	\end{array}$ &
$\raisebox{-.5\height}{\begin{tikzpicture}[x=1cm,y=.8cm]
\node (g1) at (-4,0) [gauge,label=below:{$2$}] {};
\node (g2) at (-3,0) [gauge,label=below:{$4$}] {};
\node (g3) at (-2,0) [gauger,label=below:{$6$}] {};
\node (g4) at (-1,0) [gauge,label=below:{$9$}] {};
\node (g5) at (0,0) [gauge,label=below:{$12$}] {};
\node (g6) at (1,0) [gauge,label=below:{$15$}] {};
\node (g7) at (2,0) [gauge,label=below:{$18$}] {};
\node (g8) at (3,0) [gauge,label=below:{$12$}] {};
\node (g9) at (4,0) [gauge,label=below:{$6$}] {};
\node (g10) at (2,1) [gauge,label=right:{$9$}] {};
\draw (g1)--(g2)--(g3)--(g4)--(g5)--(g6)--(g7)--(g8)--(g9);
\draw (g7)--(g10);
\end{tikzpicture}}$ & 
$\renewcommand{\arraystretch}{1.5}
\begin{array}{c}
\textcolor{blue}{SU(3)} \times E_7 \\
\downarrow \\
\textcolor{blue}{G_2} \times E_7
\end{array}$ \\
\hline

\end{tabular}
 \caption{Quivers resulting from adding all possible simply laced elementary slices to the existing nodes of sub-regular $G_2$ (\ref{sr G2}) to absorb $a$.}
    \label{tab:RHS sl sr G2}
\end{table}

\begin{table}[htbp!]
    \centering
    
 \hspace*{-0cm}\begin{tabular}{|c|c|c|} \hline
 
Added Slice & Quiver & Global Symmetry \\

\hline

$\begin{array}{c}
	b_k\\
	\\
	k \geq 3
	\end{array}$ &
$\raisebox{-.5\height}{\begin{tikzpicture}[x=1cm,y=.8cm]
\node (g0) at (-2.5,0) [gauge,label=below:{$2$}] {};
\node (g1) at (-1.5,0) [gauge,label=below:{$4$}] {};
\node (g2) at (-0.5,0) [gauger,label=below:{$6$}] {};
\node (g3) at (0.5,0) {$\cdots$};
\node (g4) at (1.5,0) [gauge,label=below:{$6$}] {};
\node (g5) at (2.5,0) [gauge,label=below:{$3$}] {};
\node (g11) at (-0.5,1) [gauge,label=above:{$3$}] {};
\draw (g0)--(g1)--(g2)--(g3)--(g4);
\draw (g2)--(g11);
\draw[transform canvas={yshift=-1.5pt}] (g5)--(g4);
\draw[transform canvas={yshift=1.5pt}] (g5)--(g4);
\draw (1.9,-0.2)--(2.1,0)--(1.9,0.2);
\draw [decorate,decoration={brace,mirror,amplitude=6pt}] (-0.5,-0.8) --node[below=6pt] {$k-2$} (1.5,-0.8);
\end{tikzpicture}}$ & 
$\renewcommand{\arraystretch}{1.5}
\begin{array}{c}
\textcolor{blue}{SU(3)} \times SU(2) \times SO(2k-3) \\
\downarrow \\
\textcolor{blue}{G_2} \times SU(2) \times SO(2k-3)
\end{array}$ \\
\hline

$\begin{array}{c}
	c_k\\
	\\
	k \geq 2
	\end{array}$ &
$\raisebox{-.5\height}{\begin{tikzpicture}[x=1cm,y=.8cm]
\node (g0) at (-2.5,0) [gauge,label=below:{$2$}] {};
\node (g1) at (-1.5,0) [gauge,label=below:{$4$}] {};
\node (g2) at (-0.5,0) [gauger,label=below:{$3$}] {};
\node (g3) at (0.5,0) {$\cdots$};
\node (g4) at (1.5,0) [gauge,label=below:{$3$}] {};
\node (g5) at (2.5,0) [gauge,label=below:{$3$}] {};
\draw (g0)--(g1);
\draw (g2)--(g3)--(g4);
\draw[transform canvas={yshift=-1.5pt}] (g1)--(g2);
\draw[transform canvas={yshift=1.5pt}] (g1)--(g2);
\draw (-1.1,0.2)--(-0.9,0)--(-1.1,-0.2);
\draw[transform canvas={yshift=-1.5pt}] (g5)--(g4);
\draw[transform canvas={yshift=1.5pt}] (g5)--(g4);
\draw (2.1,-0.2)--(1.9,0)--(2.1,0.2);
\draw [decorate,decoration={brace,mirror,amplitude=6pt}] (-0.5,-0.8) --node[below=6pt] {$k-1$} (1.5,-0.8);
\end{tikzpicture}}$ & 
$\renewcommand{\arraystretch}{1.5}
\begin{array}{c}
\textcolor{blue}{SU(3)} \times Sp(k-1) \\
\downarrow \\
\textcolor{blue}{G_2} \times Sp(k-1)
\end{array}$ \\
\hline

$\begin{array}{c}
	f_4\\
	\end{array}$ &
$\raisebox{-.5\height}{\begin{tikzpicture}[x=1cm,y=.8cm]
\node (g0) at (-2.5,0) [gauge,label=below:{$2$}] {};
\node (g1) at (-1.5,0) [gauge,label=below:{$4$}] {};
\node (g2) at (-0.5,0) [gauger,label=below:{$6$}] {};
\node (g3) at (0.5,0) [gauge,label=below:{$9$}] {};
\node (g4) at (1.5,0) [gauge,label=below:{$6$}] {};
\node (g5) at (2.5,0) [gauge,label=below:{$3$}] {};
\draw (g0)--(g1)--(g2)--(g3);
\draw (g4)--(g5);
\draw[transform canvas={yshift=-1.5pt}] (g3)--(g4);
\draw[transform canvas={yshift=1.5pt}] (g3)--(g4);
\draw (0.9,-0.2)--(1.1,0)--(0.9,0.2);
\end{tikzpicture}}$ & 
$\renewcommand{\arraystretch}{1.5}
\begin{array}{c}
\textcolor{blue}{SU(3)} \times Sp(3) \\
\downarrow \\
\textcolor{blue}{G_2} \times Sp(3)
\end{array}$ \\
\hline

$\begin{array}{c}
    g_2
    \end{array}$ &
$\raisebox{-.5\height}{\begin{tikzpicture}[x=1cm,y=.8cm]
\node (g1) at (-1.5,0) [gauge,label=below:{$2$}] {};
\node (g2) at (-0.5,0) [gauge,label=below:{$4$}] {};
\node (g3) at (0.5,0) [gauger,label=below:{$6$}] {};
\node (g4) at (1.5,0) [gauge,label=below:{$3$}] {};
\draw (g1)--(g2)--(g3);
\draw[transform canvas = {yshift=-2pt}] (g4)--(g3);
\draw[transform canvas = {yshift=0pt}] (g4)--(g3);
\draw[transform canvas = {yshift =2pt}] (g4)--(g3);
\draw (0.9,0.2)--(1.1,0)--(0.9,-0.2);
\end{tikzpicture}}$ &
$\renewcommand{\arraystretch}{1.5}
\begin{array}{c}
\textcolor{blue}{SU(3)} \times SU(2)\\
\downarrow \\
\textcolor{blue}{G_2} \times SU(2)
\end{array}$\\
\hline

\end{tabular}
 \caption{Quivers resulting from adding all possible non-simply laced elementary slices to the existing nodes of sub-regular $G_2$ (\ref{sr G2}) to absorb $a$.}
    \label{tab:RHS nsl sr G2}
\end{table}

\begin{table}[htbp!]
    \centering
    
 \hspace*{-0.5cm}\begin{tabular}{|c|c|c|} \hline
 
Added Slice & Quiver & Global Symmetry \\

\hline

$\begin{array}{c}
	a_k\\
	\\
	k\geq 1
	\end{array}$ &
$\raisebox{-.5\height}{\begin{tikzpicture}[x=1cm,y=.8cm]
\node (g1) at (-1.5,1) [gauge,label=left:{$3$}] {};
\node (g2) at (-1.5,0) {$\vdots$};
\node (g3) at (-1.5,-1) [gauge,label=left:{$3$}] {};
\node (g4) at (-0.5,0) [gauger,label=below:{$3$}] {};
\node (g5) at (0.5,0) [gauge,label=below:{$2$}] {};
\node (g6) at (1.5,0) [gauge,label=below:{$1$}] {};
\draw (g1)--(g2)--(g3)--(g4)--(g5)--(g6);
\draw (g1)--(g4);
\draw [decorate,decoration={brace,mirror,amplitude=6pt}] (-2.1,1) --node[left=6pt] {$k$} (-2.1,-1);
\end{tikzpicture}}$ & 
$\renewcommand{\arraystretch}{1.5}
\begin{array}{c}
 SU(k+1) \times \textcolor{blue}{SU(3)}\\
\downarrow \\
 SU(k+1) \times \textcolor{blue}{G_2}
\end{array}$ \\

\hline

$\begin{array}{c}
	d_k\\
	\\
	k \geq 4
	\end{array}$ &
$\raisebox{-.5\height}{\begin{tikzpicture}[x=1cm,y=.8cm]
    \node (g9) at (-3,0) [gauge,label=below:{$3$}] {};
    \node (g7) at (-2,0) [gauge,label=below:{$6$}] {};
    \node (g10) at (-2,1) [gauge,label=above:{$3$}] {};
    \node (g6) at (-1,0) {$\cdots$};
    \node (g1) at (0,0) [gauge,label=below:{$6$}] {};
    \node (g11) at (0,1) [gauge,label=above:{$3$}] {};
    \node (g8) at (1,0) [gauger,label=below:{$3$}] {};
    \node (g2) at (2,0) [gauge,label=below:{$2$}] {};
    \node (g5) at (3,0) [gauge,label=below:{$1$}] {};
    \draw (g9)--(g7)--(g6)--(g1)--(g8)--(g2)--(g5);
    \draw (g10)--(g7);
    \draw (g11)--(g1);
    \draw [decorate,decoration={brace,mirror,amplitude=6pt}] (-2,-0.8) --node[below=6pt] {$k-3$} (0,-0.8);
    \end{tikzpicture}}$ & 
$\renewcommand{\arraystretch}{1.5}
\begin{array}{c}
 SO(2k) \times \textcolor{blue}{SU(3)}\\
\downarrow \\
 SO(2k) \times \textcolor{blue}{G_2}
\end{array}$ \\

\hline

$\begin{array}{c}
	e_6\\
	\end{array}$ &
$\raisebox{-.5\height}{\begin{tikzpicture}[x=1cm,y=.8cm]
    \node (g9) at (-3,0) [gauge,label=below:{$3$}] {};
    \node (g7) at (-2,0) [gauge,label=below:{$6$}] {};
    \node (g10) at (-1,2) [gauge,label=left:{$3$}] {};
    \node (g6) at (-1,0) [gauge,label=below:{$9$}] {};
    \node (g1) at (0,0) [gauge,label=below:{$6$}] {};
    \node (g11) at (-1,1) [gauge,label=left:{$6$}] {};
    \node (g8) at (1,0) [gauger,label=below:{$3$}] {};
    \node (g2) at (2,0) [gauge,label=below:{$2$}] {};
    \node (g5) at (3,0) [gauge,label=below:{$1$}] {};
    \draw (g9)--(g7)--(g6)--(g1)--(g8)--(g2)--(g5);
    \draw (g10)--(g11)--(g6);
    \end{tikzpicture}}$ & 
$\renewcommand{\arraystretch}{1.5}
\begin{array}{c}
 E_6 \times \textcolor{blue}{SU(3)}\\
\downarrow \\
 E_6 \times \textcolor{blue}{G_2}
\end{array}$ \\

\hline

$\begin{array}{c}
	e_7\\
	\end{array}$ &
$\raisebox{-.5\height}{\begin{tikzpicture}[x=1cm,y=.8cm]
     \node (g13) at (-4,0) [gauge,label=below:{$3$}] {};
     \node (g12) at (-3,0) [gauge,label=below:{$6$}] {};
    \node (g9) at (-2,0) [gauge,label=below:{$9$}] {};
    \node (g7) at (-1,0) [gauge,label=below:{$12$}] {};
    \node (g10) at (-1,1) [gauge,label=above:{$6$}] {};
    \node (g6) at (0,0) [gauge,label=below:{$9$}] {};
    \node (g11) at (1,0) [gauge,label=below:{$6$}] {};
    \node (g8) at (2,0) [gauger,label=below:{$3$}] {};
    \node (g2) at (3,0) [gauge,label=below:{$2$}] {};
    \node (g5) at (4,0) [gauge,label=below:{$1$}] {};
    \draw (g13)--(g12)--(g9)--(g7)--(g6)--(g11)--(g8)--(g2)--(g5);
    \draw (g10)--(g7);
    \end{tikzpicture}}$ & 
$\renewcommand{\arraystretch}{1.5}
\begin{array}{c}
 E_7 \times \textcolor{blue}{SU(3)}\\
\downarrow \\
 E_7 \times \textcolor{blue}{G_2}
\end{array}$ \\

\hline

$\begin{array}{c}
	e_8\\
	\end{array}$ &
$\raisebox{-.5\height}{\begin{tikzpicture}[x=1cm,y=.8cm]
    \node (g14) at (-4.5,0) [gauge,label=below:{$6$}] {};
    \node (g13) at (-3.5,0) [gauge,label=below:{$12$}] {};
    \node (g12) at (-2.5,0) [gauge,label=below:{$18$}] {};
    \node (g9) at (-1.5,0) [gauge,label=below:{$15$}] {};
    \node (g7) at (-0.5,0) [gauge,label=below:{$12$}] {};
    \node (g10) at (-2.5,1) [gauge,label=above:{$9$}] {};
    \node (g6) at (0.5,0) [gauge,label=below:{$9$}] {};
    \node (g11) at (1.5,0) [gauge,label=below:{$6$}] {};
    \node (g8) at (2.5,0) [gauger,label=below:{$3$}] {};
    \node (g2) at (3.5,0) [gauge,label=below:{$2$}] {};
    \node (g5) at (4.5,0) [gauge,label=below:{$1$}] {};
    \draw (g14)--(g13)--(g12)--(g9)--(g7)--(g6)--(g11)--(g8)--(g2)--(g5);
    \draw (g10)--(g12);
    \end{tikzpicture}}$ & 
$\renewcommand{\arraystretch}{1.5}
\begin{array}{c}
 E_8 \times \textcolor{blue}{SU(3)}\\
\downarrow \\
 E_8 \times \textcolor{blue}{G_2}
\end{array}$ \\

\hline

\end{tabular}
 \caption{Quivers resulting from adding all possible simply laced elementary slices to the empty nodes of sub-regular $G_2$ (\ref{sr G2}) to absorb $a$.}
    \label{tab:LHS sl sr G2}
\end{table}

\begin{table}[htbp!]
    \centering
    
 \hspace*{-0.5cm}\begin{tabular}{|c|c|c|} \hline
 
Added Slice & Quiver & Global Symmetry \\

\hline

$\begin{array}{c}
	b_k\\
	\\
	k\geq 3
	\end{array}$ &
$\raisebox{-.5\height}{\begin{tikzpicture}[x=1cm,y=.8cm]
\node (g1) at (-3,0) [gauge,label=below:{$3$}] {};
\node (g2) at (-2,0) [gauge,label=below:{$6$}] {};
\node (g3) at (-1,0) {$\cdots$};
\node (g4) at (0,0) [gauge,label=below:{$6$}] {};
\node (g5) at (1,0) [gauger,label=below:{$3$}] {};
\node (g12) at (2,0) [gauge,label=below:{$2$}] {};
\node (g6) at (3,0) [gauge,label=below:{$1$}] {};
\node (g7) at (0,1) [gauge,label=above:{$3$}] {};
\draw (g2)--(g3)--(g4)--(g5)--(g12)--(g6);
\draw (g7)--(g4);
\draw[transform canvas={yshift=-1.5pt}] (g1)--(g2);
\draw[transform canvas={yshift=1.5pt}] (g1)--(g2);
\draw (-2.4,0.2)--(-2.6,0)--(-2.4,-0.2);
\draw [decorate,decoration={brace,mirror,amplitude=6pt}] (-2,-0.8) --node[below=6pt] {$k-2$} (0,-0.8);
\end{tikzpicture}}$ & 
$\renewcommand{\arraystretch}{1.5}
\begin{array}{c}
 SO(2k+1) \times \textcolor{blue}{SU(3)}\\
\downarrow \\
 SO(2k+1) \times \textcolor{blue}{G_2}
\end{array}$ \\
\hline

$\begin{array}{c}
	c_k\\
	\\
	k\geq 2
	\end{array}$ &
$\raisebox{-.5\height}{\begin{tikzpicture}[x=1cm,y=.8cm]
\node (g1) at (-3,0) [gauge,label=below:{$3$}] {};
\node (g2) at (-2,0) [gauge,label=below:{$3$}] {};
\node (g3) at (-1,0) {$\cdots$};
\node (g4) at (0,0) [gauge,label=below:{$3$}] {};
\node (g5) at (1,0) [gauger,label=below:{$3$}] {};
\node (g12) at (2,0) [gauge,label=below:{$2$}] {};
\node (g6) at (3,0) [gauge,label=below:{$1$}] {};
\draw (g2)--(g3)--(g4);
\draw (g5)--(g12)--(g6);
\draw[transform canvas={yshift=-1.5pt}] (g1)--(g2);
\draw[transform canvas={yshift=1.5pt}] (g1)--(g2);
\draw (-2.6,0.2)--(-2.4,0)--(-2.6,-0.2);
\draw[transform canvas={yshift=-1.5pt}] (g4)--(g5);
\draw[transform canvas={yshift=1.5pt}] (g4)--(g5);
\draw (0.6,0.2)--(0.4,0)--(0.6,-0.2);
\draw [decorate,decoration={brace,mirror,amplitude=6pt}] (-2,-0.8) --node[below=6pt] {$k-1$} (0,-0.8);
\end{tikzpicture}}$ & 
$\renewcommand{\arraystretch}{1.5}
\begin{array}{c}
 Sp(k) \times \textcolor{blue}{SU(3)}\\
\downarrow \\
 Sp(k) \times \textcolor{blue}{G_2}
\end{array}$ \\
\hline

$\begin{array}{c}
	f_4\\
	\end{array}$ &
$\raisebox{-.5\height}{\begin{tikzpicture}[x=1cm,y=.8cm]
\node (g1) at (-3,0) [gauge,label=below:{$3$}] {};
\node (g2) at (-2,0) [gauge,label=below:{$6$}] {};
\node (g3) at (-1,0) [gauge,label=below:{$9$}] {};
\node (g4) at (0,0) [gauge,label=below:{$6$}] {};
\node (g5) at (1,0) [gauger,label=below:{$3$}] {};
\node (g12) at (2,0) [gauge,label=below:{$2$}] {};
\node (g6) at (3,0) [gauge,label=below:{$1$}] {};
\draw (g2)--(g1);
\draw (g3)--(g4)--(g5)--(g12)--(g6);
\draw[transform canvas={yshift=-1.5pt}] (g3)--(g2);
\draw[transform canvas={yshift=1.5pt}] (g3)--(g2);
\draw (-1.4,0.2)--(-1.6,0)--(-1.4,-0.2);
\end{tikzpicture}}$ & 
$\renewcommand{\arraystretch}{1.5}
\begin{array}{c}
 F_4 \times \textcolor{blue}{SU(3)}\\
\downarrow \\
 F_4 \times \textcolor{blue}{G_2}
\end{array}$ \\
\hline

$\begin{array}{c}
    g_2
    \end{array}$ &
$\raisebox{-.5\height}{\begin{tikzpicture}[x=1cm,y=.8cm]
\node (g1) at (-2,0) [gauge,label=below:{$3$}] {};
\node (g2) at (-1,0) [gauge,label=below:{$6$}] {};
\node (g3) at (0,0) [gauger,label=below:{$3$}] {};
\node (g4) at (1,0) [gauge,label=below:{$2$}] {};
\node (g5) at (2,0) [gauge,label=below:{$1$}] {};
\draw (g2)--(g3)--(g4)--(g5);
\draw[transform canvas = {yshift=-2pt}] (g1)--(g2);
\draw[transform canvas = {yshift=0pt}] (g1)--(g2);
\draw[transform canvas = {yshift =2pt}] (g1)--(g2);
\draw (-1.4,0.2)--(-1.6,0)--(-1.4,-0.2);
\end{tikzpicture}}$ &
$\renewcommand{\arraystretch}{1.5}
\begin{array}{c}
 G_2 \times \textcolor{blue}{SU(3)}\\
\downarrow \\
 G_2 \times \textcolor{blue}{G_2}
\end{array}$ \\
\hline

\end{tabular}
 \caption{Quivers resulting from adding all possible non-simply laced elementary slices to the empty nodes of sub-regular $G_2$ (\ref{sr G2}) to absorb $a$.}
    \label{tab:LHS nsl sr G2}
\end{table}

\clearpage

\noindent As in the previous section, because of the rank one node in (\ref{sr G2}), we can also perform a single quiver addition to absorb this node before absorbing $a$ and find yet more quivers whose Hasse diagrams contain (\ref{sr G2}) and therefore experience the $SU(3) \rightarrow G_2$ enhancement. Again, there are infinitely many such quivers one could construct. However, unlike in the case of Section \ref{A type to B type section}, here the actual enhancement won't be of a different type to that of the above quivers, and so won't be of particular interest to us. To illustrate this, consider the quiver after adding an $A_1$:
\begin{equation} \label{sr G2 add A1}
    \begin{tikzpicture}[x=1cm,y=.8cm]
        \node (g0) at (-1,1) [gauge,label=right:{$3$}] {};
        \node (g1) at (-1,0) [gauge, label=below:{$2$}] {};
        \node (g2) at (0,0) [gauger,label=below:{$1$}] {};
        \node (g3) at (1,0) [gauge,label=below:{$1$}] {};
        \node (g4) at (1.5,0) {$.$};
        \draw (g0)--(g1)--(g2);
        \draw[transform canvas = {yshift=-1.5pt}] (g3)--(g2);
        \draw[transform canvas = {yshift=1.5pt}] (g3)--(g2);
        \draw (g0) to [out=45,in=135,looseness=10] (g0);
    \end{tikzpicture}
\end{equation}
The reason that (\ref{sr G2 add A1}) gives us no new quivers of interest can be made explicit if we ungauge on the rank one unbalanced node,\footnote{Recall we always choose to ungauge on a long node, and ungauging on any long node is equivalent to ungauging on any other.} in which case it becomes
\begin{equation}
     \begin{tikzpicture}[x=1cm,y=.8cm]
        \node (g0) at (-2,1) [gauge,label=right:{$3$}] {};
        \node (g1) at (-2,0) [gauge, label=below:{$2$}] {};
        \node (g2) at (-1,0) [flavor,label=below:{$1$}] {};
        \node (g4) at (0,0) {$\times$};
        \node (g5) at (1,0) [flavor,label=below:{$2$}] {};
        \node (g3) at (2,0) [gauge,label=below:{$1$}] {};
        \node (g6) at (2.5,0) {$.$};
        \draw (g0)--(g1)--(g2);
        \draw (g5)--(g3);
        \draw (g0) to [out=45, in=135, looseness=10] (g0);
    \end{tikzpicture}
\end{equation}
The Coulomb branch moduli space of this quiver theory is then the product of $A_1 = SU(2)$ and the sub-regular nilpotent orbit of $G_2$. We don't worry about product moduli spaces, because if we understand individual ones then we can construct and know everything about the products too. Indeed, if we for example add $a_2$ thrice on to the empty nodes of (\ref{sr G2 add A1}) and again ungauge on the unbalanced rank one node, we get
\begin{equation}
    \begin{tikzpicture}[x=1cm,y=.8cm]
        \node (g0) at (-3.5,0) [gauge, label=below:{$3$}] {};
        \node (g1) at (-2.5,0) [gauger, label=below:{$3$}] {};
        \node (g2) at (-1.5,0) [gauge,label=below:{$2$}] {};
        \node (g3) at (-0.5,0) [flavor,label=below:{$1$}] {};
        \node (g6) at (0.5,0) {$\times$};
        \node (g7) at (1.5,0) [flavor,label=below:{$2$}] {};
        \node (g4) at (2.5,0) [gauge,label=below:{$1$}] {};
        \node (g5) at (-3,1) [gauge, label=above:{$3$}] {};
        \node (g8) at (3,0) {$.$};
        \draw (g5)--(g0)--(g1)--(g2)--(g3);
        \draw (g1)--(g5);
        \draw (g7)--(g4);
    \end{tikzpicture}
\end{equation}
The Coulomb branch of the quiver on the left is in the top row of Table \ref{tab:LHS sl sr G2}, and the Coulomb branch of the quiver on the right is the minimal nilpotent orbit of $A_1$, so we do indeed have nothing new. Thus Tables \ref{tab:RHS sl sr G2} through to \ref{tab:LHS nsl sr G2} conclude the list of quivers with a global symmetry factor experiencing enhancement to $G_2$ that can be found as a result of performing quiver addition on the sub-regular nilpotent orbit of $G_2$ (\ref{sr G2}).

\section{\texorpdfstring{Enhancement to $SO(2n)$}{Enhancement to Dn}}
\label{sec:enhancement to Dn}
There are two quivers in Table \ref{tab:adj hyper Qs} with $SO(2n)$ global symmetry, and so both can be used as a base to construct quivers with a global symmetry factor which enhances to $SO(2n)$. It is the BGS factor that enhances which distinguishes the results from adding to the third or fourth quiver in Table \ref{tab:adj hyper Qs}. Section \ref{sec: bn-1 to dn} addresses quivers derived from adding to the former, and Section \ref{sec: sun u1 to so2n} addresses quivers derived from adding to the latter.

\subsection{\texorpdfstring{Enhancement from $SO(2n-1)$ to $SO(2n)$}{Enhancement from Bnminusone to Dn}} \label{sec: bn-1 to dn}
The next to minimal nilpotent orbit of $D_n$ for $n \geq 2$ is given by the third quiver in Table \ref{tab:adj hyper Qs} \cite{1992math......4227B,Hanany:2018dvd,Bourget:2020bxh},
\begin{equation}\label{nminDn}
    \begin{tikzpicture}[x=1cm,y=.8cm]
    \node (g1) at (-1.5,1.1) [gauge,label=right:{$2$},label={[red]left:{$a$}}] {};
    \node (g2) at (-1.5,0) [gauge,label=below:{$2$},label={[red]left:{$b$}}] {};
    \node (g3) at (-0.5,0) {$\cdots$};
    \node (g4) at (0.5,0) [gauge,label=below:{$2$}] {};
    \node (g5) at (1.5,0) [gauge,label=below:{$1$}] {};
    \node (g6) at (2,0) {$.$};
    \draw (g1)--(g2)--(g3)--(g4);
    \draw[transform canvas = {yshift=-1.5pt}] (g5)--(g4);
    \draw[transform canvas = {yshift=1.5pt}] (g5)--(g4);
    \draw (0.9,0.2)--(1.1,0)--(0.9,-0.2);
    \draw (g1) to [out=45,in=135,looseness=10] (g1);
    \draw [decorate,decoration={brace,mirror,amplitude=6pt}] (-1.5,-0.8) --node[below=6pt] {$n-2$} (0.5,-0.8);
    \end{tikzpicture}
\end{equation}
Following the BGS algorithm for quivers with adjoint matter in Section \ref{sec:quivers with adj matter}, the global symmetry of (\ref{nminDn}) is $SO(2n)$. The Hasse diagram is 
\begin{equation} \label{HD nminDn}
    \begin{tikzpicture}
    \node (1) [hasse] at (0,0) {};
    \node (2) [hasse] at (0,-1) {};
    \node (3) [hasse] at (0,-2) {};
    \node (4) at (0.5,-1) {$,$};
    \draw (1) edge [] node[label=left:$A_1$] {} (2);
    \draw (2) edge [] node[label=left:$d_n$] {} (3);
    \end{tikzpicture}
\end{equation}
which matches this BGS prediction, and this global symmetry is confirmed by computing the Hilbert series. The quivers which can be constructed by adding to (\ref{nminDn}) to absorb $a$ experience an enhancement of their BGS given by
\begin{equation}
    \prod_{i} G_i \times \textcolor{blue}{SO(2n-1)} \rightarrow \prod_{i} G_i \times \textcolor{blue}{SO(2n)},
\end{equation}
for some semi-simple Lie groups $G_i$. The $n=2$ case needs to be treated separately than $n \geq 3$ because for $n=2$, $a$ is a short node.\\

\noindent For $n \geq 3$ and using the algorithm from Section \ref{sec: quiver addition}, we can again either add to the existing nodes of (\ref{nminDn}) by taking $c$ as the node to the right of $b$
\begin{equation}\label{cnodenminDngeq3}
    \begin{tikzpicture}[x=1cm,y=.8cm]
    \node (g1) at (-2,1.1) [gauge,label=right:{$2$},label={[red]left:{$a$}}] {};
    \node (g2) at (-2,0) [gauge,label=below:{$2$},label={[red]left:{$b$}}] {};
    \node (g2a) at (-1,0) [gauge,label=below:{$2$},label={[red]above:{$c$}}] {};
    \node (g3) at (0,0) {$\cdots$};
    \node (g4) at (1,0) [gauge,label=below:{$2$}] {};
    \node (g5) at (2,0) [gauge,label=below:{$1$}] {};
    \node (g6) at (2.5,0) {$,$};
    \draw (g1)--(g2)--(g2a)--(g3)--(g4);
    \draw[transform canvas = {yshift=-1.5pt}] (g5)--(g4);
    \draw[transform canvas = {yshift=1.5pt}] (g5)--(g4);
    \draw (1.4,0.2)--(1.6,0)--(1.4,-0.2);
    \draw (g1) to [out=45,in=135,looseness=10] (g1);
    \draw [decorate,decoration={brace,mirror,amplitude=6pt}] (-2,-0.8) --node[below=6pt] {$n-2$} (1,-0.8);
    \end{tikzpicture}
\end{equation}
or add to the empty nodes of (\ref{nminDn}) by taking $c$ to be some non-existent node to the left of $b$
\begin{equation}\label{cnode empty nminDngeq3}
    \begin{tikzpicture}[x=1cm,y=.8cm]
    \node (g0) at (-2.5,0) [emptygauge,label={[red]left:{$c$}},label={[gray]below:{$0$}}] {};
    \node (g1) at (-1.5,1.1) [gauge,label=right:{$2$},label={[red]left:{$a$}}] {};
    \node (g2) at (-1.5,0) [gauge,label=below:{$2$},label={[red]315:{$b$}}] {};
    \node (g3) at (-0.5,0) {$\cdots$};
    \node (g4) at (0.5,0) [gauge,label=below:{$2$}] {};
    \node (g5) at (1.5,0) [gauge,label=below:{$1$}] {};
    \node (g6) at (2,0) {$.$};
    \draw (g1)--(g2)--(g3)--(g4);
    \draw[transform canvas = {yshift=-1.5pt}] (g5)--(g4);
    \draw[transform canvas = {yshift=1.5pt}] (g5)--(g4);
    \draw (0.9,0.2)--(1.1,0)--(0.9,-0.2);
    \draw [dashed,gray] (g0)--(g2);
    \draw (g1) to [out=45,in=135,looseness=10] (g1);
    \draw [decorate,decoration={brace,mirror,amplitude=6pt}] (-1.5,-0.8) --node[below=6pt] {$n-2$} (0.5,-0.8);
    \end{tikzpicture}
\end{equation}
As previously, when adding to the empty nodes, there will be no restriction on the affine slices we can add. Adding to the existing nodes requires that all nodes to the right of $b$ in (\ref{nminDn}) be a subset of the slice that is added. Thus we can only add $b_k$ and $f_4$ to the existing nodes, and doing so will fix the value of $n$ for these cases. The results of all these additions may be found in Table \ref{tab:RHS nmin Dn} for the existing nodes addition, and Tables \ref{tab:lhs sl addition nmin Dn} and \ref{tab:lhs nsl addition nmin Dn} for the empty nodes addition.\\

\noindent In the case where $n=2$, we are considering the next to minimal nilpotent orbit of $D_2$,
\begin{equation} \label{nminD2}
    \begin{tikzpicture}[x=1cm,y=.8cm]
    \node (g1) at (0,1.1) [gauge,label=right:{$2$},label={[red]left:{$a$}}] {};
    \node (g2) at (0,0) [gauge,label=below:{$1$},label={[red]left:{$b$}}] {};
    \node (g3) at (0.5,0) {$,$};
    \draw[transform canvas={xshift=-1.5pt}] (g1)--(g2);
    \draw[transform canvas={xshift=1.5pt}] (g1)--(g2);
    \draw (-0.2,0.6)--(0,0.4)--(0.2,0.6);
    \draw (g1) to [out=45,in=135,looseness=10] (g1);
    \end{tikzpicture}
\end{equation}
the global symmetry of which can be computed to be $SU(2)^2$ using the Hilbert series. This quiver is different to the previous cases from Table \ref{tab:adj hyper Qs} that we've looked at thus far, in that here $a$ is connected to the other nodes in the quiver via a non-simply laced edge. Because of this, the BGS algorithm given in Section \ref{sec:quivers with adj matter} and the quiver addition algorithm given in Section \ref{sec: quiver addition} do not apply. Because $a$ is connected to $b$ via a non-simply laced edge, if we view the adjoint hypermultiplet as having arisen from a double same slice subtraction, $b$ must have been becoming unbalanced by one more each time and must have been a short node, in line with the rules for rebalancing on a short node in Step $2)a)ii)$ in Appendix \ref{app:quiver subtraction}. Thus the only quivers we can add to this are those for which there is a rank one node which is on the short end of a non-simply laced edge of multiplicity two. The only quivers for elementary slices for which this is the case are affine $b_k$, $c_k$ and $f_4$ and twisted affine $a_{k}$, $d_k$, and $e_6$. They all experience an enhancement from their BGS as
\begin{equation}
    \prod_{i} G_i \times \textcolor{blue}{SU(2)} \rightarrow \prod_{i} G_i \times \textcolor{blue}{SU(2)^2},
\end{equation}
for some semi-simple Lie groups $G_i$. The quivers derived from these additions are displayed in Table \ref{tab:adding to nmin D2}.\\

\noindent For all cases $n \geq 2$ listed across this section, unlike before in Section \ref{A type to B type section}, we cannot find more quivers with such an enhancement by performing a different quiver addition on (\ref{nminDn}) prior to absorbing the adjoint hypermultiplet as there are no long rank one nodes present.

\begin{table}[htbp!]
    \centering
    
 \hspace*{-0.5cm}\begin{tabular}{|c|c|c|} \hline
 
Added Slice & Quiver & Global Symmetry \\

\hline

$\begin{array}{c}
	b_k\\
	\\
	k = n-2
	\end{array}$ &
$\raisebox{-.5\height}{\begin{tikzpicture}[x=1cm,y=.8cm]
\node (g0) at (-2.5,0) [gauge,label=below:{$2$}] {};
\node (g1) at (-1.5,0) [gauge,label=below:{$4$}] {};
\node (g2) at (-0.5,0) [gauge,label=below:{$6$}] {};
\node (g3) at (0.5,0) {$\cdots$};
\node (g4) at (1.5,0) [gauge,label=below:{$6$}] {};
\node (g5) at (2.5,0) [gauge,label=below:{$3$}] {};
\node (g11) at (-0.5,1) [gauger,label=above:{$2$}] {};
\draw (g0)--(g1)--(g2)--(g3)--(g4);
\draw (g2)--(g11);
\draw[transform canvas={yshift=-1.5pt}] (g5)--(g4);
\draw[transform canvas={yshift=1.5pt}] (g5)--(g4);
\draw (1.9,-0.2)--(2.1,0)--(1.9,0.2);
\draw [decorate,decoration={brace,mirror,amplitude=6pt}] (-0.5,-0.8) --node[below=6pt] {$k-2=n-4$} (1.5,-0.8);
\end{tikzpicture}}$ & 
$\renewcommand{\arraystretch}{1.5}
\begin{array}{c}
\textcolor{blue}{SO(2n-1)} \\
\downarrow \\
\textcolor{blue}{SO(2n)}
\end{array}$ \\
\hline

$\begin{array}{c}
	f_4\\
	\\
	n=6
	\end{array}$ &
$\raisebox{-.5\height}{\begin{tikzpicture}[x=1cm,y=.8cm]
\node (g0) at (-2.5,0) [gauge,label=below:{$2$}] {};
\node (g1) at (-1.5,0) [gauge,label=below:{$4$}] {};
\node (g2) at (-0.5,0) [gauge,label=below:{$6$}] {};
\node (g3) at (0.5,0) [gauge,label=below:{$8$}] {};
\node (g4) at (1.5,0) [gauge,label=below:{$5$}] {};
\node (g5) at (2.5,0) [gauger,label=below:{$2$}] {};
\draw (g0)--(g1)--(g2)--(g3);
\draw (g4)--(g5);
\draw[transform canvas={yshift=-1.5pt}] (g3)--(g4);
\draw[transform canvas={yshift=1.5pt}] (g3)--(g4);
\draw (0.9,-0.2)--(1.1,0)--(0.9,0.2);
\end{tikzpicture}}$ & 
$\renewcommand{\arraystretch}{1.5}
\begin{array}{c}
\textcolor{blue}{SO(11)} \\
\downarrow \\
\textcolor{blue}{SO(12)}
\end{array}$ \\
\hline

\end{tabular}
 \caption{Quivers resulting from adding all possible elementary slices to the existing nodes of n.min $D_n$ (\ref{nminDn}) to absorb $a$ for $n \geq 3$.}
    \label{tab:RHS nmin Dn}
\end{table}

\begin{table}[hbt!]
    \centering
    
 \hspace*{-2.375cm}\begin{tabular}{|c|c|c|} \hline
 
Added Slice & Quiver & Global Symmetry \\

\hline

$\begin{array}{c}
	a_k\\
	\\
	k\geq 1, n \geq 3
	\end{array}$ &
$\raisebox{-.5\height}{\begin{tikzpicture}[x=1cm,y=.8cm]
\node (g1) at (-3,0) [gauge,label=below:{$2$}] {};
\node (g2) at (-2,0) {$\cdots$};
\node (g3) at (-1,0) [gauger,label=below:{$2$}] {};
\node (g4) at (0,0) [gauge,label=below:{$2$}] {};
\node (g5) at (1,0) {$\cdots$};
\node (g6) at (2,0) [gauge,label=below:{$2$}] {};
\node (g7) at (3,0) [gauge,label=below:{$1$}] {};
\node (g9) at (-2,1) [gauge,label=above:{$2$}] {};
\draw (g1)--(g2)--(g3)--(g4)--(g5)--(g6);
\draw (g1)--(g9);
\draw (g3)--(g9);
\draw[transform canvas={yshift=-1.5pt}] (g7)--(g6);
\draw[transform canvas={yshift=1.5pt}] (g7)--(g6);
\draw (2.4,0.2)--(2.6,0)--(2.4,-0.2);
\draw [decorate,decoration={brace,mirror,amplitude=6pt}] (-3,-0.8) --node[below=6pt] {$k$} (-1,-0.8);
\draw [decorate,decoration={brace,mirror,amplitude=6pt}] (0,-0.8) --node[below=6pt] {$n-2$} (2,-0.8);
\end{tikzpicture}}$ & 
$\renewcommand{\arraystretch}{1.5}
\begin{array}{c}
SU(k+1) \times \textcolor{blue}{SO(2n-1)}  \\
\downarrow \\
SU(k+1) \times \textcolor{blue}{SO(2n)} 
\end{array}$ \\
\hline

$\begin{array}{c}
	d_k\\
	\\
	k\geq 4, n=k+1
	\end{array}$ &
$\raisebox{-.5\height}{\begin{tikzpicture}[x=1cm,y=.8cm]
\node (g1) at (-4,0) [gauge,label=below:{$2$}] {};
\node (g2) at (-3,0) [gauge,label=below:{$4$}] {};
\node (g3) at (-2,0) {$\cdots$};
\node (g4) at (-1,0) [gauge,label=below:{$4$}] {};
\node (g5) at (0,0) [gauger,label=below:{$2$}] {};
\node (g12) at (1,0) [gauge,label=below:{$2$}] {};
\node (g6) at (2,0) {$\cdots$};
\node (g7) at (3,0) [gauge,label=below:{$2$}] {};
\node (g8) at (4,0) [gauge,label=below:{$1$}] {};
\node (g10) at (-3,1) [gauge,label=above:{$2$}] {};
\node (g11) at (-1,1) [gauge,label=above:{$2$}] {};
\draw (g1)--(g2)--(g3)--(g4)--(g5)--(g12)--(g6)--(g7);
\draw (g10)--(g2);
\draw (g11)--(g4);
\draw[transform canvas={yshift=-1.5pt}] (g7)--(g8);
\draw[transform canvas={yshift=1.5pt}] (g7)--(g8);
\draw (3.4,0.2)--(3.6,0)--(3.4,-0.2);
\draw [decorate,decoration={brace,mirror,amplitude=6pt}] (-3,-0.8) --node[below=6pt] {$k-3$} (-1,-0.8);
\draw [decorate,decoration={brace,mirror,amplitude=6pt}] (1,-0.8) --node[below=6pt] {$n-2$} (3,-0.8);
\end{tikzpicture}}$ & 
$\renewcommand{\arraystretch}{1.5}
\begin{array}{c}
SO(2k) \times \textcolor{blue}{SO(2n-1)}  \\
\downarrow \\
SO(2k) \times \textcolor{blue}{SO(2n)} 
\end{array}$ \\
\hline

$\begin{array}{c}
	e_6\\
	\\
	n \geq 3
	\end{array}$ &
$\raisebox{-.5\height}{\begin{tikzpicture}[x=1cm,y=.8cm]
\node (g1) at (-4,0) [gauge,label=below:{$2$}] {};
\node (g2) at (-3,0) [gauge,label=below:{$4$}] {};
\node (g3) at (-2,0) [gauge,label=below:{$6$}] {};
\node (g4) at (-1,0) [gauge,label=below:{$4$}] {};
\node (g5) at (0,0) [gauger,label=below:{$2$}] {};
\node (g12) at (1,0) [gauge,label=below:{$2$}] {};
\node (g6) at (2,0) {$\cdots$};
\node (g7) at (3,0) [gauge,label=below:{$2$}] {};
\node (g8) at (4,0) [gauge,label=below:{$1$}] {};
\node (g10) at (-2,1) [gauge,label=left:{$4$}] {};
\node (g11) at (-2,2) [gauge,label=left:{$2$}] {};
\draw (g1)--(g2)--(g3)--(g4)--(g5)--(g12)--(g6)--(g7);
\draw (g3)--(g10)--(g11);
\draw[transform canvas={yshift=-1.5pt}] (g7)--(g8);
\draw[transform canvas={yshift=1.5pt}] (g7)--(g8);
\draw (3.4,0.2)--(3.6,0)--(3.4,-0.2);
\draw [decorate,decoration={brace,mirror,amplitude=6pt}] (1,-0.8) --node[below=6pt] {$n-2$} (3,-0.8);
\end{tikzpicture}}$ & 
$\renewcommand{\arraystretch}{1.5}
\begin{array}{c}
E_6 \times \textcolor{blue}{SO(2n-1)}  \\
\downarrow \\
E_6 \times \textcolor{blue}{SO(2n)} 
\end{array}$ \\
\hline

$\begin{array}{c}
	e_7\\
	n \geq 3
	\end{array}$ &
$\raisebox{-.5\height}{\begin{tikzpicture}[x=1cm,y=.8cm]
\node (g1) at (-5,0) [gauge,label=below:{$2$}] {};
\node (g2) at (-4,0) [gauge,label=below:{$4$}] {};
\node (g3) at (-3,0) [gauge,label=below:{$6$}] {};
\node (g4) at (-2,0) [gauge,label=below:{$8$}] {};
\node (g13) at (-1,0) [gauge,label=below:{$6$}] {};
\node (g5) at (-0,0) [gauge,label=below:{$4$}] {};
\node (g12) at (1,0) [gauger,label=below:{$2$}] {};
\node (g6) at (2,0) [gauge,label=below:{$2$}] {};
\node (g7) at (3,0) {$\cdots$};
\node (g14) at (4,0) [gauge,label=below:{$2$}] {};
\node (g8) at (5,0) [gauge,label=below:{$1$}] {};
\node (g10) at (-2,1) [gauge,label=left:{$4$}] {};
\draw (g1)--(g2)--(g3)--(g4)--(g13)--(g5)--(g12)--(g6)--(g7)--(g14);
\draw (g4)--(g10);
\draw[transform canvas={yshift=-1.5pt}] (g14)--(g8);
\draw[transform canvas={yshift=1.5pt}] (g14)--(g8);
\draw (4.4,0.2)--(4.6,0)--(4.4,-0.2);
\draw [decorate,decoration={brace,mirror,amplitude=6pt}] (2,-0.8) --node[below=6pt] {$n-2$} (4,-0.8);
\end{tikzpicture}}$ & 
$\renewcommand{\arraystretch}{1.5}
\begin{array}{c}
E_7 \times \textcolor{blue}{SO(2n-1)}  \\
\downarrow \\
E_7 \times \textcolor{blue}{SO(2n)} 
\end{array}$ \\
\hline

$\begin{array}{c}
	e_8\\
	n \geq 3
	\end{array}$ &
$\raisebox{-.5\height}{\begin{tikzpicture}[x=1cm,y=.8cm]
\node (g1) at (-5.5,0) [gauge,label=below:{$4$}] {};
\node (g2) at (-4.5,0) [gauge,label=below:{$8$}] {};
\node (g3) at (-3.5,0) [gauge,label=below:{$12$}] {};
\node (g4) at (-2.5,0) [gauge,label=below:{$10$}] {};
\node (g13) at (-1.5,0) [gauge,label=below:{$8$}] {};
\node (g5) at (-0.5,0) [gauge,label=below:{$6$}] {};
\node (g15) at (0.5,0) [gauge,label=below:{$4$}] {};
\node (g12) at (1.5,0) [gauger,label=below:{$2$}] {};
\node (g6) at (2.5,0) [gauge,label=below:{$2$}] {};
\node (g7) at (3.5,0) {$\cdots$};
\node (g14) at (4.5,0) [gauge,label=below:{$2$}] {};
\node (g8) at (5.5,0) [gauge,label=below:{$1$}] {};
\node (g10) at (-3.5,1) [gauge,label=left:{$6$}] {};
\draw (g1)--(g2)--(g3)--(g4)--(g13)--(g5)--(g15)--(g12)--(g6)--(g7)--(g14);
\draw (g3)--(g10);
\draw[transform canvas={yshift=-1.5pt}] (g14)--(g8);
\draw[transform canvas={yshift=1.5pt}] (g14)--(g8);
\draw (4.9,0.2)--(5.1,0)--(4.9,-0.2);
\draw [decorate,decoration={brace,mirror,amplitude=6pt}] (2.5,-0.8) --node[below=6pt] {$n-2$} (4.5,-0.8);
\end{tikzpicture}}$ & 
$\renewcommand{\arraystretch}{1.5}
\begin{array}{c}
E_8 \times \textcolor{blue}{SO(2n-1)}  \\
\downarrow \\
E_8 \times \textcolor{blue}{SO(2n)} 
\end{array}$ \\
\hline

\end{tabular}
 \caption{Quivers resulting from adding all possible simply laced slices to the empty nodes of n.min $D_n$ (\ref{nminDn}) to absorb $a$ for $n \geq 3$.}
    \label{tab:lhs sl addition nmin Dn}
\end{table}

\begin{table}[htbp!]
    \centering
    
 \hspace*{-0.75cm}\begin{tabular}{|c|c|c|} \hline
 
Added Slice & Quiver & Global Symmetry \\

\hline

$\begin{array}{c}
	b_k\\
	\\
	k\geq 3, n \geq 3
	\end{array}$ &
$\raisebox{-.5\height}{\begin{tikzpicture}[x=1cm,y=.8cm]
\node (g1) at (-4,0) [gauge,label=below:{$2$}] {};
\node (g2) at (-3,0) [gauge,label=below:{$4$}] {};
\node (g3) at (-2,0) {$\cdots$};
\node (g4) at (-1,0) [gauge,label=below:{$4$}] {};
\node (g5) at (0,0) [gauger,label=below:{$2$}] {};
\node (g12) at (1,0) [gauge,label=below:{$2$}] {};
\node (g6) at (2,0) {$\cdots$};
\node (g7) at (3,0) [gauge,label=below:{$2$}] {};
\node (g8) at (4,0) [gauge,label=below:{$1$}] {};
\node (g11) at (-1,1) [gauge,label=above:{$2$}] {};
\draw (g2)--(g3)--(g4)--(g5)--(g12)--(g6)--(g7);
\draw (g11)--(g4);
\draw[transform canvas={yshift=-1.5pt}] (g7)--(g8);
\draw[transform canvas={yshift=1.5pt}] (g7)--(g8);
\draw (3.4,0.2)--(3.6,0)--(3.4,-0.2);
\draw[transform canvas={yshift=-1.5pt}] (g1)--(g2);
\draw[transform canvas={yshift=1.5pt}] (g1)--(g2);
\draw (-3.4,0.2)--(-3.6,0)--(-3.4,-0.2);
\draw [decorate,decoration={brace,mirror,amplitude=6pt}] (-3,-0.8) --node[below=6pt] {$k-2$} (-1,-0.8);
\draw [decorate,decoration={brace,mirror,amplitude=6pt}] (1,-0.8) --node[below=6pt] {$n-2$} (3,-0.8);
\end{tikzpicture}}$ & 
$\renewcommand{\arraystretch}{1.5}
\begin{array}{c}
SO(2k+1) \times \textcolor{blue}{SO(2n-1)}  \\
\downarrow \\
SO(2k+1) \times \textcolor{blue}{SO(2n)} 
\end{array}$ \\
\hline

$\begin{array}{c}
	c_k\\
	\\
	k\geq 2, n \geq 3
	\end{array}$ &
$\raisebox{-.5\height}{\begin{tikzpicture}[x=1cm,y=.8cm]
\node (g1) at (-4,0) [gauge,label=below:{$2$}] {};
\node (g2) at (-3,0) [gauge,label=below:{$2$}] {};
\node (g3) at (-2,0) {$\cdots$};
\node (g4) at (-1,0) [gauge,label=below:{$2$}] {};
\node (g5) at (0,0) [gauger,label=below:{$2$}] {};
\node (g12) at (1,0) [gauge,label=below:{$2$}] {};
\node (g6) at (2,0) {$\cdots$};
\node (g7) at (3,0) [gauge,label=below:{$2$}] {};
\node (g8) at (4,0) [gauge,label=below:{$1$}] {};
\draw (g2)--(g3)--(g4);
\draw (g5)--(g12)--(g6)--(g7);
\draw[transform canvas={yshift=-1.5pt}] (g7)--(g8);
\draw[transform canvas={yshift=1.5pt}] (g7)--(g8);
\draw (3.4,0.2)--(3.6,0)--(3.4,-0.2);
\draw[transform canvas={yshift=-1.5pt}] (g1)--(g2);
\draw[transform canvas={yshift=1.5pt}] (g1)--(g2);
\draw (-3.6,0.2)--(-3.4,0)--(-3.6,-0.2);
\draw[transform canvas={yshift=-1.5pt}] (g4)--(g5);
\draw[transform canvas={yshift=1.5pt}] (g4)--(g5);
\draw (-0.4,0.2)--(-0.6,0)--(-0.4,-0.2);
\draw [decorate,decoration={brace,mirror,amplitude=6pt}] (-3,-0.8) --node[below=6pt] {$k-1$} (-1,-0.8);
\draw [decorate,decoration={brace,mirror,amplitude=6pt}] (1,-0.8) --node[below=6pt] {$n-2$} (3,-0.8);
\end{tikzpicture}}$ & 
$\renewcommand{\arraystretch}{1.5}
\begin{array}{c}
Sp(k) \times \textcolor{blue}{SO(2n-1)}  \\
\downarrow \\
Sp(k) \times \textcolor{blue}{SO(2n)} 
\end{array}$ \\
\hline

$\begin{array}{c}
	f_4\\
	\\
	n \geq 3
	\end{array}$ &
$\raisebox{-.5\height}{\begin{tikzpicture}[x=1cm,y=.8cm]
\node (g1) at (-4,0) [gauge,label=below:{$2$}] {};
\node (g2) at (-3,0) [gauge,label=below:{$4$}] {};
\node (g3) at (-2,0) [gauge,label=below:{$6$}] {};
\node (g4) at (-1,0) [gauge,label=below:{$4$}] {};
\node (g5) at (0,0) [gauger,label=below:{$2$}] {};
\node (g12) at (1,0) [gauge,label=below:{$2$}] {};
\node (g6) at (2,0) {$\cdots$};
\node (g7) at (3,0) [gauge,label=below:{$2$}] {};
\node (g8) at (4,0) [gauge,label=below:{$1$}] {};
\draw (g1)--(g2);
\draw (g3)--(g4)--(g5)--(g12)--(g6)--(g7);
\draw[transform canvas={yshift=-1.5pt}] (g7)--(g8);
\draw[transform canvas={yshift=1.5pt}] (g7)--(g8);
\draw (3.4,0.2)--(3.6,0)--(3.4,-0.2);
\draw[transform canvas={yshift=-1.5pt}] (g3)--(g2);
\draw[transform canvas={yshift=1.5pt}] (g3)--(g2);
\draw (-2.4,0.2)--(-2.6,0)--(-2.4,-0.2);
\draw [decorate,decoration={brace,mirror,amplitude=6pt}] (1,-0.8) --node[below=6pt] {$n-2$} (3,-0.8);
\end{tikzpicture}}$ & 
$\renewcommand{\arraystretch}{1.5}
\begin{array}{c}
F_4 \times \textcolor{blue}{SO(2n-1)}  \\
\downarrow \\
F_4 \times \textcolor{blue}{SO(2n)} 
\end{array}$ \\
\hline

$\begin{array}{c}
	g_2\\
	\\
	n \geq 3
	\end{array}$ &
$\raisebox{-.5\height}{\begin{tikzpicture}[x=1cm,y=.8cm]
\node (g3) at (-3,0) [gauge,label=below:{$2$}] {};
\node (g4) at (-2,0) [gauge,label=below:{$4$}] {};
\node (g5) at (-1,0) [gauger,label=below:{$2$}] {};
\node (g12) at (0,0) [gauge,label=below:{$2$}] {};
\node (g6) at (1,0) {$\cdots$};
\node (g7) at (2,0) [gauge,label=below:{$2$}] {};
\node (g8) at (3,0) [gauge,label=below:{$1$}] {};
\draw (g4)--(g5)--(g12)--(g6)--(g7);
\draw[transform canvas={yshift=-1.5pt}] (g7)--(g8);
\draw[transform canvas={yshift=1.5pt}] (g7)--(g8);
\draw (2.4,0.2)--(2.6,0)--(2.4,-0.2);
\draw[transform canvas={yshift=-2pt}] (g3)--(g4);
    \draw[transform canvas={yshift=0pt}] (g3)--(g4);
    \draw[transform canvas={yshift=2pt}] (g3)--(g4);
\draw (-2.4,0.2)--(-2.6,0)--(-2.4,-0.2);
\draw [decorate,decoration={brace,mirror,amplitude=6pt}] (0,-0.8) --node[below=6pt] {$n-2$} (2,-0.8);
\end{tikzpicture}}$ & 
$\renewcommand{\arraystretch}{1.5}
\begin{array}{c}
G_2 \times \textcolor{blue}{SO(2n-1)}  \\
\downarrow \\
G_2 \times \textcolor{blue}{SO(2n)} 
\end{array}$ \\
\hline

\end{tabular}
 \caption{Quivers resulting from adding all possible non-simply laced slice to the empty nodes of n.min $D_n$ (\ref{nminDn}) to absorb $a$ for $n \geq 3$}
    \label{tab:lhs nsl addition nmin Dn}
\end{table}

\begin{table}[htbp!]
    \centering
    
 \hspace*{-0.3cm}\begin{tabular}{|c|c|c|} \hline
 
Added Slice & Quiver & Global Symmetry \\

\hline

$\begin{array}{c}
	b_k\\
	\\
	k\geq 3
	\end{array}$ &
$\raisebox{-.5\height}{\begin{tikzpicture}[x=1cm,y=.8cm]
\node (g0) at (-2.5,0) [gauge,label=below:{$1$}] {};
\node (g1) at (-1.5,0) [gauger,label=below:{$2$}] {};
\node (g2) at (-0.5,0) [gauge,label=below:{$4$}] {};
\node (g3) at (0.5,0) {$\cdots$};
\node (g4) at (1.5,0) [gauge,label=below:{$4$}] {};
\node (g5) at (2.5,0) [gauge,label=below:{$2$}] {};
\node (g7) at (1.5,1) [gauge,label=above:{$2$}] {};
\draw (g0)--(g1);
\draw (g2)--(g3)--(g4)--(g5);
\draw (g7)--(g4);
\draw[transform canvas={yshift=-1.5pt}] (g1)--(g2);
\draw[transform canvas={yshift=1.5pt}] (g1)--(g2);
\draw (-0.9,0.2)--(-1.1,0)--(-0.9,-0.2);
\draw [decorate,decoration={brace,mirror,amplitude=6pt}] (-0.5,-0.8) --node[below=6pt] {$k-2$} (1.5,-0.8);
\end{tikzpicture}}$ & 
$\renewcommand{\arraystretch}{1.5}
\begin{array}{c}
\textcolor{blue}{SU(2)} \times SO(2k) \\
\downarrow \\
\textcolor{blue}{SU(2)^2} \times SO(2k)
\end{array}$ \\
\hline

$\begin{array}{c}
	c_k\\
	\\
	k \geq 2
	\end{array}$ &
$\raisebox{-.5\height}{\begin{tikzpicture}[x=1cm,y=.8cm]
    \node (g9) at (-3.5,0) [gauge,label=below:{$2$}] {};
    \node (g7) at (-2.5,0) [gauge,label=below:{$2$}] {};
    \node (g6) at (-1.5,0) {$\cdots$};
    \node (g1) at (-0.5,0) [gauger,label=below:{$2$}] {};
    \node (g8) at (0.5,0) [gauge,label=below:{$2$}] {};
    \node (g2) at (1.5,0) {$\cdots$};
    \node (g3) at (2.5,0) [gauge,label=below:{$2$}] {};
    \node (g4) at (3.5,0) [gauge,label=below:{$2$}] {};
    \node (g10) at (-0.5,1) [gauge,label=above:{$1$}] {};
    \draw (g7)--(g6)--(g1)--(g8)--(g2)--(g3);
    \draw (g10)--(g1);
    \draw[transform canvas = {yshift=-1.5pt}] (g9)--(g7);
    \draw[transform canvas = {yshift=1.5pt}] (g9)--(g7);
    \draw (-3.1,0.2)--(-2.9,0)--(-3.1,-0.2);
    \draw[transform canvas = {yshift=-1.5pt}] (g3)--(g4);
    \draw[transform canvas = {yshift=1.5pt}] (g3)--(g4);
    \draw (3.1,0.2)--(2.9,0)--(3.1,-0.2);
    \draw [decorate,decoration={brace,mirror,amplitude=6pt}] (0.5,-0.8) --node[below=6pt] {$k-j-1$} (2.5,-0.8);
    \draw [decorate,decoration={brace,mirror,amplitude=6pt}] (-2.5,-0.8) --node[below=6pt] {$j$} (-0.5,-0.8);
    \end{tikzpicture}}$ & 
$\renewcommand{\arraystretch}{1.5}
\begin{array}{c}
\textcolor{blue}{SU(2)} \times Sp(j) \times Sp(k-j)\\
\downarrow \\
\textcolor{blue}{SU(2)^2} \times Sp(j) \times Sp(k-j)
\end{array}$ \\

\hline

$\begin{array}{c}
	f_4\\
	\end{array}$ &
$\raisebox{-.5\height}{\begin{tikzpicture}[x=1cm,y=.8cm]
\node (g1) at (-2.5,0) [gauge,label=below:{$1$}] {};
\node (g2) at (-1.5,0) [gauger,label=below:{$2$}] {};
\node (g3) at (-0.5,0) [gauge,label=below:{$4$}] {};
\node (g4) at (0.5,0) [gauge,label=below:{$6$}] {};
\node (g5) at (1.5,0) [gauge,label=below:{$4$}] {};
\node (g12) at (2.5,0) [gauge,label=below:{$2$}] {};
\draw (g3)--(g2)--(g1);
\draw (g4)--(g5)--(g12);
\draw[transform canvas={yshift=-1.5pt}] (g3)--(g4);
\draw[transform canvas={yshift=1.5pt}] (g3)--(g4);
\draw (0.1,0.2)--(-0.1,0)--(0.1,-0.2);
\end{tikzpicture}}$ & 
$\renewcommand{\arraystretch}{1.5}
\begin{array}{c}
\textcolor{blue}{SU(2)} \times SO(9) \\
\downarrow \\
\textcolor{blue}{SU(2)^2} \times SO(9) 
\end{array}$ \\
\hline

$\begin{array}{c}
	a_{2k-1}\\
	\\
	k\geq 3
	\end{array}$ &
$\raisebox{-.5\height}{\begin{tikzpicture}[x=1cm,y=.8cm]
\node (g0) at (-2.5,0) [gauge,label=below:{$1$}] {};
\node (g1) at (-1.5,0) [gauger,label=below:{$2$}] {};
\node (g2) at (-0.5,0) [gauge,label=below:{$4$}] {};
\node (g3) at (0.5,0) {$\cdots$};
\node (g4) at (1.5,0) [gauge,label=below:{$4$}] {};
\node (g5) at (2.5,0) [gauge,label=below:{$4$}] {};
\node (g7) at (-0.5,1) [gauge,label=above:{$2$}] {};
\draw (g0)--(g1)--(g2)--(g3)--(g4);
\draw (g7)--(g2);
\draw[transform canvas={yshift=-1.5pt}] (g4)--(g5);
\draw[transform canvas={yshift=1.5pt}] (g4)--(g5);
\draw (2.1,0.2)--(1.9,0)--(2.1,-0.2);
\draw [decorate,decoration={brace,mirror,amplitude=6pt}] (-0.5,-0.8) --node[below=6pt] {$k-2$} (1.5,-0.8);
\end{tikzpicture}}$ & 
$\renewcommand{\arraystretch}{1.5}
\begin{array}{c}
\textcolor{blue}{SU(2)} \times Sp(k) \\
\downarrow \\
\textcolor{blue}{SU(2)^2} \times Sp(k)
\end{array}$ \\
\hline

$\begin{array}{c}
	d_{k+1}\\
	\\
	k \geq 2
	\end{array}$ &
$\raisebox{-.5\height}{\begin{tikzpicture}[x=1cm,y=.8cm]
\node (g0) at (-2.5,0) [gauge,label=below:{$1$}] {};
\node (g1) at (-1.5,0) [gauger,label=below:{$2$}] {};
\node (g2) at (-0.5,0) [gauge,label=below:{$4$}] {};
\node (g3) at (0.5,0) {$\cdots$};
\node (g4) at (1.5,0) [gauge,label=below:{$4$}] {};
\node (g5) at (2.5,0) [gauge,label=below:{$2$}] {};
\draw (g0)--(g1);
\draw (g2)--(g3)--(g4);
\draw[transform canvas={yshift=-1.5pt}] (g1)--(g2);
\draw[transform canvas={yshift=1.5pt}] (g1)--(g2);
\draw (-0.9,0.2)--(-1.1,0)--(-0.9,-0.2);
\draw[transform canvas={yshift=-1.5pt}] (g5)--(g4);
\draw[transform canvas={yshift=1.5pt}] (g5)--(g4);
\draw (1.9,-0.2)--(2.1,0)--(1.9,0.2);
\draw [decorate,decoration={brace,mirror,amplitude=6pt}] (-0.5,-0.8) --node[below=6pt] {$k-1$} (1.5,-0.8);
\end{tikzpicture}}$ & 
$\renewcommand{\arraystretch}{1.5}
\begin{array}{c}
\textcolor{blue}{SU(2)} \times SO(2k+1) \\
\downarrow \\
\textcolor{blue}{SU(2)^2} \times SO(2k+1)
\end{array}$ \\
\hline

$\begin{array}{c}
	e_6\\
	\end{array}$ &
$\raisebox{-.5\height}{\begin{tikzpicture}[x=1cm,y=.8cm]
\node (g1) at (-2.5,0) [gauge,label=below:{$1$}] {};
\node (g2) at (-1.5,0) [gauger,label=below:{$2$}] {};
\node (g3) at (-0.5,0) [gauge,label=below:{$4$}] {};
\node (g4) at (0.5,0) [gauge,label=below:{$6$}] {};
\node (g5) at (1.5,0) [gauge,label=below:{$8$}] {};
\node (g12) at (2.5,0) [gauge,label=below:{$4$}] {};
\draw (g1)--(g2)--(g3)--(g4);
\draw (g5)--(g12);
\draw[transform canvas={yshift=-1.5pt}] (g5)--(g4);
\draw[transform canvas={yshift=1.5pt}] (g5)--(g4);
\draw (1.1,0.2)--(0.9,0)--(1.1,-0.2);
\end{tikzpicture}}$ & 
$\renewcommand{\arraystretch}{1.5}
\begin{array}{c}
\textcolor{blue}{SU(2)} \times F_4 \\
\downarrow \\
\textcolor{blue}{SU(2)^2} \times F_4 
\end{array}$ \\
\hline

\end{tabular}
 \caption{Quivers resulting from adding all possible elementary slices to n.min $D_2$ (\ref{nminD2}) to absorb $a$.}
    \label{tab:adding to nmin D2}
\end{table}

\clearpage

\subsection{\texorpdfstring{Enhancement from $SU(n) \times U(1)$ to $SO(2n)$}{Enhancement from AncrossU1 to Dn}} \label{sec: sun u1 to so2n}
The final quiver we consider for enhancement to $D$-type symmetry is that of (\ref{nminBn}) with a further $\mathbb{Z}_2$ quotient taken, to combine the two spinor nodes into another rank two node with an adjoint hypermultiplet \cite{Bourget:2020bxh}:
\begin{equation}\label{double adj Dn quiver}
    \begin{tikzpicture}[x=1cm,y=.8cm]
    \node (g1) at (-1.5,1.1) [gauge,label=right:{$2$},label={[red]left:{$a_1$}}] {};
    \node (g2) at (-1.5,0) [gauge,label=below:{$2$},label={[red]left:{$b_1$}}] {};
    \node (g3) at (-0.5,0) {$\cdots$};
    \node (g4) at (0.5,0) [gauge,label=below:{$2$},label={[red]right:{$b_2$}}] {};
    \node (g5) at (0.5,1.1) [gauge,label=left:{$2$},label={[red]right:{$a_2$}}] {};
    \draw (g1)--(g2)--(g3)--(g4)--(g5);
    \draw (g1) to [out=45,in=135,looseness=10] (g1);
    \draw (g5) to [out=45,in=135,looseness=10] (g5);
    \draw [decorate,decoration={brace,mirror,amplitude=6pt}] (-1.5,-0.8) --node[below=6pt] {$n-2$} (0.5,-0.8);
    \end{tikzpicture}
\end{equation}
This quiver has $SO(2n)$ global symmetry because under the modified BGS algorithm of Section \ref{sec:quivers with adj matter} we see the shape of the twisted affine $d_n$ quiver. As mentioned in Section \ref{sec: quiver addition}, a modification on the quiver addition algorithm is needed in this case. This is because if we take $c_1$ and $c_2$ to be the existing nodes the right and left of $b_1$ and $b_2$ respectively, then one can check that it is impossible to absorb $a_1$ and $a_2$ by adding on to these nodes. Thus the only legitimate additions one can perform here are by taking $c_1$ and $c_2$ to be the empty nodes to the left and right of $b_1$ and $b_2$ respectively, 
\begin{equation}\label{cnode double adj Dn quiver}
    \begin{tikzpicture}[x=1cm,y=.8cm]
    \node (g0) at (-2.5,0) [emptygauge, label={[gray]below:{$0$}}, label={[red]left:{$c_1$}}] {};
    \node (g1) at (-1.5,1.1) [gauge,label=right:{$2$},label={[red]left:{$a_1$}}] {};
    \node (g2) at (-1.5,0) [gauge,label=below:{$2$},label={[red]315:{$b_1$}}] {};
    \node (g3) at (-0.5,0) {$\cdots$};
    \node (g4) at (0.5,0) [gauge,label=below:{$2$},label={[red]225:{$b_2$}}] {};
    \node (g5) at (0.5,1.1) [gauge,label=left:{$2$},label={[red]right:{$a_2$}}] {};
    \node (g6) at (1.5,0) [emptygauge, label={[gray]below:{$0$}}, label={[red]right:{$c_2$}}] {};
    \node (g7) at (2,0) {$,$};
    \draw (g1)--(g2)--(g3)--(g4)--(g5);
    \draw [dashed,gray] (g0)--(g2);
    \draw [dashed,gray] (g4)--(g6);
    \draw (g1) to [out=45,in=135,looseness=10] (g1);
    \draw (g5) to [out=45,in=135,looseness=10] (g5);
    \draw [decorate,decoration={brace,mirror,amplitude=6pt}] (-1.5,-0.8) --node[below=6pt] {$n-2$} (0.5,-0.8);
    \end{tikzpicture}
\end{equation}
and adding some affine quiver $g_k$ to $c_1$ and another affine quiver $\tilde{g}_l$ to $c_2$, so that upon double subtraction of both $g_k$ and $\tilde{g}_l$, both $a_1$ and $a_2$ must appear in order to rebalance. There will be $45$ such quivers so we shall not list them all here, but they take the form
\begin{equation}\label{double adj so2n enhancement}
     \begin{tikzpicture}[x=1cm,y=.8cm]
    \node (g-1) at (-3.5,0) {$\cdots$};
    \node (g0) at (-3,0) {$\cdots$};
    \node (g1) at (-2,0) [gauger,label=below:{$2$}] {};
    \node (g2) at (-1,0) [gauge,label=below:{$2$}] {};
    \node (g3) at (-0,0) {$\cdots$};
    \node (g4) at (1,0) [gauge,label=below:{$2$}] {};
    \node (g5) at (2,0) [gauger,label=below:{$2$}] {};
    \node (g6) at (3,0) {$\cdots$};
    \node (g7) at (3.5,0) {$\cdots$};
    \draw (g0)--(g1)--(g2)--(g3)--(g4)--(g5)--(g6);
    \draw [decorate,decoration={brace,mirror,amplitude=6pt}] (-4,-0.8) --node[below=6pt] {$2 \, g_k$} (-2,-0.8);
    \draw [decorate,decoration={brace,mirror,amplitude=6pt}] (-1,-0.8) --node[below=6pt] {$n-2$} (1,-0.8);
    \draw [decorate,decoration={brace,mirror,amplitude=6pt}] (2,-0.8) --node[below=6pt] {$2\, \tilde{g}_l$} (4,-0.8);
    \end{tikzpicture}
\end{equation}
and will exhibit the enhancement from the BGS as 
\begin{equation}
    G_k \times \textcolor{blue}{SU(n-1) \times U(1)} \times \tilde{G}_l \rightarrow G_k \times \textcolor{blue}{SO(2n)} \times \tilde{G}_l.
\end{equation}
To be explicit, we show an example of adding $b_3$ to the left of node $a_1$, and $d_4$ to the right of node $a_2$ - that is, $g_k=b_3$ and $\tilde{g}_l = d_4$:
\begin{equation}
     \begin{tikzpicture}[x=1cm,y=.8cm]
    \node (g-1) at (-4,0) [gauge,label=below:{$2$}] {};
    \node (g0) at (-3,0) [gauge,label=below:{$4$}] {};
    \node (g1) at (-2,0) [gauger,label=below:{$2$}] {};
    \node (g8) at (-3,1) [gauge,label=left:{$2$}] {};
    \node (g2) at (-1,0) [gauge,label=below:{$2$}] {};
    \node (g3) at (-0,0) {$\cdots$};
    \node (g4) at (1,0) [gauge,label=below:{$2$}] {};
    \node (g5) at (2,0) [gauger,label=below:{$2$}] {};
    \node (g6) at (3,0) [gauge,label=below:{$4$}] {};
    \node (g7) at (4,0) [gauge,label=below:{$2$}] {};
    \node (g9) at (2.5,1) [gauge,label=left:{$2$}] {};
    \node (g10) at (3.5,1) [gauge,label=right:{$2$}] {};
    \draw (g0)--(g1)--(g2)--(g3)--(g4)--(g5)--(g6)--(g7);
    \draw (g6)--(g9);
    \draw (g6)--(g10);
    \draw (g8)--(g0);
    \draw[transform canvas={yshift=-2pt}] (g-1)--(g0);
    \draw[transform canvas={yshift=2pt}] (g-1)--(g0);
    \draw (-3.4,0.2)--(-3.6,0)--(-3.4,-0.2);
     \draw [decorate,decoration={brace,mirror,amplitude=6pt}] (-1,-0.8) --node[below=6pt] {$n-2$} (1,-0.8);
    \end{tikzpicture}
\end{equation}
Several of these cases have been checked via the Hilbert series, and the enhancement has been confirmed for all of these.\\

\noindent Again, due to the lack of long rank one nodes present, a different addition prior to absorbing $a_1$ or $a_2$ cannot happen as it did in Section \ref{A type to B type section}, and so the quivers given in Tables \ref{tab:RHS nmin Dn}, \ref{tab:lhs sl addition nmin Dn}, \ref{tab:lhs nsl addition nmin Dn}, \ref{tab:adding to nmin D2} and those encapsulated by the general form of (\ref{double adj so2n enhancement}) comprise the full list of quivers derived from quiver addition on the next to minimal nilpotent orbit of $D_n$ (\ref{nminDn}) and (\ref{double adj Dn quiver}) whose global symmetry contains a factor which is enhanced to $SO(2n)$ from that predicted by the BGS.

\section{\texorpdfstring{Enhancement to $SU(3)$}{Enhancement to SU3}}
\label{sec:enhancement to A2}
The final quiver from Table \ref{tab:adj hyper Qs} that we know how to perform quiver addition on is
\begin{equation} \label{A2 enhancement}
    \begin{tikzpicture}[x=1cm,y=.8cm]
    \node (g1) at (0,1.1) [gauge,label=right:{$4$},label={[red]left:{$a$}}] {};
    \node (g2) at (0,0) [gauge,label=below:{$2$},label={[red]left:{$b$}}] {};
    \node (g3) at (0.5,0) {$.$};
    \draw (g1)--(g2);
    \draw (g1) to [out=45,in=135,looseness=10] (g1);
    \end{tikzpicture}
\end{equation}
Note that although the greatest common divisor of the ranks of the gauge nodes here is greater than one, the Coulomb branch of this quiver is still a Hyperk\"{a}hler cone. The Hasse diagram for (\ref{A2 enhancement}) is a question for future work. The ammended BGS algorithm in Section \ref{sec:quivers with adj matter} tells us that the BGS of this quiver is $SU(3)$, and indeed this estimated global symmetry is confirmed upon Hilbert series computation. As a result, the quivers presented in this section all experience the enhancement in a factor of their BGS to $SU(3)$:
\begin{equation}
    \prod_{i} G_i \times \textcolor{blue}{SU(2)} \rightarrow \prod_{i} G_i \times \textcolor{blue}{SU(3)},
\end{equation}
for some semi-simple Lie groups $G_i$.\\

\noindent Following the terminology of the quiver addition algorithm in Section \ref{sec: quiver addition}, there are no possible existing $c$-nodes in (\ref{A2 enhancement}), so the only available $c$-nodes here are empty:
\begin{equation} \label{cnode A2 enhancement}
    \begin{tikzpicture}[x=1cm,y=.8cm]
    \node (g0) at (-1,0) [emptygauge,label={[gray]below:{$0$}}, label={[red]left:{$c$}}] {};
    \node (g1) at (0,1.1) [gauge,label=right:{$4$},label={[red]left:{$a$}}] {};
    \node (g2) at (0,0) [gauge,label=below:{$2$},label={[red]315:{$b$}}] {};
    \node (g3) at (0.5,0) {$,$};
    \draw (g1)--(g2);
    \draw [dashed,gray] (g2)--(g0);
    \draw (g1) to [out=45,in=135,looseness=10] (g1);
    \end{tikzpicture}
\end{equation}
A result of this is that, as before, we may add any affine slice. Since the node $a$ is of rank four, to absorb it the slice that we add $g_k$ must be added four times. The results of performing these quiver additions are given in Tables \ref{tab: sl A2 enhancement} and \ref{tab:nsl A2 enhancement}. For those quivers in these tables which depend on the parameter $k$, for low values of $k$ Hilbert series computations have confirmed this enhancement. The $e_6$ and $g_2$ cases have also been confirmed, but we have been unable to compute the Hilbert series for the $e_7$, $e_8$ and $f_4$ cases, and so these enhancements remain as conjectures. Also, note that the greatest common divisor of the node ranks of the quivers in Tables \ref{tab: sl A2 enhancement} and \ref{tab:nsl A2 enhancement} is greater than one. As mentioned in Section \ref{sec:Intro}, this is often an indication of a diverging Hilbert series. Here the quivers are too complex to confirm or deny this by computation, but since they are derived from (\ref{A2 enhancement}) which we know does not suffer this divergence, there is a possibility that these quivers are also exempt. \\

\noindent Again as with (\ref{nminDn}) and (\ref{double adj Dn quiver}) of Section \ref{sec:enhancement to Dn}, since (\ref{A2 enhancement}) has no long rank one nodes we cannot add onto this quiver other than to absorb $a$, and so do not get further constructions as in Section \ref{A type to B type section}. As a result, Tables \ref{tab: sl A2 enhancement} and \ref{tab:nsl A2 enhancement} conclude all possible quivers derived from performing quiver addition on (\ref{A2 enhancement}) to absorb $a$ whose global symmetry contains a factor which is enhanced to $SU(3)$.

\clearpage

\begin{table}[htbp!]
    \centering
    
 \hspace*{-0.05cm}\begin{tabular}{|c|c|c|} \hline
 
Added Slice & Quiver & Global Symmetry \\

\hline

$\begin{array}{c}
	a_k\\
	\\
	k\geq 1
	\end{array}$ &
$\raisebox{-.5\height}{\begin{tikzpicture}[x=1cm,y=.8cm]
\node (g1) at (-1.5,0) [gauge,label=below:{$2$}] {};
\node (g2) at (-0.5,0) [gauger,label=below:{$4$}] {};
\node (g4) at (0.5,1) [gauge,label=above:{$4$}] {};
\node (g5) at (0.5,0) {$\cdots$};
\node (g6) at (1.5,0) [gauge,label=below:{$4$}] {};
\draw (g1)--(g2)--(g5)--(g6)--(g4);
\draw (g4)--(g2);
\draw [decorate,decoration={brace,mirror,amplitude=6pt}] (-0.5,-0.8) --node[below=6pt] {$k$} (1.5,-0.8);
\end{tikzpicture}}$ & 
$\renewcommand{\arraystretch}{1.5}
\begin{array}{c}
\textcolor{blue}{SU(2)} \times SU(k+1) \\
\downarrow \\
\textcolor{blue}{SU(3)} \times SU(k+1)
\end{array}$ \\
\hline

$\begin{array}{c}
	d_k\\
	\\
	k\geq 4
	\end{array}$ &
$\raisebox{-.5\height}{\begin{tikzpicture}[x=1cm,y=.8cm]
\node (g0) at (-2.5,0) [gauge,label=below:{$2$}] {};
\node (g1) at (-1.5,0) [gauger,label=below:{$4$}] {};
\node (g2) at (-0.5,0) [gauge,label=below:{$8$}] {};
\node (g3) at (0.5,0) {$\cdots$};
\node (g4) at (1.5,0) [gauge,label=below:{$8$}] {};
\node (g5) at (2.5,0) [gauge,label=below:{$4$}] {};
\node (g12) at (1.5,1) [gauge,label=above:{$4$}] {};
\node (g13) at (-0.5,1) [gauge,label=above:{$4$}] {};
\draw (g0)--(g1)--(g2)--(g3)--(g4)--(g5);
\draw (g12)--(g4);
\draw (g13)--(g2);
\draw [decorate,decoration={brace,mirror,amplitude=6pt}] (-0.5,-0.8) --node[below=6pt] {$k-3$} (1.5,-0.8);
\end{tikzpicture}}$ & 
$\renewcommand{\arraystretch}{1.5}
\begin{array}{c}
\textcolor{blue}{SU(2)} \times SO(2k) \\
\downarrow \\
\textcolor{blue}{SU(3)} \times SO(2k)
\end{array}$ \\
\hline

$\begin{array}{c}
	e_6\\
	\end{array}$ &
$\raisebox{-.5\height}{\begin{tikzpicture}[x=1cm,y=.8cm]
\node (g0) at (-2.5,0) [gauge,label=below:{$2$}] {};
\node (g1) at (-1.5,0) [gauger,label=below:{$4$}] {};
\node (g2) at (-0.5,0) [gauge,label=below:{$8$}] {};
\node (g3) at (0.5,0) [gauge,label=below:{$12$}] {};
\node (g4) at (1.5,0) [gauge,label=below:{$8$}] {};
\node (g5) at (2.5,0) [gauge,label=below:{$4$}] {};
\node (g12) at (0.5,1) [gauge,label=left:{$8$}] {};
\node (g13) at (0.5,2) [gauge,label=left:{$4$}] {};
\draw (g0)--(g1)--(g2)--(g3)--(g4)--(g5);
\draw (g13)--(g12)--(g3);
\end{tikzpicture}}$ & 
$\renewcommand{\arraystretch}{1.5}
\begin{array}{c}
\textcolor{blue}{SU(2)} \times E_6 \\
\downarrow \\
\textcolor{blue}{SU(3)} \times E_6
\end{array}$ \\
\hline

$\begin{array}{c}
	e_7\\
	\end{array}$ &
$\raisebox{-.5\height}{\begin{tikzpicture}[x=1cm,y=.8cm]
\node (g1) at (-3.5,0) [gauge,label=below:{$2$}] {};
\node (g2) at (-2.5,0) [gauger,label=below:{$4$}] {};
\node (g3) at (-1.5,0) [gauge,label=below:{$8$}] {};
\node (g4) at (-0.5,0) [gauge,label=below:{$12$}] {};
\node (g5) at (0.5,0) [gauge,label=below:{$16$}] {};
\node (g6) at (1.5,0) [gauge,label=below:{$12$}] {};
\node (g7) at (2.5,0) [gauge,label=below:{$8$}] {};
\node (g8) at (3.5,0) [gauge,label=below:{$4$}] {};
\node (g9) at (0.5,1) [gauge,label=left:{$8$}] {};
\draw (g1)--(g2)--(g3)--(g4)--(g5)--(g6)--(g7)--(g8);
\draw (g5)--(g9);
\end{tikzpicture}}$ & 
$\renewcommand{\arraystretch}{1.5}
\begin{array}{c}
\textcolor{blue}{SU(2)} \times E_7 \\
\downarrow \\
\textcolor{blue}{SU(3)} \times E_7
\end{array}$ \\
\hline

$\begin{array}{c}
	e_8\\
	\end{array}$ &
$\raisebox{-.5\height}{\begin{tikzpicture}[x=1cm,y=.8cm]
\node (g1) at (-4,0) [gauge,label=below:{$2$}] {};
\node (g2) at (-3,0) [gauger,label=below:{$4$}] {};
\node (g3) at (-2,0) [gauge,label=below:{$8$}] {};
\node (g4) at (-1,0) [gauge,label=below:{$12$}] {};
\node (g5) at (0,0) [gauge,label=below:{$16$}] {};
\node (g6) at (1,0) [gauge,label=below:{$20$}] {};
\node (g7) at (2,0) [gauge,label=below:{$24$}] {};
\node (g8) at (3,0) [gauge,label=below:{$16$}] {};
\node (g9) at (4,0) [gauge,label=below:{$8$}] {};
\node (g10) at (2,1) [gauge,label=left:{$12$}] {};
\draw (g1)--(g2)--(g3)--(g4)--(g5)--(g6)--(g7)--(g8)--(g9);
\draw (g7)--(g10);
\end{tikzpicture}}$ & 
$\renewcommand{\arraystretch}{1.5}
\begin{array}{c}
\textcolor{blue}{SU(2)} \times E_8 \\
\downarrow \\
\textcolor{blue}{SU(3)} \times E_8
\end{array}$ \\
\hline

\end{tabular}
 \caption{Quivers resulting from adding all possible simply laced elementary slices to (\ref{A2 enhancement}) to absorb $a$.}
    \label{tab: sl A2 enhancement}
\end{table}

\begin{table}[htbp!]
    \centering
    
\begin{tabular}{|c|c|c|} \hline
 
Added Slice & Quiver & Global Symmetry \\

\hline

$\begin{array}{c}
	b_k\\
	\\
	k \geq 3
	\end{array}$ &
$\raisebox{-.5\height}{\begin{tikzpicture}[x=1cm,y=.8cm]
\node (g0) at (-2.5,0) [gauge,label=below:{$2$}] {};
\node (g1) at (-1.5,0) [gauger,label=below:{$4$}] {};
\node (g2) at (-0.5,0) [gauge,label=below:{$8$}] {};
\node (g3) at (0.5,0) {$\cdots$};
\node (g4) at (1.5,0) [gauge,label=below:{$8$}] {};
\node (g5) at (2.5,0) [gauge,label=below:{$4$}] {};
\node (g11) at (-0.5,1) [gauge,label=above:{$4$}] {};
\draw (g0)--(g1)--(g2)--(g3)--(g4);
\draw (g2)--(g11);
\draw[transform canvas={yshift=-1.5pt}] (g5)--(g4);
\draw[transform canvas={yshift=1.5pt}] (g5)--(g4);
\draw (1.9,-0.2)--(2.1,0)--(1.9,0.2);
\draw [decorate,decoration={brace,mirror,amplitude=6pt}] (-0.5,-0.8) --node[below=6pt] {$k-2$} (1.5,-0.8);
\end{tikzpicture}}$ & 
$\renewcommand{\arraystretch}{1.5}
\begin{array}{c}
\textcolor{blue}{SU(2)} \times SO(2k+1) \\
\downarrow \\
\textcolor{blue}{SU(3)} \times SO(2k+1)
\end{array}$ \\
\hline

$\begin{array}{c}
	c_k\\
	\\
	k \geq 2
	\end{array}$ &
$\raisebox{-.5\height}{\begin{tikzpicture}[x=1cm,y=.8cm]
\node (g0) at (-2.5,0) [gauge,label=below:{$2$}] {};
\node (g1) at (-1.5,0) [gauger,label=below:{$4$}] {};
\node (g2) at (-0.5,0) [gauge,label=below:{$4$}] {};
\node (g3) at (0.5,0) {$\cdots$};
\node (g4) at (1.5,0) [gauge,label=below:{$4$}] {};
\node (g5) at (2.5,0) [gauge,label=below:{$4$}] {};
\draw (g0)--(g1);
\draw (g2)--(g3)--(g4);
\draw[transform canvas={yshift=-1.5pt}] (g1)--(g2);
\draw[transform canvas={yshift=1.5pt}] (g1)--(g2);
\draw (-1.1,0.2)--(-0.9,0)--(-1.1,-0.2);
\draw[transform canvas={yshift=-1.5pt}] (g5)--(g4);
\draw[transform canvas={yshift=1.5pt}] (g5)--(g4);
\draw (2.1,-0.2)--(1.9,0)--(2.1,0.2);
\draw [decorate,decoration={brace,mirror,amplitude=6pt}] (-0.5,-0.8) --node[below=6pt] {$k-1$} (1.5,-0.8);
\end{tikzpicture}}$ & 
$\renewcommand{\arraystretch}{1.5}
\begin{array}{c}
\textcolor{blue}{SU(2)} \times Sp(k) \\
\downarrow \\
\textcolor{blue}{SU(3)} \times Sp(k)
\end{array}$ \\
\hline

$\begin{array}{c}
	f_4\\
	\end{array}$ &
$\raisebox{-.5\height}{\begin{tikzpicture}[x=1cm,y=.8cm]
\node (g0) at (-2.5,0) [gauge,label=below:{$2$}] {};
\node (g1) at (-1.5,0) [gauger,label=below:{$4$}] {};
\node (g2) at (-0.5,0) [gauge,label=below:{$8$}] {};
\node (g3) at (0.5,0) [gauge,label=below:{$12$}] {};
\node (g4) at (1.5,0) [gauge,label=below:{$8$}] {};
\node (g5) at (2.5,0) [gauge,label=below:{$4$}] {};
\draw (g0)--(g1)--(g2)--(g3);
\draw (g4)--(g5);
\draw[transform canvas={yshift=-1.5pt}] (g3)--(g4);
\draw[transform canvas={yshift=1.5pt}] (g3)--(g4);
\draw (0.9,-0.2)--(1.1,0)--(0.9,0.2);
\end{tikzpicture}}$ & 
$\renewcommand{\arraystretch}{1.5}
\begin{array}{c}
\textcolor{blue}{SU(2)} \times F_4 \\
\downarrow \\
\textcolor{blue}{SU(3)} \times F_4
\end{array}$ \\
\hline

$\begin{array}{c}
    g_2
    \end{array}$ &
$\raisebox{-.5\height}{\begin{tikzpicture}[x=1cm,y=.8cm]
\node (g1) at (-1.5,0) [gauge,label=below:{$2$}] {};
\node (g2) at (-0.5,0) [gauger,label=below:{$4$}] {};
\node (g3) at (0.5,0) [gauge,label=below:{$8$}] {};
\node (g4) at (1.5,0) [gauge,label=below:{$4$}] {};
\draw (g1)--(g2)--(g3);
\draw[transform canvas = {yshift=-2pt}] (g4)--(g3);
\draw[transform canvas = {yshift=0pt}] (g4)--(g3);
\draw[transform canvas = {yshift =2pt}] (g4)--(g3);
\draw (0.9,0.2)--(1.1,0)--(0.9,-0.2);
\end{tikzpicture}}$ &
$\renewcommand{\arraystretch}{1.5}
\begin{array}{c}
\textcolor{blue}{SU(2)} \times G_2\\
\downarrow \\
\textcolor{blue}{SU(3)} \times G_2
\end{array}$\\
\hline

\end{tabular}
 \caption{Quivers resulting from adding all possible non-simply laced elementary slices to (\ref{A2 enhancement}) to absorb $a$.}
    \label{tab:nsl A2 enhancement}
\end{table}

\clearpage

\section*{Acknowledgements}
We would like to thank Antoine Bourget, Julius Grimminger, Zhenghao Zhong and Rudolph Kalveks for great help and support throughout this project with both computations and conceptual understanding. Also Jacques Distler for helpful email exchanges providing a double check for some of the results presented, Gabi Zafrir for correspondence regarding the previous appearances of symmetry enhancements and Siyul Lee for his result of the Hilbert series for the $d_4/S_4$ discrete quotient. The work of KG is supported by STFC DTP research studentship grant ST/V506734/1, and AH and KG are both supported by STFC grant ST/P000762/1 and ST/T000791/1.

\appendix

\section{Quiver Subtraction}
\label{app:quiver subtraction}
Here we give the techniques involved in quiver subtraction that are needed to perform the computations in this paper: how to identify which elementary slices can be subtracted from a given unframed unitary $3d$ $\mathcal{N}=4$ quiver, and how to perform said subtractions. This algorithm can be used to find the Coulomb branch Hasse diagram of an unframed quiver, and in particular to learn about its global symmetry. The techniques in the method we are about to outline are all detailed in other papers \cite{Cabrera:2018ann,Bourget:2019aer,Bourget:2020mez}, but we compile them here for the readers convenience. It will be necessary to consider the magnetic quivers for all known elementary slices, and the complete up to date collection can be found by compiling Table $1$ of \cite{Bourget:2021siw}, and Table \ref{tab:twisted affine Dds} of this note.\\

\subsection*{Quiver subtraction algorithm} 
Consider an unframed unitary $3d$ $\mathcal{N}=4$ quiver $Q$. A valid elementary slice that can be subtracted from $Q$ is one whose magnetic quiver $\sigma$ has nodes that ``lie within" $Q$. By this, we mean that a connected subset of the nodes in $Q$, $Q_\sigma$, are in the shape of $\sigma$,\footnote{More precisely, $Q_\sigma$ being in the shape of $\sigma$ means that, up to some permutations of rows, the Cartan matrix describing the links between nodes in $Q_\sigma$ is equal to the Cartan matrix describing the links between nodes in $\sigma$.} and that the rank of each node in $Q_\sigma$ are greater than or equal to that of its corresponding nodes in $\sigma$. Note that there may be many possible choices for $Q_\sigma$ for a given $\sigma$. The following steps then show us how to perform the subtraction $Q-\sigma$ for a particular $Q_\sigma$.
\begin{inparaenum}
    \item \textbf{Subtract}. For each node in $Q_\sigma$, subtract the rank of the corresponding node in $\sigma$.
    \item \textbf{Rebalance}. Identify the nodes which have undergone a change in excess, as defined in (\ref{excess defn}), due to the subtraction. Call the set of such nodes $E$, and define $E_l$ as the subset of $E$ that are long nodes and $E_s$ as the subset of $E$ that are short nodes. The nodes in $E$ must be ``rebalanced" so that they have the same excess as before. This is done differently depending on the scenario:
    \begin{inparaenum}
        \item If the subtraction performed was \textit{not} identical to the previous subtraction (i.e the same slice $\sigma$ being taken from exactly the same subset $Q_\sigma$ of $Q$), then add a $U(1)$ node $u$ to the quiver, and connect it to the nodes in $E$ in the following way:
        \begin{inparaenum}
            \item Connect all nodes in $E_l$ to $u$ with sufficiently many simply laced edges such that the excess of the nodes in $E_l$ is restored to the values they took before the subtraction.
            \item  Connect all nodes in $E_s$ to $u$ with sufficiently many \textit{non-}simply laced edges whose multiplicity is equal to the ``shortness" of the node in $E_s$ in question, such that $u$ is the long node and the excess of the nodes in $E_s$ is restored to the values they took before the subtraction.
        \end{inparaenum}
        \item If the subtraction performed was the $n^{th}$ ($n \geq 2$) in a string of identical subtractions (i.e. subtracting the same slice $\sigma$ from precisely the same subset $Q_\sigma$ of $Q$ multiple times in a row), then, calling the $U(1)$ node added to rebalance after the \textit{first} such subtraction $u$, apply the following:
        \begin{inparaenum}
            \item If $n=2$, add an adjoint hypermultiplet to $u$, and increase the rank of this node to two.
            \item If $n \geq 3$, $u$ will already have an adjoint hypermultiplet by the previous step, so simply raise the rank of $u$ by one (this will mean its rank is $n$).
        \end{inparaenum}
    \end{inparaenum}
\end{inparaenum}

\paragraph{Example} Consider the quiver
\begin{equation}\label{Qs eg}
\begin{tikzpicture}[x=1cm,y=.8cm]
\node (g1) at (-2,0) [gauge,label=below:{$1$}] {};
\node (g1b) at (-2,-1) {$\textcolor{red}{0}$};
\node (g2) at (-1,0) [gauge,label=below:{$2$}] {};
\node (g2b) at (-1,-1) {$\textcolor{red}{0}$};
\node (g3) at (0,0) [gauge,label=below:{$3$}] {};
\node (g3b) at (0,-1) {$\textcolor{red}{0}$};
\node (g4) at (1,0) [gauge,label=below:{$2$}] {};
\node (g4b) at (1,-1) {$\textcolor{red}{0}$};
\node (g5) at (2,0) [gauge,label=below:{$1$}] {};
\node (g5b) at (2,-1) {$\textcolor{red}{0}$};
\node (g12) at (-0.5,1) [gauge,label=left:{$1$}] {};
\node (g12b) at (-0.5,1.5) {$\textcolor{red}{1}$};
\node (g13) at (0.5,1) [gauge,label=right:{$1$}] {};
\node (g13b) at (0.5,1.5) {$\textcolor{red}{1}$};
\node (g14) at (2.5,0) {$,$};
\draw (g1)--(g2)--(g3)--(g4)--(g5);
\draw (g3)--(g12);
\draw (g3)--(g13);
\end{tikzpicture}
\end{equation}
where the ranks of the gauge nodes are given in black, and the excess of each node is given in red. We can see that the central five nodes form the shape of the magnetic quiver for the $d_4$ slice and have appropriately large ranks, so this slice can be subtracted. No other slice is a subset of the nodes of this quiver, so this is the only possible subtraction. Performing Step $1$ gives
\begin{equation}\label{QS eg}
\begin{tikzpicture}[x=1cm,y=.8cm]
\node (g1) at (-2,0) [gauge,label=below:{$1$}] {};
\node (g1b) at (-2,-1) {$\textcolor{red}{-1}$};
\node (g2) at (-1,0) [gauge,label=below:{$1$}] {};
\node (gb2) at (-1,-1) {$\textcolor{red}{0}$};
\node (g3) at (0,0) [gauge,label=below:{$1$}] {};
\node (g3b) at (0,-1) {$\textcolor{red}{0}$};
\node (g4) at (1,0) [gauge,label=below:{$1$}] {};
\node (g4b) at (1,-1) {$\textcolor{red}{0}$};
\node (g5) at (2,0) [gauge,label=below:{$1$}] {};
\node (g5b) at (2,-1) {$\textcolor{red}{-1}$};
\node (g14) at (2.5,0) {$.$};
\draw (g1)--(g2)--(g3)--(g4)--(g5);
\end{tikzpicture}
\end{equation}
Here, the nodes which have changed excess are the two end nodes. They are both long nodes, and so $E=E_l$ ($E_s$ is empty). This means that in Step $2$, we must follow option $a)i)$.\footnote{An example of a quiver where one must instead follow Step $2)b)$ during the quiver subtraction process is (\ref{double e6 subtraction eg}).} The nodes in $E$ both need just one extra flavour to restore their balance, and so we connect the new $U(1)$ to either end node with just a single simply laced edge:
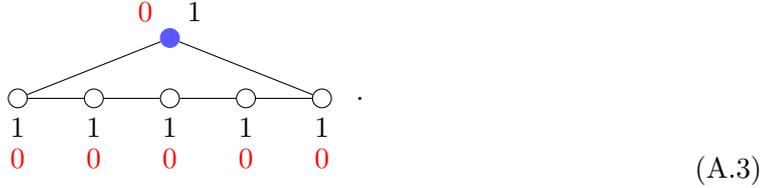
\begin{equation}\label{a5 quiv}
\begin{tikzpicture}[x=1cm,y=.8cm]
\node (g1) at (-2,0) [gauge,label=below:{$1$}] {};
\node (g1b) at (-2,-1) {$\textcolor{red}{0}$};
\node (g2) at (-1,0) [gauge,label=below:{$1$}] {};
\node (gb2) at (-1,-1) {$\textcolor{red}{0}$};
\node (g3) at (0,0) [gauge,label=below:{$1$}] {};
\node (g3b) at (0,-1) {$\textcolor{red}{0}$};
\node (g4) at (1,0) [gauge,label=below:{$1$}] {};
\node (g4b) at (1,-1) {$\textcolor{red}{0}$};
\node (g5) at (2,0) [gauge,label=below:{$1$}] {};
\node (g5b) at (2,-1) {$\textcolor{red}{0}$};
\node (g6) at (0,1) [gaugeb,label=45:{$1$},label=135:{$\textcolor{red}{0}$}] {};
\node (g14) at (2.5,0) {$.$};
\draw (g1)--(g2)--(g3)--(g4)--(g5)--(g6)--(g1);
\end{tikzpicture}
\end{equation}
We have coloured the rebalancing node $u$ in blue. (\ref{a5 quiv}) is precisely the elementary slice $\sigma=a_5$, and so clearly this is all that can be subtracted. Doing so leaves nothing left, and so this concludes the exploration of the foliation of the Coulomb branch of (\ref{Qs eg}), telling us that its Hasse diagram is
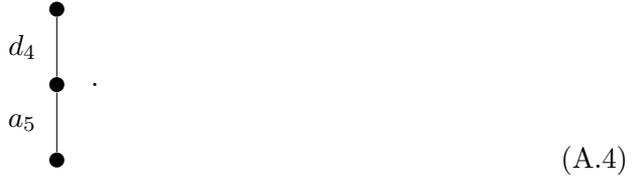
\begin{equation}\label{HD eg}
    \begin{tikzpicture}
    \node (1) [hasse] at (0,0) {};
    \node (2) [hasse] at (0,-1) {};
    \node (3) [hasse] at (0,-2) {};
    \node (4) at (0.5,-1) {$.$};
    \draw (1) edge [] node[label=left:$d_4$] {} (2);
    \draw (2) edge [] node[label=left:$a_5$] {} (3);
    \end{tikzpicture}
\end{equation}
The interpretation here is if you pick a point on the Coulomb branch moduli space, it will lie on one of the three leaves in the Hasse diagram (\ref{HD eg}), each of which correspond to a certain set of massless states. The Coulomb branch of the quiver theory of (\ref{Qs eg}) is obviously the closure of the top leaf (the whole Hasse diagram), as this is the moduli space we're studying. A generic point on this Coulomb branch will have the maximal number of massless states: $10$ massless vectormultiplets,\footnote{The number of massless vectormultiplets can be read from the quiver. The quaternionic dimension of the Coulomb branch will be the number of monople operators we have, which is equal to the rank of the gauge group. The complex or real dimension then is twice or four times this respectively.} and no massless hypermultiplets (by the BPS formula). All such points live on the top leaf, associated to the quiver (\ref{Qs eg}). There are then certain points on the Coulomb branch which have fewer massless states: just $5$ massless vectormultiplets now, plus a massless hypermultiplet.\footnote{This massless hypermultiplet opens up Higgs branch directions in the Hasse diagram for the full moduli space as it may now acquire a VEV \cite{Grimminger:2020dmg}: there is a set of moduli associated to this hypermultiplet which can be tuned away from zero to explore the Higgs branch. It can be seen from the brane picture that these moduli correspond to the $D_4$ variety. Such a set of moduli is called a transverse slice, as discussed later in this paragraph.} Such points lie on the middle leaf of (\ref{HD eg}), and are associated to the quiver (\ref{a5 quiv}). The Coulomb branch of (\ref{a5 quiv}) is then the closure of this middle leaf (i.e all points on this leaf and the bottom leaf). At one particular point on the Coulomb branch all $20$ vectormultiplets and $22$ hypermultiplets are massless. This is the origin of the Coulomb branch, and is the sole point which lives on the bottom leaf of (\ref{HD eg}). The whole Hasse diagram obviously contains its bottom half, and so the Coulomb branch moduli space of (\ref{Qs eg}) contains that of (\ref{a5 quiv}). The transverse slice between these two moduli spaces is $d_4$, and between the origin and the (\ref{a5 quiv}) moduli space is $a_5$. The transverse slices connecting two leaves tell us the moduli that need tuning away from or to zero to move between the two corresponding leaves. The lowest elementary slice is $a_5$, and thus $SU(6)$ must be at least a subgroup of the global symmetry. Indeed this is confirmed upon computation of the Hilbert series, which tells us that the global symmetry is $SU(6) \times U(1)$. \hfill $\square$

\section{Fugacity Maps}
\label{app:fugmap}
In this appendix we discuss the notion of fugacity maps, and end by giving the fugacity map for the quiver in (\ref{double e6 RHS example quiver}) and its derivation.\\ 

\noindent The Hilbert series counts the chiral operators that parameterise the moduli space of a theory, ``graded" by their representations under global symmetries. The Coulomb branch is parameterised by monopole operators charged under the topological symmetry, and in the UV this is comprised of a $U(1)$ global symmetry for each $U(1)$ factor in the gauge group: $U(1)^r$, where $r$ is the rank of the gauge group. Upon flowing to the IR, this symmetry often grows due to the appearance of extra monopole operators. To account for this topological global symmetry in the monopole formula, a fugacity $z_i$ is introduced for each gauge node in the quiver, and is raised to the power of the charges of the monopole operators under this symmetry.\footnote{Note that since the symmetry comes from Abelian factors in the gauge group, this charge will just be the sum of the magnetic weights involved.} When we sum over the magnetic lattice in the monopole formula we can see the $U(1)^r$ topological symmetry from the UV becoming larger in the IR: the coefficients of the powers of $t$, which will be functions of the fugacities $z_i$, form characters of the representations of some group. This group is exactly the topological global symmetry of the Coulomb branch moduli space.\\

\noindent We are used to working with characters which are given as arbitrary fugacities graded by the weights of the representation in question. We will call characters expressed in this way fundamental weight characters, and their fugacities fundamental weight fugacities $x_i$. However, the coefficients of $t$ appearing in the Hilbert series are not always immediately characters of this form. Often, the $z_i$ must undergo some sort of mapping before becoming the fundamental weight fugacities $x_i$ so that the coefficients they form can be readily recognised as some sum of fundamental weight characters of the relevant global symmetry group. Such a mapping is called a \textbf{fugacity map.} Note that in an unframed quiver, the map can change depending on where you ungauge. Recall that the $t^2$ coefficient of the Hilbert series of a moduli space forms the character for the adjoint representation of its global symmetry. This means that we can isolate just the $t^2$ term to find the fugacity map.\\

\noindent In the simplest cases, the $t^2$ coefficient in the Hilbert series comes out as the character of the global symmetry in terms of the simple roots (we will call such characters simple root characters\footnote{Recall that characters encode a representation. Said representation contains certain weights, written in terms of a linear combination of fundamental weights, and the coefficients in this linear combination are how we grade the fundamental weight fugacities in the fundamental weight character. For the simple root character, the only difference is that the weights are written in terms of a linear combination of the simple roots instead. This linear combination will clearly have different coefficients to the equivalent linear combination of fundamental weights, hence the different character.}) without any manipulation of fugacities. In these cases the fugacity map required is simply given by the Cartan matrix. This is the case for example for any affine Dynkin diagram (see Table \ref{tab:affine Dds}) when the ungauging is performed on the affine node.\\

\paragraph{Example} Consider the affine $A_2$ Dynkin diagram, which we know has global symmetry $SU(3)$. The ungauged quiver is
\begin{equation}\label{ungauged affine A2}\begin{tikzpicture}[x=1cm,y=.8cm]
\node (g1) at (-0.5,0) [gauge,label=below:{$1$}] {};
\node (g4) at (-0.5,1) [flavor,label=left:{$1$}] {};
\node (g2) at (0.5,0) [gauge,label=below:{$1$}] {};
\node (g3) at (0.5,1) [flavor,label=right:{$1$}] {};
\draw (g4)--(g1)--(g2)--(g3);
\node (f) at (1,0) {$.$};
\end{tikzpicture}
\end{equation}
If we call the simple roots of $SU(3)$ $\alpha_1$ and $\alpha_2$, then the full root system is 
\begin{equation}\label{rootsyst su3}
    \{ \alpha_1, \ \alpha_2, \ \alpha_1 + \alpha_2, \ -\alpha_1, \ -\alpha_2, \ -\alpha_1-\alpha_2 \}.
\end{equation}
Assigning fugacities $z_1$ and $z_2$ to the remaining gauge nodes, the Hilbert series of (\ref{ungauged affine A2}) can be computed to $t^2$ as
\begin{equation}\label{a2 t2 coeff}
    1+(2 + \frac{1}{z_1} + z_1 + \frac{1}{z_2} + \frac{1}{z_1 z_2} + z_2 + z_1 z_2) \, t^2 + \mathcal{O}(t^4).
\end{equation}
The $t^2$ term has unrefined dimension $8$, and so we expect an $SU(3)$ global symmetry. This can be confirmed by inspecting the refined $t^2$ coefficient: it is indeed the root decomposition of the algebra of $SU(3)$. The Cartan subalgebra is encoded in the constant term equal to $rank(SU(3))=2$, and all positive and negative roots are encoded by the products of $z_1$, $z_2$ and their reciprocals: $z_1$ and $z_2$ are raised to the powers of the coefficients of the simple roots that are equal to these positive and negative roots. That is, if one uses the identification
\begin{equation}
    c_{1}\, \alpha_1 +c_{2}\, \alpha_2 \ \longleftrightarrow \ z_1^{c_{1}}\, z_2^{c_{2}},
\end{equation}
we see that the root system of $SU(3)$ (\ref{rootsyst su3}) and the two Cartan elements completely comprises the $t^2$ coefficient of the Hilbert series (\ref{a2 t2 coeff}). This tells us that (\ref{a2 t2 coeff}) is written in terms of simple root characters of its global symmetry $SU(3)$. To convert to the more familiar fundamental weight characters then, we need to apply the Cartan matrix as our fugacity map, yielding the Hilbert series
$$1+ (2 + \frac{x_1}{x_2^2} + \frac{1}{x_1 x_2} + \frac{x_1^2}{x_2} + \frac{x_2}{x_1^2} + x_1 x_2 + \frac{x_2^2}{x_1}) \, t^2 + \mathcal{O}(t^4),$$
the $t^2$ coefficient of which we indeed recognise as the usual fundamental weight character of the adjoint representation of $SU(3)$. This confirms the $SU(3)$ global symmetry. \hfill $\square$\\

\noindent In most cases the fugacity maps are a bit trickier to find. However there are some well known tricks that work in a lot of instances, and we will try to illustrate these in the example of (\ref{double e6 RHS example quiver}).\\

\paragraph{Example} Consider the quiver $Q$ of (\ref{double e6 RHS example quiver}). Label the fugacities assigned to the nodes as follows\footnote{The general intuition for doing this is that there is a balanced $D_6$ Dynkin diagram, and so we label the nodes in this with index corresponding to the weight that node represents in highest weight notation.}:
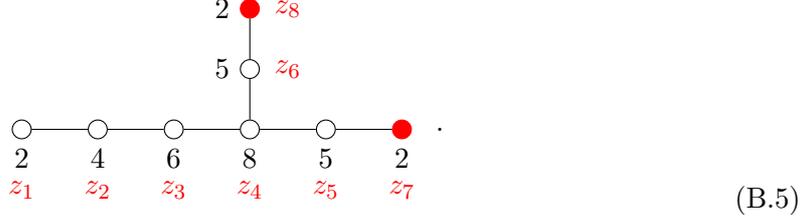
\begin{equation}
    \begin{tikzpicture}[x=1cm,y=.8cm]
    \node (g2) at (-2,0) [gauge,label=below:{$2$}] {};
    \node (g2b) at (-2,-1) {$\textcolor{red}{z_1}$};
    \node (g3) at (-1,0) [gauge,label=below:{$4$}] {};
    \node (g3b) at (-1,-1) {$\textcolor{red}{z_2}$};
    \node (g4) at (0,0) [gauge,label=below:{$6$}] {};
    \node (g4b) at (0,-1) {$\textcolor{red}{z_3}$};
    \node (g5) at (1,0) [gauge,label=below:{$8$}] {};
    \node (g5b) at (1,-1) {$\textcolor{red}{z_4}$};
    \node (g6) at (2,0) [gauge, label=below:{$5$}] {};
    \node (g6b) at (2,-1) {$\textcolor{red}{z_5}$};
    \node (g7) at (1,1) [gauge, label=left:{$5$}] {};
    \node (g7b) at (1.5,1) {$\textcolor{red}{z_6}$};
    \node (g8) at (1,2) [gauger, label=left:{$2$}] {};
    \node (g7b) at (1.5,2) {$\textcolor{red}{z_8}$};
    \node (g9) at (3,0) [gauger, label=below:{$2$}] {};
    \node (g9b) at (3,-1) {$\textcolor{red}{z_7}$};
    \node (f) at (3.5,0) {$.$};
    \draw (g2)--(g3)--(g4)--(g5)--(g6)--(g9);
    \draw (g5)--(g7)--(g8);
    \end{tikzpicture}
\end{equation}
The $t^2$ coefficient that we get unrefined is $79$, which is the dimension of $SO(13) \times U(1)$.\footnote{The character of a product group irreducible representation is equal to the sum of the characters of the individual irreducible representations of each group in the product.} Refined, it has virtually no fractional terms, which means it can't be in the form of the root system of this global symmetry.\footnote{The weight system of a real representation is comprised of some set of weights and their inverses, and possibly some trivial elements. The adjoint representation is real, and so here the weights are plus and minus the positive root system, in addition to the Cartan elements. This means the refined $t^2$ coefficient would have an even number of fractional and non-fractional terms that would be inverses of one another.} We have two unbalanced nodes, which we don't expect to contribute to the non-Abelian global symmetry. A nice trick that often works is to map a fugacity $z_{i}$ corresponding to an unbalanced node of rank $r_{i}$ to the $r_{i}^{th}$ root of the inverse of the product of all other fugacities raised to the power of their node ranks:
\begin{equation}\label{unbalanced node fug map}
    z_i \longrightarrow \sqrt[\leftroot{-2}\uproot{2}r_i]{\frac{1}{\prod_{j \neq i} z_j^{r_j}}}.
\end{equation}
In this case there are two unbalanced nodes: $z_7$ and $z_8$. We choose to pick $z_8$ to be our $z_i$ of (\ref{unbalanced node fug map}), and thus our map here will be 
\begin{equation}
    z_8 \longrightarrow \sqrt[\leftroot{-2}\uproot{2}2]{\frac{1}{z_1^2 \, z_2^4 \, z_3^6 \, z_4^8 \, z_5^5 \, z_6^5 \, z_7^2}}.
\end{equation}
After applying this, the $z_7$ fugacity also drops out of the Hilbert series, and so we have just $z_1,...,z_6$ left. In the resulting Hilbert series there are terms of the form 
$$\sqrt{z_5 \, z_6}, \ \sqrt{\frac{z_5}{z_6}}$$
appearing, and we don't have fractional powers in root systems. So a natural map to take next is 
$$z_5 \longrightarrow \frac{z_9}{z_{10}} , \ z_6 \longrightarrow z_9 \, z_{10}.$$
In fact we find that then just sending 
$$z_9 \longrightarrow z_9 \, z_{10}$$
gives us the root decomposition of $SO(13) \times U(1)$,\footnote{Recall the adjoint representation of $U(1)$ is just the trivial representation.} i.e. its simple root character. We can then just note that the simple roots are given by the Cartan matrix acting on the fundamental weights to find the more recognisable fundamental weight characters. Overall the fugacity map between the $z_i$ in the monopole formula and the fundamental weight fugacities $x_i$ of $B_6$ is given by 
\begin{equation}\label{fug map for e6 eg}
    \begin{pmatrix}
    z_1\\
    z_2\\
    z_3\\
    z_4\\
    z_5\\
    z_6\\
    z_7\\
    z_8\\
    \end{pmatrix}
    =M \, C \, 
    \begin{pmatrix}
    x_1\\
    x_2\\
    x_3\\
    x_4\\
    x_5\\
    x_6\\
    f\\
    \end{pmatrix}
\end{equation}
where $f$ is just some auxhiliary fugacity that disappears in the Hilbert series under this map, and
\begin{equation}
   M= \begin{pmatrix}
    1 & 0 & 0 & 0 & 0 & 0 & 0\\
    0 & 1 & 0 & 0 & 0 & 0 & 0\\
    0 & 0 & 1 & 0 & 0 & 0 & 0\\
    0 & 0 & 0 & 1 & 0 & 0 & 0\\
    0 & 0 & 0 & 0 & 1 & 0 & 0\\
    0 & 0 & 0 & 0 & 1 & 2 & 0\\
    0 & 0 & 0 & 0 & 0 & 0 & 1\\
    -1 & -2 & -3 & -4 & -5 & -5 & -1\\
    \end{pmatrix}, \ \ \
    C =\begin{pmatrix}
    2 & -1 & 0 & 0 & 0 & 0 & 0\\
    -1 & 2 & -1 & 0 & 0 & 0 & 0\\
    0 & -1 & 2 & -1 & 0 & 0 & 0\\
    0 & 0 & -1 & 2 & -1 & 0 & 0\\
    0 & 0 & 0 & -1 & 2 & -2 & 0\\
    0 & 0 & 0 & 0 & -1 & 2 & 0\\
    0 & 0 & 0 & 0 & 0 & 0 & 1\\
    \end{pmatrix}.
\end{equation}
$M$ is the matrix used to multiply the simple roots to find the $z_i$, and $C$ is the Cartan matrix of $B_6$ (with an extra trivial row tagged along to respect the auxiliary fugacity). Note that the matrix multiplication isn't meant in the usual sense here: rather than the entries of the matrix being coefficients of the vector they multiply, they are instead the powers that the vector elements (that they would traditionally multiply) are raised to. The fundamental weight fugacities $x_i$ are indexed in the usual order corresponding to the $B$ type Dynkin diagram:
\begin{equation}
    \begin{tikzpicture}[x=1cm,y=.8cm]
    \node (g2) at (-2.5,0) [gauge,label=below:{$\textcolor{red}{x_1}$}] {};
    \node (g3) at (-1.5,0) [gauge,label=below:{$\textcolor{red}{x_2}$}] {};
    \node (g4) at (-0.5,0) [gauge,label=below:{$\textcolor{red}{x_3}$}] {};
    \node (g5) at (0.5,0) [gauge,label=below:{$\textcolor{red}{x_4}$}] {};
    \node (g6) at (1.5,0) [gauge, label=below:{$\textcolor{red}{x_5}$}] {};
    \node (g7) at (2.5,0) [gauge, label=below:{$\textcolor{red}{x_6}$}] {};
    \node (f) at (3,0) {$.$};
    \draw (g2)--(g3)--(g4)--(g5)--(g6);
    \draw[transform canvas={yshift=-1.5pt}] (g6)--(g7);
    \draw[transform canvas={yshift=1.5pt}] (g6)--(g7);
    \draw (1.9,-0.2)--(2.1,0)--(1.9,0.2);
    \end{tikzpicture}
\end{equation} \hfill $\square$\\
\noindent This concludes our discussion on finding fugacity maps.

\section{Discrete Projections}
\label{app:discrete projection}
As mentioned in Section \ref{sec:quivers with adj matter} (and at multiple subsequent points) the Coulomb branches of quivers containing nodes with an adjoint hypermultiplet can be realised as discrete quotients of the Coulomb branches of quivers with a bouquet of $U(1)$ nodes (\ref{bouquet to adjoint conjecture}) \cite{Hanany:2018vph,Hanany:2018cgo,Hanany:2018dvd,Bourget:2020bxh}. In this appendix, we explain how this can be verified, and see it in practise with an example. As in the rest of the paper we omit ``Coulomb branch" from phrases (for example Coulomb branch HWG will be just HWG), as it is always assumed in this appendix that we are discussing the Coulomb branch.\\

\noindent Suppose we have two quivers $A$ and $B$ which are conjectured to satisfy the relation
\begin{equation}\label{appC eg}
   \mathcal{C}(B)/S_n=\mathcal{C}(A),
\end{equation}
due to $B$ having some $S_n$ outer automorphism. The way we prove this conjecture to the best of our ability is to show the equality of Hilbert series. This is done by calculating the Hilbert series for $B$, and finding some $S_n$ action on its generators and relations that obtains the Hilbert series of $A$. To be more precise, on the left hand side of (\ref{appC eg}) we are trying to calculate the $S_n$ gauged Coulomb branch of $B$. This Coulomb branch should be $S_n$ invariant, and so this gauging is realised by finding some action of $S_n$ on the Coulomb branch and performing the corresponding Molien sum on the the Hilbert series of $B$ to find the Hilbert series of $\mathcal{C}(B)/S_n$ such that it matches the Hilbert series of $\mathcal{C}(A)$, which can just be computed in the usual manner using the monopole formula.\\

\noindent Schematically, this Molien sum (which is responsible for making an object gauge invariant) over our discrete group $G=S_n$ goes like 
\begin{equation}
    HS(\mathcal{C}(B)/G)=\frac{1}{|G|}\sum_{g \in G} g \cdot HS(\mathcal{C}(B)).
\end{equation}
where $HS$ stands for Hilbert series and $\cdot$ is an action of $G$. The action of $G$ on a Hilbert series is fully determined by its action on the generators and relations of said Hilbert series. Each element of $G$ will in general act differently, and the contributions from each of these actions are summed together before their total is divided by the cardinality of $G$. To find the action of $G$, we need to analyse and understand its representations and characters. This is just an exercise in the theory of finite groups. The example we show here will be that of $G=S_2=\mathbb{Z}_2$, for which the game is a bit easier as there are just two one dimensional representations. However the method can be extended and applied to any finite group, provided the representations and characters are known and understood.\\

\paragraph{Example} Consider the magnetic quiver of the next to minimal nilpotent orbit of $B_3$:
\begin{equation}\label{nminB3}
     \begin{tikzpicture}[x=1cm,y=.8cm]
    \node (g2) at (-0.4,0) [gauge,label=below:{$2$}] {};
    \node (g1) at (-0.4,1.1) [gauge,label=left:{$2$}] {};
    \node (g3) at (0.4,0.8) [gauge,label=right:{$1$}] {};
    \node (g4) at (0.4,-0.8) [gauge,label=right:{$1$}] {};
    \node (f) at (0.9,0) {$.$};
    \draw (g1)--(g2)--(g3);
    \draw (g2)--(g4);
    \draw (g1) to [out=45,in=135,looseness=10] (g1);
    \end{tikzpicture}
\end{equation}
In this case, the conjecture (\ref{bouquet to adjoint conjecture}) tells us that
\begin{equation}\label{disc quot app eg}
    \mathcal{C}\left( \raisebox{-.5\height}{\begin{tikzpicture}[x=1cm,y=.8cm]
    \node (g1) at (-1.5,0) [gauge,label=below:{$2$}] {};
    \node (g2) at (-2.3,0.8) [gauge,label=left:{$1$}] {};
    \node (g3) at (-2.3,-0.8) [gauge,label=left:{$1$}] {};
    \node (g4) at (-0.7,0.8) [gauge,label=right:{$1$}] {};
    \node (g5) at (-0.7,-0.8) [gauge,label=right:{$1$}] {};
    \draw (g2)--(g1)--(g3);
    \draw (g1)--(g4);
    \draw (g1)--(g5);
    \end{tikzpicture}}
    \right)\Big/\mathbb{Z}_2 \, =\, \mathcal{C}\left(
    \raisebox{-.5\height}{\begin{tikzpicture}[x=1cm,y=.8cm]
    \node (g7) at (2,0) [gauge,label=below:{$2$}] {};
    \node (g6) at (2,1.1) [gauge,label=left:{$2$}] {};
    \node (g8) at (2.8,0.8) [gauge,label=right:{$1$}] {};
    \node (g9) at (2.8,-0.8) [gauge,label=right:{$1$}] {};
    \draw (g6)--(g7)--(g8);
    \draw (g7)--(g9);
    \draw (g6) to [out=45,in=135,looseness=10] (g6);
    \end{tikzpicture}}
    \right).
\end{equation}
Let's see how to show this. First, let's compute the Hilbert series of the quiver on the right hand side, which we call $Q_{nminB_3}$, so we know what we are looking to obtain from the left hand side via the discrete quotient. We can find the HWG of $Q_{nminB_3}$, which completely encodes the refined Hilbert series, to be
\begin{equation} \label{nminb3 hwg}
    HWG_{B_3}(Q_{nminB_3})= PE[\mu_2 t^2+\mu_1^2 t^4]=\frac{1}{(1-\mu_2 t^2)(1-\mu_1^2 t^4)},
\end{equation}
where the $B_3$ subscript on $HWG$ is used to illustrate that $\{\mu_1,\mu_2,\mu_3 \}$ are the Dynkin label fugacities (i.e. highest weight fugacities) for $B_3$.\\

\noindent Now it's time to realise this as the $\mathbb{Z}_2$ quotient of the minimal nilpotent orbit of $d_4$, which is the Coulomb branch of the quiver on the left hand side of (\ref{disc quot app eg}), that we'll call $Q_{minD_4}$. The HWG of this quiver can be easily computed to be
\begin{equation} \label{mind4 hwg}
    HWG_{D_4}(Q_{minD_4})=PE[\tilde{\mu}_2 t^2]=\frac{1}{1-\tilde{\mu}_2 t^2},
\end{equation}
where as before the $D_4$ subscript on $HWG$ tells us that $\{\tilde{\mu}_1,\tilde{\mu}_2,\tilde{\mu}_3,\tilde{\mu}_4 \}$ are the Dynkin label fugacities for $D_4$. The goal now is to find a $\mathbb{Z}_2$ action on this that will reproduce (\ref{nminb3 hwg}). However in order to do this we need to have (\ref{mind4 hwg}) in terms of the same fugacities as (\ref{nminb3 hwg}), so we must decompose the adjoint representation of $D_4$ into irreducible representations of $B_3$. Writing the characters of $D_4$ in terms of fundamental weight fugacities $\{x_1,x_2,x_3,x_4\}$ and those of $B_3$ in terms of $\{x_1,x_2,x_3\}$, then under the fugacity map $x_4 \rightarrow x_3$ one can find that the adjoint representation of $D_4$ decomposes into the sum of the adjoint and fundamental representation of $B_3$:
\begin{equation}
    \tilde{\mu}_2 \rightarrow \mu_2 + \mu_1.
\end{equation}
Based on this, we guess that the HWG of the minimal nilpotent orbit of $d_4$, i.e. $\mathcal{C}(Q_{minD_4})$, in terms of $B_3$ Dynkin label fugacities is
\begin{equation} \label{mind4 hwg in terms of b3}
    HWG_{B_3}(Q_{minD_4})=PE[(\mu_2 + \mu_1)t^2]=\frac{1}{(1-\mu_2 t^2)(1-\mu_1 t^2)}.
\end{equation}
We need to check that this is correct, as it could be that these representations of $B_3$ will actually overcount the representations of $D_4$ we wanted, and to correct this we'd need to impose relations. The way to check for this is to turn (\ref{mind4 hwg in terms of b3}) into a refined or unrefined Hilbert series, and compare it to that obtained from using the monopole formula on $Q_{minD_4}$ (after the appropriate fugacity map in the refined case). Here, under performing this check we see that (\ref{mind4 hwg}) and (\ref{mind4 hwg in terms of b3}) yield the same Hilbert series, and so (\ref{mind4 hwg in terms of b3}) is indeed the correct HWG for the minimal nilpotent orbit of $D_4$ in terms of $B_3$ fugacities.\\

\noindent Now all that's left to do is find the $\mathbb{Z}_2$ action that when used in a Molien sum on (\ref{mind4 hwg in terms of b3}) will yield (\ref{nminb3 hwg}). The action on the generators will fully determine the action everywhere, and we can see here our generators are $\mu_1 t^2$ and $\mu_2 t^2$. The group $\mathbb{Z}_2$ has two elements: the identity and some other element which squares to the identity, e.g. $\{1,-1\}$. So our Molien sum looks like
\begin{equation}
    HWG_{B_3}(Q_{minD_4}])/\mathbb{Z}_2 = \frac{1}{2} \sum_{g \in \{1,-1\}} PE[g \cdot \mu_1 t^2 + g \cdot \mu_2 t^2].
\end{equation}
The representations of $\mathbb{Z}_2$ that $\mu_1$ and $\mu_2$ are in determine the action of $\mathbb{Z}_2$ on them. There are just two representations of $\mathbb{Z}_2$: the trivial representation and the sign representation. The trivial representation is obviously invariant under all group elements, and the sign representation is mapped to plus or minus itself by the elements $1$ or $-1$ of $\mathbb{Z}_2$ respectively. It turns out that if the $\mu_1$ is in the sign representation and $\mu_2$ is in the trivial representation, we reproduce (\ref{nminb3 hwg}):
\begin{equation}
\begin{split}
    HWG(Q_{minD_4})/\mathbb{Z}_2 &= \frac{1}{2} \left( PE[1 \cdot \mu_1 t^2 + 1 \cdot \mu_2 t^2] + PE[-1 \cdot \mu_1 t^2 + -1 \cdot \mu_2 t^2] \right)\\
    &= \frac{1}{2} \left( PE[\mu_1 t^2 + \mu_2 t^2] + PE[-\mu_1 t^2 + \mu_2 t^2] \right)\\
    &= PE[\mu_2 t^2 + \mu_1^2 t^4]\\
    &= HWG(Q_{nminB_3}).
\end{split}
\end{equation}
This completes the proof of the equality of Hilbert series for $\mathcal{C}(Q_{minD_4})/S_2$ and $\mathcal{C}(Q_{nminB_3})$, and hence validating the conjecture (\ref{disc quot app eg}) to the best of our ability.\footnote{We say only to the best of our ability as the Hilbert series is not a complete characterisation of the Coulomb branch moduli space, but at present it is the most complete encapsulation that we have.} $\square$ \\

\bibliographystyle{JHEP}
\bibliography{bibli.bib}

\end{document}